\documentclass[aps,pra,superscriptaddress,floatfix,notitlepage,nofootinbib,longbibliography,11pt]{revtex4-1} 
\pdfoutput=1
\usepackage{amsmath,amsthm,amsfonts,amssymb,braket}
\usepackage{graphicx}
\usepackage{color}

\makeatletter
\def\l@subsubsection#1#2{}
\makeatother

\usepackage[colorlinks=true,citecolor=blue,linkcolor=blue]{hyperref}
\usepackage[capitalize,nameinlink]{cleveref}

\newtheorem{lemma}{Lemma}[section]
\newtheorem{theorem}[lemma]{Theorem}
\newtheorem{corollary}[lemma]{Corollary}

\theoremstyle{definition}
\newtheorem{definition}[lemma]{Definition}
\newtheorem{remark}[lemma]{Remark}

\newcommand{\CC}{{\mathbb{C}}}
\newcommand{\RR}{{\mathbb{R}}}
\newcommand{\ZZ}{{\mathbb{Z}}}
\newcommand{\Z}{{\mathbb{Z}}}
\newcommand{\FF}{{\mathbb{F}}}

\newcommand{\diag}{{\mathrm{diag}}}

\DeclareMathOperator*{\Tr}{{Tr}}
\DeclareMathOperator*{\Coe}{{Coe}}

\DeclareMathOperator{\im}{\mathrm{im}}
\DeclareMathOperator{\rank}{\mathrm{rank}}
\DeclareMathOperator{\coker}{\mathrm{coker}}
\DeclareMathOperator{\ann}{\mathrm{ann}}

\newcommand{\commA}{{\cal A}}
\newcommand{\opZ}{{\cal Z}}
\newcommand{\opX}{{\cal X}}

\newcommand{\stab}{\mathcal{S}}
\newcommand{\pauligroup}{\mathcal{P}}
\newcommand{\dist}{\mathrm{dist}}

\newcommand{\Circuit}{\textsf{Circuit}}
\newcommand{\Circshift}{\textsf{CircShift}}
\newcommand{\All}{\textsf{All}}
\newcommand{\Id}{\mathrm{Id}}
\newcommand{\Swap}{\mathrm{SWAP}}
\newcommand{\fone}{f_1}
\newcommand{\ftwo}{f_2}
\newcommand{\fthree}{f_3}

\newcommand{\grotimes}{{\otimes_{gr}}}

\begin{document}

\title{Nontrivial Quantum Cellular Automata in Higher Dimensions}

\author{Jeongwan Haah}
\affiliation{Quantum Architectures and Computation, Microsoft Research, Redmond, WA 98052, USA}

\author{Lukasz Fidkowski}
\affiliation{Department of Physics, University of Washington, Seattle WA 98195, USA}

\author{Matthew B.~Hastings}
\affiliation{Station Q, Microsoft Research, Santa Barbara, CA 93106-6105, USA}
\affiliation{Quantum Architectures and Computation, Microsoft Research, Redmond, WA 98052, USA}

\begin{abstract}
We construct a three-dimensional quantum cellular automaton (QCA),
an automorphism of the local operator algebra on a lattice of qubits, 
which disentangles the ground state of the Walker-Wang three fermion model.
We show that if this QCA can be realized by a quantum circuit of constant depth,
then there exists a two-dimensional commuting projector Hamiltonian 
which realizes the three fermion topological order
which is widely believed not to be possible.
We conjecture in accordance with this belief
that this QCA is not a quantum circuit of constant depth, and
we provide two further pieces of evidence to support the conjecture.
We show that this QCA maps every local Pauli operator to a local Pauli operator,
but is not a Clifford circuit of constant depth.
Further, we show that if the three-dimensional QCA can be realized by a quantum circuit of constant depth, 
then there exists a two-dimensional QCA acting on fermionic degrees of freedom 
which cannot be realized by a quantum circuit of constant depth; 
i.e., we prove the existence of a nontrivial QCA in either three or two dimensions.
The square of our three-dimensional QCA can be realized by a quantum circuit of constant depth,
and this suggests the existence of a $\ZZ_2$ invariant of a QCA in higher dimensions, 
%(more likely, a classification by several $\ZZ_2$ indices), 
totally distinct from the classification by positive rationals 
(i.e., by one integer index for each prime) in one dimension.

In an appendix, unrelated to the main body of this paper, 
we give a fermionic generalization of a result of Bravyi and Vyalyi~\cite{bravyi2005commutative}
on ground states of 2-local commuting Hamiltonians.
\end{abstract}

\maketitle

\clearpage
\tableofcontents
\clearpage

\section{Introduction}\label{sec:introduction}

A quantum cellular automaton (QCA) 
is an automorphism of the $*$-algebra of operators acting on the Hilbert space of a quantum system 
which obeys certain locality constraints.
For finite size system, any such QCA can be described by conjugation by a unitary:
the QCA $\alpha$ is defined by $\alpha(O)=U^\dagger O U$ for some unitary $U$,
where $O$ is an arbitrary operator.

There are at least four distinct cases of QCA that are considered in the literature.
First, the QCA may be ``strictly local'' or may ``have tails.''
In the strictly local case, 
we require that for any operator $O$ that is supported on a single site,
the image $\alpha(O)$ be supported on the set of sites 
within some bounded distance of the given site,
while in the case with tails, one merely requires that 
$\alpha(O)$ can be approximated by an operator supported within some distance $r$ of the site,
up to an error decaying in $r$ (perhaps exponentially).
The present paper considers only strictly local QCA,
so we will avoid giving a precise definition of the case with tails.
Throughout this paper, a QCA will always means such a strictly local QCA,
even if not explicitly mentioned, 
though when we review previous work we will in some cases
discuss the difference between strictly local and tails.
Second, the degrees of freedom may be qudits, 
or one may have both qudits and fermionic degrees of freedom.
The present paper will consider both types of QCA 
(indeed, part of our understanding of the nontriviality of a certain three-dimensional qudit QCA
will depend upon the existence or non-existence 
of a two-dimensional fermionic QCA).
We require that a fermionic QCA preserve the fermion parity.

Here ``nontriviality'' can be understood in a variety of ways,
such as whether a QCA can be implemented by a quantum circuit of bounded depth and range.
We define this more precisely below.

While QCA in one dimension are fully classified
(at least for the case of strictly local QCA),
following the work of Gross, Nesme, Vogts and Werner (abbreviated GNVW later),
in the qudit case~\cite{Gross_2012}
and in the case with fermionic degrees of freedom~\cite{fermionGNVW1,fermionGNVW2},
the classification of QCA in varying dimensions is much less well-developed.
In fact, in higher dimensions, we do not have {\it any} examples of nontrivial QCA,%
\footnote{
There is one possible candidate however in $5$ dimensions 
using fermionic degrees of freedom and allowing tails in the QCA.
Consider a system with Majorana degrees of freedom and restrict to ``Gaussian QCA,''
which map Majorana operators to linear combinations of Majorana operators.
This linear map of Majorana operators is described by an orthogonal matrix,
and the $K$-theory classification~\cite{Kitaev_2009,Ryu_2010} of so-called ``chiral orthogonal matrices,'' symmetry class BDI,
has an integer invariant in dimensions $4k+1$ as well as ${\mathbb Z}_2$ invariants in dimensions $8k+6,8k+7$.
However, it is not clear if this invariant remains 
if one allows more general QCA which need not be Gaussian.
}
except for a classification using homology~\cite{FreedmanHastings} 
which uses the index theory of GNVW in higher dimensions 
by considering various ways of dimensionally reducing a higher dimensional manifold to one dimension;
the indices found there can be cancelled by tensoring the higher-dimensional QCA 
with various QCA which act as shift operators on a line of sites which is a nontrivial cycle of the manifold.
Indeed, in two dimensions, it is known~\cite{FreedmanHastings} 
that all strictly local QCA are trivial in a sense defined below.

The purpose of the present paper is 
to show the existence of nontrivial QCA in higher dimensions.
We will show that (as made precise later) either
a nontrivial three-dimensional qudit QCA exists or a nontrivial two-dimensional fermionic QCA exists.

\subsection{Classification of Hamiltonians and relation to classification of QCA; Outline of proof}

Let us briefly discuss the classification of Hamiltonians, rather than QCA.
We do this to emphasize that much more is known about the classification
of Hamiltonians than of QCA {\it and} 
we do this because our nontrivial QCA are constructed by considering quantum Hamiltonians.
While we only work out one particular case in detail,
our technique suggests a way of constructing further examples of QCA by using known results about Hamiltonians.

If we restrict to Hamiltonians which are sums of commuting projectors, 
then in one-dimension the structure of ground states 
is fully understood for qudit degrees of freedom following Bravyi and Vyalyi~\cite{bravyi2005commutative}.
In \cref{sec:fBV} of the present paper,
we extend this result to fermionic degrees of freedom.
Thus in one-dimension, the classification of Hamiltonians and QCA are both fully understood.
In two dimensions, commuting projectors Hamiltonians are known 
which realize nontrivial topological order; 
examples include the Levin-Wen models~\cite{Levin_2005}.
However, there are few if any classification results in two dimensions 
to show that we have in some sense exhausted the possible topological phases. 
In three or more dimensions, fewer examples of topologically ordered phases are known 
(examples of commuting projector Hamitonians include higher dimensional toric codes~\cite{Kitaev_2003,Freedman_2002,Dennis_2002},
Dijkgraaf-Witten models~\cite{dijkgraaf1990topological},
generalized double semion models~\cite{Freedman_2016})
but there are also some models such as the cubic code~\cite{Haah_2011} with exotic properties 
which suggest that in three or more dimensions we are very far from exhausting all possible kinds of topological order.
Thus, while we do not have a classification of Hamiltonians in higher dimensions, 
we have a conjectural classification in two-dimensions,
and we have many nontrivial examples in higher dimensions.

One particularly interesting kind of Hamiltonian is the three-dimensional Walker-Wang model~\cite{Walker_2011,Curt}.
There is actually a whole class of Walker-Wang models;
in this paper we will focus on a specific one, 
based on the so-called three fermion anyon model described below.
In this case, the model has (loosely speaking) no bulk topological order.
However, if the model is terminated at a surface, 
with any additional surface terms chosen so that the Hamiltonian remains a sum of commuting projectors,
then it is believed that some type of topological order necessarily appears at the surface 
if one wants to give the model a unique (or unique up to topological degeneracy) ground state.

Our construction of a three-dimensional QCA is as follows.
We consider the three-dimensional Walker-Wang model for the three fermion model.
We show that there exists a QCA which disentangles the ground state of this model.
While we first construct this QCA by an explicit computation using the polynomial method~\cite{Haah2013}, 
in fact the existence of this QCA follows on more general grounds 
using the idea of a ``locally flippable separator'' that we define.
Then, if this three-dimensional QCA could be described by a quantum circuit 
(or even if it were trivial in a more general sense given later),
then we show that this would allow us to consider the Walker-Wang model on a slab 
with top and bottom two-dimensional faces 
and to disentangle the top from the bottom.
This in turn would allow us to construct a commuting projector model for the three fermion model.
Such a commuting projector model is widely believed not to be possible,
since (at least within the TQFT description) the three fermion model 
has a chiral central charge~\cite{frohlich1990braid,rehren1989braid} 
while a commuting projector model should not have such a charge 
(following the ideas of Appendix~D of \cite{Kitaev_2005}),
although we do not know any proof of this in the literature.
Thus, this gives strong evidence that the three-dimensional QCA is nontrivial.

While we do not prove that such a commuting projector model for the three fermion model does not exist, 
we are able to give two useful intermediate steps.
First, we show that no Pauli stabilizer Hamiltonian 
(i.e., a Hamiltonian which is a sum of commuting terms, each of which is product of Pauli operators) 
for the three fermion model exists.
This implies that the three-dimensional QCA is nontrivial as a Clifford QCA,
in that it cannot be decomposed as a quantum circuit of Clifford gates, possibly composed with a shift.
We also show that even if there is a commuting projector model for the three fermion model,
we can construct a nontrivial two-dimensional fermionic QCA.

\subsection{Definitions and Results}

We consider systems whose Hilbert space has a tensor product structure
(or a graded tensor product structure if there are fermions).
We mostly consider finite size systems throughout this paper, proving bounds uniform in system size.
We will sometimes consider infinite systems,
especially if considering finite systems only makes the notation cumbersome
without making exposition clear;
however, every infinite system we consider will have a translation-invariant Hamiltonian,
and hence can always be made into a finite system by periodic boundary conditions.%
\footnote{
While we allow the finite size systems to be translation non-invariant to work in the most general setting,
we note that a QCA on a finite size system on a torus, with the linear size of the torus sufficiently large compared to the range of the QCA, can be used to construct a translation invariant QCA
in an infinite system by going to the universal cover of the torus.
Further, even if the ambient space is not a torus, 
given a translation non-invariant QCA one can construct a translation invariant QCA 
that agrees with the given QCA on a disk using the ``torus trick''~\cite{hastings2013classifying}.
}
Using translationally-invariant systems 
lets us use the polynomial method to prove certain properties of the three fermion Walker-Wang model 
and other Pauli stabilizer Hamiltonians; 
many of these properties were expected to be true on physical grounds previously.
Later we will consider circuits or QCAs, which might not be translationally invariant, 
acting on the terms of such a translationally invariant Pauli stabilizer Hamiltonian; 
even though the resultant Hamiltonian will not be translationally invariant, 
the fact that it is the image of a translationally invariant Hamiltonian 
will let us still apply many of the results for translationally invariant systems.

The Hilbert space is the tensor product of Hilbert spaces for several degrees of freedom.
To define the geometry of the system, we have a set of  ``sites.''
We assume that there is some metric ${\rm dist}(i,j)$ measuring distance between sites.
Our three-dimensional QCA will use a lattice of sites labeled by three-integers $(x,y,z)$.  
There will be a finite number of sites on a $3$-torus; 
we identify sites under translation by $(L,0,0)$, $(0,L,0)$, $(0,0,L)$ 
for some integers $L$ which is the linear size of the system 
(of course, one could also allow different lengths in $x,y,z$ directions without much change).
The metric that we use is a graph metric with sites being distance $1$ if they differ in only one coordinate by $\pm 1$.
Our two-dimensional QCA will use two coordinates $(x,y)$, 
identified under translation by $(L,0)$ or $(0,L)$ with the same metric.

Each degree of freedom will be associated with some site.
There may be multiple degrees of freedom per site.
The metric on sites defines a metric between degrees of freedom: 
the distance between any two degrees of freedom is the distance between the corresponding sites.
Certain degrees of freedom will be referred to as ``qudits.''
For each qudit,
there is a Hilbert space of dimension $D_i$.  
We assume that all $D_i$ are prime without loss of generality 
(if one has a system with a composite dimension on some degree of freedom, 
one may instead define a new system replacing that degree of freedom with several degrees of freedom, 
one for each prime factor).
Other degrees of freedom will be referred to being ``fermionic.''
Let there be $N_f$ such fermionic degrees of freedom.
In this case, for each such degree of freedom $i$, 
we define a pair of operators $\gamma_i,\gamma'_i$ 
obeying the canonical Majorana anti-commutation relations: 
$\{\gamma_i,\gamma_j\}=\{\gamma'_i,\gamma'_j\}=2\delta_{i,j}$ and $\{\gamma_i,\gamma'_j\}=0$.
The Hilbert space of the whole system will be a tensor product of the qudit Hilbert spaces 
times a Hilbert space of dimension $2^{N_f}$.
An operator $O$ will be said to be supported on some set $S$ of sites 
if $O$ is in the subalgebra generated by the Majorana operators $\gamma_i,\gamma'_i$ for
a degree of freedom $i$ on a site in $S$ 
and by the operators on the qudit Hilbert spaces in sites in $S$.

Remark: while we develop much of the formalism using both Majorana and qudit degrees of freedom, 
it may be useful at first reading of this paper to consider only the qudit case.

An operator is said to have odd fermion parity if it anti-commutes 
with the product of $\gamma_i \gamma'_i$ over all fermionic degrees of freedom
and even fermion parity otherwise.
Given two operators $A,B$, we define the supercommutator as
$[A,B\} = \{A,B\}$ if both $A$ and $B$ have odd fermion parity and $[A,B\} = [A,B]$ otherwise.

\begin{definition}
A QCA has range $R$ if, for any operator $O$ supported on a site $i$, 
$\alpha(O)$ is supported on the set of sites within distance $R$ of site $i$.
\end{definition}

Note that we can bound the range of the inverse of a QCA as follows:
\begin{lemma}[\cite{Arrighi2007}]
If $\alpha$ has range $R$, then $\alpha^{-1}$ has range $R$.
\end{lemma}
\begin{proof}
Since $\alpha$ has range $R$, $[\alpha(O_i),O_j]=0$ if $O_i,O_j$ are supported on $i,j$ respectively and ${\rm dist}(i,j)>R$.
So, since $\alpha$ is an automorphism, $[O_i,\alpha^{-1}(O_j)]=0$.
Since $O_i$ is arbitrary, this implies that $\alpha^{-1}$ has range $R$.
\end{proof}

\begin{definition}
A {\bf quantum circuit} is a QCA $\alpha$ that can be written in the form
$\alpha=\alpha_d \circ \alpha_{d-1} \circ \cdots \circ \alpha_{1},$
where each $\alpha_a$ is a QCA that can be written in the form
\[
\alpha_a(O) = \prod_{S \in G_a} U_{S,a}^\dagger O U_{S,a}
\]
where $G_a$ is a collection of disjoint sets of sites, 
%where any two distinct sets $S,T\in G_a$ are disjoint 
and where $U_{S,a}$ is a unitary of even fermion parity supported on $S$.
Each $U_{S,a}$ is called a {\bf gate} on set $S$.
The number $d$ is called the {\bf depth} of the quantum circuit.
We require that the diameter of all $S\in G_a$ for all $a$ be bounded by 
some constant, called the {\bf range} of the gates.
\end{definition}
Note that the range of $\alpha$ is bounded by $d$ times the range of the gates.

Remark: we require that $U_{S,a}$ have even fermion parity 
so that $U_{S,a}^\dagger O U_{S,a}=O$ for any operator $O$ supported on the complement of $S$;
if $U$ supported on $S$ had odd fermion parity,
then for any fermion odd operator $O$ supported on the complement of $S$,
we have $U^\dagger O  U = -O$.
We want to avoid this because we want the gates in the circuit 
to only act nontrivially on operators supported near it.

We consider finite size systems throughout this paper.  
Of course, on a finite size system, any QCA $\alpha$ can be written as a quantum circuit of depth $1$,
using just a single gate acting on the entire system.
Thus, what will be interesting for us is to consider a family of systems of increasing size, 
with a family of QCA for each system having range $R=O(1)$.
We will say that this family of QCA can be realized by a family of quantum circuits 
if each QCA can be realized by a quantum circuit, 
all having depth $O(1)$ with all having range of the gates $O(1)$.

From now on, we will avoid explicitly mentioning the families of QCA, unless needed;
we will simply consider whether certain quantities such as a range are $O(1)$ or not.
Since, for example, the three-dimensional Walker-Wang model 
is defined in a translation invariant fashion, 
it can be constructed for any system size.
Similarly, our explicit construction using the polynomial method of a QCA 
that disentangles the Walker-Wang model 
will give a translation invariant QCA 
so it can be readily defined on any finite size system of size large enough compared to the range of the QCA.

Since we consider families, we can define a group structure for QCAs.
Given two families of QCAs, both of which have the same degrees of freedom, 
we can define their product by composing them.
Since each QCA in each family has range $O(1)$, their product has range $O(1)$.
Note that if, instead of considering families, we considered a single finite size, 
then there is no way to define such a group structure while preserving a meaningful bound on the range of the QCA.

We define three groups of (families of) QCAs.
\begin{definition}
The group $\Circuit$, consists of quantum circuits.
The group $\Circshift$, is defined to consist to quantum circuits, composed with ``shifts,''
where a {\bf shift} QCA has the property that for each prime $p$,
there is a permutation $f_p$ on the set of all qudits $i$ of dimension $D_i=p$.
The shift QCA maps the operators generalized Pauli operators $Z_i,X_i$ on $i$ to $Z_{f_p(i)},X_{f_p(i)}$.
Further, there is a permutation $f_M$ on the set of all Majorana operators $\gamma_i,\gamma'_i$ 
and the shift QCA maps $\gamma_i \mapsto \pm f_M(\gamma_i)$ and $\gamma_i' \mapsto \pm f_M(\gamma'_i)$.
The group $\All$, consists of all QCAs.
\end{definition}
Note that a shift may map $\gamma_i$ to $\gamma_j '$ where $j$ is near $i$.

\begin{lemma}\label{lem:normalsubgroup}
Assume the metric is such that for any distance $r$, 
there is an $O(1)$ bound on the number of sites within distance $r$ of any site (this holds for all lattices that we consider).
Then,
$\Circuit$ is a normal subgroup of $\All$.
\end{lemma}
\begin{proof}
Let $\beta$ be an arbitrary QCA, and let $\alpha\in \Circuit$ be a circuit of depth~1.
We have to show that $\beta \circ \alpha \circ \beta^{-1}$ is a circuit of $O(1)$ depth.
This will imply the result for general $\alpha$,
since $\Circuit$ is generated by circuits of depth~1.

So $\alpha(O) = \prod_{S \in G} U_{S}^\dagger O U_{S}$.
Consider any unitary $U$ and any operator $O$.  
Then, since $\beta$ is an algebra homomorphism, we have
\[
\beta(U^\dagger \beta^{-1}(O) U) = \beta(U^\dagger) \beta\beta^{-1}(O) \beta(U) = \beta(U^\dagger) O \beta(U)
\]
Applying this to each gate $U_S$ we have
\[
\beta\circ\alpha\circ\beta^{-1}(O) = \prod_{S \in G} \beta(U_{S}^\dagger) O \beta(U_{S}).
\]
Since $\beta$ has bounded range, each $\beta(U_S)$ has bounded range.

However, it is possible that the support of $\beta(U_S)$ overlaps with the support of $\beta(U_T)$ for some $S\neq T$, so that this expression for
$\beta\circ\alpha\circ\alpha^{-1}$ is not a quantum circuit of depth~$1$.
However, using the assumption on the metric, 
we can write this as a circuit with a bounded number of rounds: 
In each round, we implement some number of the gates $\beta(U_S)$, 
choosing them not to overlap in any given round.
\end{proof}

When we classify QCAs, it is useful to allow ``stabilization" by adding additional degrees of freedom. 
Given any system defined by a set of sites, a metric, and degrees of freedom, and given any QCA $\alpha$, we stabilize as follows.  
We consider a system with the same set of sites and metric, but with additional degrees of freedom 
(i.e., the degrees of freedom of the new system are a superset of those of the original system), 
and we define a new QCA $\alpha \otimes \Id$, where $\Id$ is the identity QCA on the additional degrees of freedom.

Given these groups and this notion of stabilization, we can now define a ``nontrivial QCA":
\begin{definition}
A (family of) QCAs $\alpha$ is said to be {\bf nontrivial} 
if there is no way to stabilize (i.e., no choice of additional degrees of freedom) such that
$\alpha \otimes I$ is in $\Circshift$.
\end{definition}

Remark: we allow a shift QCA to map Majorana operators with arbitrary signs.
Thus, for example, a QCA $\alpha$ that maps all $\gamma_i,\gamma'_i$ to themselves but changes the sign of one of them 
(i.e.,  $\alpha(\gamma'_i)=\gamma'_i$ for all $i$ and $\alpha(\gamma_i)=\gamma_i$ for all $i$ but one $i=0$ 
for which $\alpha(\gamma_0)=-\gamma_0$) is a shift QCA.
Such a QCA does not preserve fermion parity.
Such a QCA cannot be described by a quantum circuit with our definition of a quantum circuit.
We consider such QCA to be a shift QCA however,
so they are trivial according to the definition above.

It is worth noting that
\begin{lemma}\label{lem:abelianModCircuit}
If stabilization is allowed, the quotient $\All/\Circuit$ is an abelian group.
\end{lemma}
\begin{proof}
Let $\alpha,\beta$ be arbitrary QCA acting on the same system.
Stabilize the system by tensoring with an additional copy of itself.
Let $\Swap$ be the QCA that swaps degrees of freedom between the two copies.
Then,
$(\alpha \otimes \Id) \circ \Swap \circ (\beta \otimes \Id) \circ \Swap=\alpha \otimes \beta$.
Using \cref{lem:normalsubgroup}, 
this implies that $\alpha \circ \beta$ is equal to $\alpha \otimes \beta$ up to an element of $\Circuit$~\cite{Arrighi2007}.

Also, $\Swap \circ (\alpha \otimes \beta) \circ \Swap=\beta \otimes \alpha$, 
so $\alpha \otimes \beta$ is equal to $\beta \otimes \alpha$ up to an element of $\Circuit$.
So, $(\alpha \otimes \Id) \circ (\beta \otimes \Id)$ is equal to $(\beta \otimes \Id) \circ (\alpha \otimes \Id)$ 
up to an element of $\Circuit$.
\end{proof}

As we are discussing groups, let us mention our convention on the verb ``generate.''
A group is said to be generated by a set of elements (that may be infinite),
if the group is equal to the set of all \emph{finite} products of the generators.
Likewise, an algebra is said to be generated by a set of elements (that may be infinite)
if the algebra is equal to the set of all \emph{finite} linear combinations
of \emph{finite} products of the generators.
In fact, this is not just a convention, 
but is the natural way to begin defining groups and algebras,
since infinite products are not a priori defined.
We note this convention explicitly here to avoid any confusion.
For example,
on an infinite lattice, the even subalgebra of Majorana operators
is \emph{generated} by elements of form $\gamma \gamma'$.
One might think that a suitably defined infinite product may leave one $\gamma$ at a site,
and push all the rest to the infinity (which is sometimes useful),
but we do not allow such a product in the algebra generated by finitely supported operators.

\subsection{Outline}

In \cref{separatorsection} we introduce the concept of a ``separator'' and relate separators to QCAs.
Separators are related to the idea of a commuting projector Hamiltonian,
and we will see that, roughly, the classification of separators is the same as the classification of QCAs modulo shifts.
In \cref{WWsection}, we give the main results, constructing a QCA $\alpha_{WW}$ 
from the Walker-Wang model and showing nontriviality of this QCA as a Clifford circuit 
and showing that if this QCA is trivial then a nontrivial two-dimensional fermionic QCA exists.
In \cref{sec:pauli}, we develop a background theory using polynomials under which we find $\alpha_{WW}$,
and prove that $\alpha_{WW}^2 \otimes \Id$ belongs to $\Circshift$.
We also prove that every translation invariant Pauli algebra in one dimension contains a locally generated maximal commutative algebra.
This is used to show that every Pauli stabilizer Hamiltonian that is topologically ordered in two dimensions
contains a nontrivial boson, which in turn is used in the proof that $\alpha_{WW}$ is not a Clifford circuit.
In \cref{discussionsection}, we discuss some future directions.
\Cref{sec:fBV} contains a result on ``2-local'' commuting Hamiltonians with fermionic degrees of freedom;
this result is unrelated to the rest of the paper.

\section{Separators}
\label{separatorsection}

In this section, we define the idea of ``separators.''
Separators are closely related to commuting projector Hamiltonians, so we briefly recall that idea first.
Such Hamiltonians are a sum of local commuting projectors, so that
$H=\sum_i \Pi_i$ where $\Pi_i$ are projectors with $[\Pi_i,\Pi_j]=0$ and with each $\Pi_i$ supported on some set of bounded diameter.
These Hamiltonians are especially interesting when the ground state has zero energy so that mutual $0$ eigenspace of all projectors $\Pi_i$ has a nonzero dimension.

In quantum information theory, a ``syndrome'' of a given state 
refers to the result of measuring all projectors $P_i$ for that state, 
i.e., it is an assignment of a value $0$ or $1$ for each $i$.
For certain commuting projector Hamiltonians, 
each syndrome defines a one-dimensional subspace of states.
A ``separator'' generalizes this idea.
The following definition is for systems with a finite number of sites, 
each having a finite dimensional degrees of freedom.
\begin{definition}\label{def:separator}
A {\bf separator} is an indexed set of unitary operators $\opZ_a$ satisfying all of the following properties:
\begin{enumerate}
\item Every $\opZ_a$ has fermion parity even.
\item For each index $a$ there is an integer $D_a \geq 2$ such that $\opZ_a^{D_a}=I$.
\item $[\opZ_a,\opZ_b]=0$ for all $a,b$.
\item Every $\opZ_a$ is supported on a disk of diameter $R=O(1)$.  
\item For any arbitrary assignment $a \mapsto \omega(a)$, where $\omega(a)$ is a $D_i$-th root of unity,
%there is a one-dimensional subspace of the Hilbert space for which each $\opZ_a$ has eigenvalue $\omega(a)$, 
% i.e., 
the space of states $\ket\psi$ such that $\opZ_a \ket\psi=\omega(a)\ket\psi$ for all $a$ is one-dimensional.
\end{enumerate}
For every element $\opZ_a$ of the separator, if $D_a=2$, 
we will specify that the element is either a ``qudit element'' or a ``fermionic element.''
Every element with $D_a>2$ is a ``qudit element.'' 
\end{definition}
We identify two separators if they are the same up to relabeling of the index $a$.

One may think of the assignment $\omega(a)$ as 
specifying the result of a multi-outcome projective measurement for each $a$.
That is, since the $\opZ_a$ commute, they can be simultaneously diagonalized, 
and we can think of the set of $\omega(a)$ as defining the outcomes of measuring all the $\opZ_a$.
% We will refer to a given $\opZ_a$ as a ``element'' of the separator.
Note that the number of possible choices of $\omega(a)$, 
which may be regarded as the number of possible outcomes one can obtain by performing all of these multi-outcome measurements,
must equal to the dimension of the Hilbert space of the system,
so $\prod_a D_a$ is equal to the dimension of the Hilbert space of the system.

While we have demanded that the eigenvalues of $\opZ_a$ be $D_a$-th roots of unity,
this is just chosen for convenience later.
We can define a separator from any commuting set of operators,
each having $D_a$ distinct eigenvalues and having bounded support,
if the common eigenspace of these operators is always one-dimensional.

For each qudit $j$, we define ``computational'' basis states $\ket 0, \ket 1, \ldots, \ket{D_j-1}$,
and also a generalized Pauli operator 
\begin{align}
Z_j = \sum_{k=0}^{D_j -1} \exp(2\pi i k /D_j) \ket k \bra k.
\end{align}
We say that
\begin{definition}
A separator is {\bf trivial}
if there is a bijection $f(\cdot)$ from elements of the separator to degrees of freedom such that,
for each index $a$,
$\opZ_a = Z_{f(a)}$ for a qudit element
and $\opZ_a = i \gamma'_{f(a)} \gamma_{f(a)}$ for a fermionic element.
\end{definition}

\begin{definition}
A {\bf local flipper} $\{ \tilde \opX_a\}$ associated with a separator $\{ \opZ_a \}$ 
is a set of unitary operators indexed by the same index set as the separator
such that  
\begin{enumerate}
\item every $\tilde \opX_a$ is supported on a disk of radius $R=O(1)$,
\item each $\tilde \opX_a$ commutes with all $\opZ_b$ with $b \neq a$ but 
$\opZ_a \tilde \opX_a =\tilde \opX_a \opZ_a \exp\left(\frac{2 \pi i}{D_a}\right)$, and
\item $\tilde \opX_a$ is of even fermion parity if $\opZ_a$ is a qudit element, 
or of odd fermion parity if $\opZ_a$ is fermionic.
\end{enumerate}
A separator is {\bf locally flippable} if there exists an associated local flipper.%
\footnote{
The term ``flipper'' is natural;
for $D_a>2$ the change is not just ``flipping'' an outcome but we use this terminology anyway.
}
\end{definition}
Clearly, then, every trivial separator is locally flippable ---
simply take $\tilde \opX_a$ to be the (generalized) Pauli operator 
$X = \sum_{j=0}^{D-1} \ket{j+1 \mod D}\bra{j}$ on site $f(a)$
for qudit elements and take $\tilde \opX_a = \gamma_{f(a)}$ for fermionic elements.

We emphasize that the definition of a locally flippable separator 
does {\em not} impose any requirements on the supercommutators $[\tilde \opX_a,\tilde \opX_b\}$.
If all of these supercommutators vanished, 
then the set of operators $\opZ_a,\tilde \opX_a$ for qudit elements 
would give a representation of a tensor product of the algebra of generalized Pauli operators, 
and the operators $\tilde \opX_a, i\tilde \opX_a \tilde \opZ_a$ for fermionic elements
would give a representation of the algebra of Majorana operators.
In this case, a locally flippable separator 
would be mapped to a trivial separator by a QCA of range $O(1)$:
We can define a QCA to map a qudit element $\opZ_a$ to the operator $Z_{f(a)}$ 
on some qudit $f(a)$ with dimension $D_{f(a)}=D_a$,
and to map a fermionic element $\opZ_a$ to $i \gamma_{f(a)} \gamma'_{f(a)}$,
where $f$ is a bijection from elements of the separator to degrees of freedom such that ${\rm dist}(a,f(a))=O(1)$.

The question whether such a mapping $f(\cdot)$ exists is not obvious,
but if there are no fermionic degrees of freedom
we will show that this is possible using the Hall marriage theorem.%
\footnote{ 
the application of the marriage theorem to problems in QCA was first suggested by M. Freedman.
}
If there are fermionic elements, 
then the application of the Hall marraige theorem 
does not seem to solve this problem of finding $f(\cdot)$, 
since we must map fermionic elements of the separator to fermionic degrees of freedom
and qudit elements of the separator to qudit degrees of freedom,
but it is possible that the numbers of these degrees of freedom do not match.
For example, the graded algebra of two fermionic degrees of freedom (four Majorana modes)
is the same as that for one fermionic degree of freedom and one qubit,
so one could have one complex fermion and one qubit on one site 
but have a separator with two fermionic elements.
So, in this case we will need to make an additional assumption of the existence of a bijection $f(\cdot)$.
We will explain how this assumption can be made to hold by stabilization later.

The construction of QCA from a locally flippable separator that we just explained briefly,
assumes that the local flipper elements are supercommuting.
However, we do {\it not} know if those supercommutators vanish.
Nevertheless, in \cref{sec:disent} we show how it is still possible 
to define a QCA that maps a locally flippable separator to a trivial separator.

\subsection{Pauli stabilizer models and Examples}

We now clarify the definition of a separator and a locally flippable separator by examples.
Let us consider any Hamiltonian which is a sum of terms, each of which is a product of Pauli operators 
(often called a ``Pauli stabilizer Hamiltonian'').

First, a toric code on a sphere has a set of terms that satisfy the first four conditions to be a separator.  
The fifth condition is not satisfied, since there is a redundancy in the terms of the Hamiltonian, 
that the product of all plaquette terms (or the product of all star terms) is equal to~$+1$.
If a single plaquette and a single vertex term are removed, then this does define a separator.
However, this separator is not locally flippable because of the conservation law of topological charges.

Second, consider a two dimensional Ising model $H = - \sum_{\langle j k \rangle} Z_j Z_k$.
The set $\{ Z_j Z_k \}$ does not define a separator because there is a redundancy among the terms.
In this case, the redundancy cannot be removed by neglecting only a small fraction of terms.

Third, we show later that the Walker-Wang model for the three fermion theory
has another set of stabilizer generators with no redundancy;
this is one of our main results.
From these, we define a separator, and further show that this stabilizer generating set is locally flippable;
this was to be expected on physical grounds since the bulk of the model has no anyons.

\subsection{Disentangling locally flippable separators by QCA} \label{sec:disent}

Our main result in this section is:
\begin{theorem} \label{thm:qcadisent}
Given any locally flippable separator in a system without fermionic degrees of freedom,
there exists a QCA $\alpha$ of range $O(R)$ such that
the image of the locally flippable separator under the QCA is a trivial separator.
That is, there is a bijection $f(\cdot)$ from elements to degrees of freedom such that
$\alpha(\opZ_a)=Z_{f(a)}$, where $Z_{f(a)}$ is the (generalized) Pauli $Z$ operator on degree of freedom $f(a)$.

If the system has fermionic degrees of freedom, 
then the same result holds under the additional assumption that 
there is a bijection $g(\cdot)$ 
from element indices $a$ of the separator with $D_a=2$ 
to degrees of freedom $g(a)$ with $D_a=2$ 
such that $\dist(a,g(a))=O(R)$ and such that 
if $\tilde \opX_a$ has even fermion parity then $g(a)$ is a qudit,
and if $\tilde \opX_a$ has odd fermion parity then $g(a)$ is a fermionic degree of freedom.
That is, under this assumption,
there is a bijection $f(\cdot)$ from elements to degrees of freedom such that
$\alpha(\opZ_a)=Z_{f(a)}$ if $\opZ_a$ is a qudit element and
$\alpha(\opZ_a)=i\gamma'_{f(a)} \gamma_{f(a)}$ otherwise. 
Here, $f(a)=g(a)$ whenever $D_a =2$.
\end{theorem}
\begin{definition}
We say that
the QCA $\alpha$ of \cref{thm:qcadisent} {\bf disentangles} the locally flippable separator.
\end{definition}
The additional assumption in the fermionic case 
can always be made to hold by stabilization
where we add degrees of freedom to the system and elements to the separator:
For each element~$\opZ_a$ of the given separator with dimension $D_a=2$, 
we add a fermionic or qubit degree of freedom near~$\opZ_a$ 
and set this added degree of freedom to be the image $g(a)$.
For any other degree of freedom $j$ with $D_j = 2$ that was existing from the outset,
we add a trivial element to the seperator.

To prove \cref{thm:qcadisent},
we first define a set of unitary operators $\opX_a$ such that
$\opX_a$ and $\tilde\opX_a$ have the {\it same} fermion parity as each other for all $a$, 
and 
\begin{align*}
[\opX_a,\opX_b\}&=0 &\text{ for all } a, b, \\
[\opZ_a,\opX_b] & = 0 & \text{ for all } a\neq b,\\
\opZ_a \opX_a &= \opX_a \opZ_a \exp\left(\frac{2 \pi i}{D_a}\right)& \text{ for all } a
\end{align*}
by considering a representation of the separator and the flipper.
Having done this, we will next show that the $\opX_a$ are supported on sets of bounded diameter,
and finally construct the disentangling QCA $\alpha$.

\begin{proof}[Proof of \cref{thm:qcadisent}]
Let $\ket 0$ be a state which is the $+1$ eigenstate of all operators $\opZ_a$.
In this proof, we write vectors such as $\vec v$ and $\vec w$
to refer to a vector of length equal to the number of elements in the separator.
We fix an arbitrary ordering of the elements in the separator 
(so one may regard the index $a$ as an integer),
and the $a$-th entry $v_a$ of $\vec v$ will be chosen from $\{ 0,\ldots,D_a-1 \}$.
For any such vector $\vec v$, let $\ket{\vec v}$ be the state given by
\begin{align}
\ket{\vec v} = \left( \prod_a \tilde \opX_a^{v_a} \right) \ket 0,
\end{align}
where we order the product in order of increasing $a$.%
\footnote{
For example, for $v=10102$ we set 
$\ket{\vec v} = \tilde \opX_1 \tilde \opX_3 \tilde \opX_5^2 \ket 0$ where $D_5\geq 3$.
}
By definition of a separator, the states $\ket{\vec v}$ are an orthonormal basis for the Hilbert space;
the vector is simply a way of describing a particular basis state in the eigenbasis of the $\opZ_a$.

Then, if $\tilde \opX_a$ has \emph{even} fermion parity,
we define $\opX_a$ to be the operator such that
\begin{align}
\opX_a \ket{\vec v} = \ket{ \vec v + \vec e_a}
\end{align}
for all $a$ and $\vec v$,
where $\vec e_a$ is the vector with a sole nonzero entry $1$ in its $a$-th component.
To define~$\opX_a$ when $\tilde \opX_a$ has \emph{odd} fermion parity,
let $p_c = 0,1$ denote the fermion parity of $\tilde \opX_c$ for any $c$;
$p_c = 0$ for even $c$, $p_c = 1$ for odd $c$.
With this notation, we define%
\footnote{
\cref{eq:defopX} covers the even fermion parity case,
but we display it separately to emphasize that
there is nothing delicate about the sign $(-1)^{p_a \sum_{ c < a } v_c p_c }$.
If $\tilde \opX_a$'s supercommuted,
then our construction of $\opX_a$ should be redundant,
and the sign factor of \cref{eq:defopX} is a necessary consistency condition.
}
\begin{align}
\opX_a \ket{\vec v} = (-1)^{p_a \sum_{ c < a } v_c p_c } \ket{ \vec v + \vec e_a }. \label{eq:defopX}
\end{align}

This definition gives $\opX_a$ the same fermion parity as $\tilde \opX_a$, 
since each state $\ket v$ has a definite fermion parity.
The commutation relations 
$[\opX_a,\opX_b\} = 0$
and $\opZ_a \opX_b = \opX_b \opZ_a \exp( 2\pi i \delta_{ab}/D_a )$
follow immediately.

In order to show the locality of the operators $\opX_a$, we use the following.
Since $\tilde \opX_a$ does not commute with $\opZ_a$, 
there must be a site $s_a$ in the support of both operators.
Then, $\tilde \opX_a$ and $\opZ_a$ are both supported within distance $R$ of $s_a$,
where $R$ is a constant that bounds the diameter of disks 
on which $\tilde \opX_a$ and $\opZ_a$ are supported.
\begin{lemma} \label{lem:lcom}
$[\tilde \opX_a,\opX_b\}=0$ if ${\rm dist}(s_a,s_b) > 4R$.
\end{lemma}

\begin{proof}
For any $a,b$, we define a unitary operator 
$\phi(a,b) = \tilde \opX_a \tilde \opX_b \tilde \opX_a^{-1} \tilde \opX_b^{-1}$,
which is diagonal in the basis $\{ \ket{\vec v} \}$.
Note that for $\dist(s_a,s_b) > 2R$, we have $[\tilde \opX_a,\tilde \opX_b\}=0$,
so $\phi(a,b) = (-1)^{p_a p_b}$.

We will compute $\opX_b \tilde \opX_a \ket{\vec v}$ and $\tilde \opX_a \opX_b \ket{\vec v}$
for any $\ket{ \vec v }$.
Up to a phase, both of these are equal to $\ket{ \vec v + \vec e_a + \vec e_b}$.
We have to compare the phases:
\begin{align}
\bra{ \vec v+\vec e_a+ \vec e_b } \opX_b \tilde \opX_a \ket{ \vec v }
&= (-1)^{\sum_{d < b} (v_d +\delta_{ad}) p_d p_b} 
\bra{ \vec v+\vec e_a}  \tilde \opX_a \ket{ \vec v }, \label{eq:ph0} \\
%%%%%%%%%%%%
\bra{ \vec v+\vec e_a+\vec e_b } \tilde \opX_a \opX_b \ket{ \vec v }
&= (-1)^{\sum_{d < b} v_d p_d p_b} 
\bra{ \vec v+\vec e_a+\vec e_b } \tilde \opX_a \ket{ \vec v + \vec e_b }
\end{align}
where $\delta$ is the Kronecker $\delta$-function.
We can express $\tilde \opX_a$ in terms of $\opX_a$ and $\phi(a,b)$ 
by successively commuting $\tilde \opX_a$ through $\tilde \opX_c$ for $c < a$:%
\footnote{
For the example $v=10102$ above, the state $\tilde \opX_4 |v\rangle$ is equal to
$\tilde \opX_4 \tilde \opX_1 \tilde \opX_3 \tilde \opX_5^2 \ket 0
=
\phi(4,1) \tilde \opX_1 \tilde \opX_4 \tilde \opX_3 \tilde \opX_5^2 \ket 0
=
\phi(4,1) \tilde \opX_1  \phi(4,3) \tilde\opX_3 \tilde \opX_4 \tilde \opX_5^2 \ket 0$.
} 
\begin{align}
\tilde \opX_a \ket{\vec v} =
\left( \prod_{c < a} (\phi(a,c) \tilde \opX_c)^{v_c} \right)
\left( \prod_{c\geq a} \tilde\opX_c^{v_c+\delta_{ac}} \right)  \ket 0 .
\end{align} 
So,
\begin{align}
\bra{ \vec v + \vec e_a } \tilde \opX_a \ket{ \vec v }
&=
\bra 0 
\left( \prod_{c} \tilde\opX_c^{v_c+\delta_{ac}} \right)^\dagger
\left( \prod_{c< a} (\phi(a,c) \tilde \opX_c)^{v_c} \right) 
\left( \prod_{c\geq a} \tilde\opX_c^{v_c+\delta_{ac}} \right)  
\ket 0, \label{eq:ph1} \\
%%%%%%%%%%%%%%%%%%
\bra{ \vec v + \vec e_a + \vec e_b } \tilde \opX_a \ket{ \vec v + \vec e_b }
&=
\bra 0
\left( \prod_{c} \tilde\opX_c^{v_c+\delta_{ac}+\delta_{bc}} \right)^\dagger
\left( \prod_{c < a} (\phi(a,c) \tilde \opX_c)^{v_c+\delta_{bc}} \right)
\left( \prod_{c\geq a} \tilde\opX_c^{v_c+\delta_{ac}+\delta_{bc}} \right)
\ket 0. \label{eq:ph2}
\end{align}
\cref{eq:ph1} and \cref{eq:ph2} both evaluate to some complex number of unit norm, a phase factor.
We need to show that this phase is the same for both, 
with an extra factor of $-1$ if both $\tilde \opX_a$ and $\opX_b$ have odd fermion parity.
Consider any term $\phi(a,c)$ in either \cref{eq:ph1} or \cref{eq:ph2}.
This in either equation acts on some state $\ket{ \vec w }$
which is one of our basis states up to a phase.%
\footnote{
For example, in \cref{eq:ph1}, suppose that $v_c=1$ for some $c<a$.
Then, $\phi(a,c)$ appears once and acts on the state
$\left(\prod_{d > c} \tilde\opX_d^{v_d+\delta_{ad}} \right) \ket 0$ up to a phase.
If $v_c=2$, then $\phi(a,c)$ appears twice
and they act on the states 
$\tilde \opX_c \left( \prod_{d>c}  \tilde\opX_d^{v_d+\delta_{ad}} \right) \ket 0$
and
$\left( \prod_{d>c} \tilde\opX_d^{v_d+\delta_{ad}} \right) \ket 0$, respectively.
}
So, each such $\phi(a,c)$ contributes some phase factor $\bra{ \vec w } \phi(a,c) \ket{ \vec w}$
for some basis state $\ket { \vec w }$ that depends on $a,b,c$;
that is, we can replace each such $\phi(a,c)$ by a phase factor 
$\bra{ \vec w } \phi(a,c) \ket{ \vec w }$.

For $c\neq b$, there is an obvious way to make a correspondence between terms $\phi(a,c)$ 
in \cref{eq:ph1} and those in \cref{eq:ph2}:
a term $\phi(a,c)$ appears exactly $v_c$ times in each equation,
and we make a correspondence between these terms \emph{in order}.
If some term contributes a phase
$\bra{ \vec w } \phi(a,c) \ket{ \vec w }$ in \cref{eq:ph1},
then in \cref{eq:ph2} it contributes a phase
$\bra{ \vec w } \phi(a,c) \ket{ \vec w }$ for $c > b$, 
or 
$\bra{ \vec w+\vec e_b } \phi(a,c) \ket{ \vec w + \vec e_b }$ for $c < b$.
(Recall that we have dropped any $\phi(a,c)$ with $c=b$.)
Clearly the two phase factors are the same if $c > b$.
For $c < b$, note that
$\bra{ \vec w+\vec e_b } \phi(a,c) \ket{ \vec w + \vec e_b }
=
\bra{ \vec w } \tilde \opX_b^{-1} \phi(a,c) \tilde \opX_b \ket{ \vec w }$,
since
$\tilde \opX_b \ket{ \vec w } = \ket{ \vec w+\vec e_b }$ up to a phase
that cancels between bra and ket.
However, for $\dist(s_a,s_b) > 4R$, either $\phi(a,c) = \pm I$, or
%we have (trivially) $\dist(s_a,s_b) > 2R$ 
$\dist(s_c,s_b) > 2R$ by a triangle inequality.
So $[\phi(a,c),\tilde \opX_b]=0$ always, and 
$\bra{ \vec w } \tilde \opX_b^{-1} \phi(a,c) \tilde \opX_b \ket{ \vec w } 
= \bra{ \vec w } \phi(a,c) \ket{ \vec w }$.
Therefore, the product of phase factors for $c\neq b$ 
is the same for both \cref{eq:ph1,eq:ph2}.

If $b < a$, 
then there is $\phi(a,b)$ from \cref{eq:ph2},
and no extra sign from \cref{eq:ph0}.
For $\dist(s_a,s_b) > 2R$, we know $\phi(a,b)= (-1)^{p_a p_b} I$.
If $b > a$,
then there is no $\phi(a,b)$ from \cref{eq:ph2},
but an extra sign $(-1)^{p_a p_b}$ from \cref{eq:ph0}.
\end{proof}

Since the vectors $\ket{\vec v}$ form a basis, 
the set of all separator and flipper elements generates the algebra of all operators,
which is a simple $\dagger$-algebra.
For each $a$, consider a simple $\dagger$-subalgebra $\mathcal E_a$ generated by $\opZ_b$ and $\opX_b$
for $b$ such that $\dist(s_a, s_b) > 4R$.
By \cref{lem:lcom}, a flipping operator~$\tilde\opX_a$ commutes with every generator of $\mathcal E_a$.
Hence, the operator $U_a = \opX_a \tilde \opX_a^{-1}$ 
that belongs to the commutant of $\mathcal E_a$ within the algebra of all operators,
must be in the algebra generated by $\opZ_c$ and $\opX_c$ for $c$ such that
$\dist(s_a,s_c) \leq 4R$.%
\footnote{
To see this formally, one can expand an operator as a linear combination of products of 
$\opZ_a$ and $\opX_a$, and examine commutation relations to remove unwanted terms.
See \cref{lem:supercommutant} for a general result.
}
Since $U_a$ is diagonal in the basis $\{ \ket{\vec v} \}$,
it belongs to the algebra generated by $\opZ_c$,
and therefore, $\opX_a = U_a \tilde \opX_a$ is supported within distance $5R$ of $s_a$.
This establishes the locality of the operators $\opX_a$ that we wished to show.

Finally, we will define a bijection $f$ from elements of the separator to sites,
and define a QCA $\alpha$ to map $\opZ_a$ to $Z_{f(a)}$ and $\opX_a$ to $X_{f(a)}$.

The inverse QCA $\alpha^{-1}$ will map $Z_{f(a)}$ to $\opZ_a$ and $X_{f(a)}$ to $\opX_a$.
We can bound the range of $\alpha^{-1}$ (and hence the range of $\alpha$) if for each element $a$, we have a bound on ${\rm dist}(s_a,f(a))$.

We claim that it is possible to choose $f(a)$ such that ${\rm dist}(s_a,f(a))\leq 5R$ 
so that $\alpha^{-1}$ has range at most $10R$.
The Hall marriage theorem enables us to establish this bound 
if there are no fermionic degrees of freedom.%
\footnote{
The marriage theorem~\cite[Chap.~22~Thm.~3]{BookProof} states that
given a bipartite graph between ``left'' vertices and ``right'' vertices,
if, for every subset $W$ of left vertices, the neighbor of $W$ in the right
has at least as many elements as $W$,
then there is a one-to-one mapping from the left to the right.
}
Let us say that for each element $a$, a degree of freedom $j$ is 
``acceptable'' if $\dist(s_a,j)\leq 5R$ and $D_a=D_j$.
Let $W$ be any subset of elements (indices) of the separator,
and $T$ be the set of sites within distance $5R$ from the set of $s_a$ for $a \in W$.
By the locality of $\opX_a$,
the set $T$ supports the algebra generated by $\opX_a$ and $\opZ_a$ for $a \in W$.
Hence, if $\prod_{a \in W} D_a$ has a prime factorization $\prod_{a \in W}D_a= 2^{n_2} 3^{n_3} \cdots$,
we see that the set $T$ must contain at least $n_d$ qudits of dimension $d$.
This implies that any degree of freedom acceptable to $a \in W$ is on some site in $T$,
and furthermore the number of acceptable qudits in $T$ 
is at least the number of elements in $W$.
So, the marriage condition is obeyed.
If there are fermionic degrees of freedom, 
we apply the Hall marriage theorem to construct $f(a)$ for elements $a$ with $D_a>2$, 
and use the assumption of the existence of $g(\cdot)$ to set $g(a)=f(a)$ for $D_j=2$.

This completes the proof of \cref{thm:qcadisent}.
\end{proof}

\subsection{Mapping from QCA to locally flippable separator and Inverse}

In \cref{thm:qcadisent}, we have shown that there exists a QCA that disentangles a locally flippable separator.
This QCA is not unique; many possible choices may be made.
We can use this to define a map $\mathcal F$ from locally flippable separators to QCAs 
by picking an arbitrary QCA that fulfills the conditions of the lemma.

Conversely, we can define a map $\mathcal G$ from QCAs to locally flippable separators in an obvious fashion 
by considering the image of the trivial separator under a given QCA.  
That is, the elements of the locally flippable separator will be the images of the elements of the trivial separator.

The composition ${\cal F} \circ {\cal G}$ is a map from QCAs to QCAs.
Let us consider first the case of a qudit QCA without any fermionic degrees of freedom.
Then, the image of any shift QCA under ${\cal G}$ is the trivial separator.
Hence, for any shift QCA $\beta$, and any QCA $\alpha$, 
we have ${\cal F}({\cal G}(\alpha \circ \beta))={\cal F}({\cal G}(\alpha))$.
This shows one advantage to considering locally flippable separators rather than QCAs.
By considering a locally flippable separator defined from a QCA, we mod out shifts.
This is useful since our goal is to classify QCAs which are not in $\Circshift$.
See a remark below \cref{lem:elemCphase}.

If we have fermionic degrees of freedom, 
then there are shifts which map the trivial separator to some other separator.  
One standard example is the Majorana chain: 
a shift on a one-dimensional chain can map a separator with elements $i \gamma_i \gamma'_i$ 
to one with elements $i \gamma_i \gamma'_{i+1}$.  
The sum of elements of the first separator is often called the ``trivial Hamiltonian,''
while the sum of elements of the second separator is often called the ``nontrivial Majorana chain Hamiltonian.''
So, in this case considering locally flippable separators does not completely mod out shifts;
it mods out shifts of all qudit degrees of freedom at least.

\section{Three-dimensional Walker-Wang model} \label{WWsection}

The Walker-Wang models~\cite{Walker_2011} are 
a class of $3$-dimensional gapped commuting projector Hamiltonian lattice models with no bulk anyon excitations.
When a $2$-dimensional surface boundary is introduced, 
a Walker-Wang model can be terminated at this surface in such a way that it remains a gapped commuting projector model.
However, such a surface will typically host a non-trivial $2$-dimensional topological order, 
i.e., anyons confined to move only on the surface~\cite{Curt}.
In fact, given any complete and consistent algebraic description of a theory of anyons 
--- i.e., a unitary modular tensor category ---
one can always build an associated Walker-Wang model 
that realizes this theory of anyons at its surface.%
\footnote{
More generally, Walker-Wang models can take as input a premodular category, 
in which case the 3D bulk also contains some deconfined topologically non-trivial excitations.
}
Its ground state can then be thought of as 
a superposition of string-net configurations in $3$ spatial dimensions.
Each string-net configuration can be interpreted as 
the space-time history of a fusion and braiding process in $2+1$ dimensions,
and the Walker-Wang ground state's amplitude of such a string-net configuration 
is proportional to the amplitude of this braiding process.
Walker-Wang models are interesting for condensed-matter physics 
primarily because they often provide exactly solvable points 
for symmetry protected topological (SPT) phases~\cite{BCFV, TopoSC, TI, ProjS}.  

In the present paper we will only be dealing with one specific example of a Walker-Wang model,
based on the $3$-fermion modular tensor category (UMTC) $\{1,\fone,\ftwo,\fthree \}$~\cite{Kitaev_2005}.
As an abelian theory of anyons, the $3$-fermion UMTC describes 
$3$ non-trivial quasiparticles $\fone,\ftwo,\fthree$, 
which are all fermions and any pair of which has full braiding phase of $-1$.
The fusion rules are $\Z_2 \times \Z_2$.
Physically, the $3$-fermion UMTC can be realized as 
a $U(1)$ Chern-Simons theory whose $4 \times 4$ $K$-matrix 
is equal to the Cartan matrix of $SO(8)$~\cite{BCFV}.
This Chern-Simons theory is a low energy description of 
a $2$-dimensional fractional quantum Hall (FQH) state of bosons, 
whose (non-commuting-projector) Hamiltonian can be constructed directly from the $K$-matrix~\cite{VL}.
Importantly, because the signature of the $K$-matrix is equal to $4$ 
--- in particular, is nonzero --- 
this FQH state is chiral.
In fact, it is believed~\cite{Kitaev_2005} that the chiral central charge 
$c_{-}$ of a 2d topologically ordered system 
is related to the anyon statistics as follows:
\begin{align}
e^{2 \pi i c_{-} / 8} = \frac{\sum_a d_a^2 \theta_a}{\sqrt{\sum_a d_a^2}}
\end{align}
where the sum is over all anyons $a$ with quantum dimensions $d_a$ and topological spins $\theta_a$.
For the $3$-fermion theory the right hand side is $-1$, 
which fixes $c_{-}$ to be equal to $4$ modulo $8$ in any $2D$ physical realization of the $3$-fermion theory.
Thus, assuming this formula, the chiral central charge 
would have to be non-zero in any $2d$ physical realization of the $3$-fermion theory.
In particular, no time reversal invariant realization would be possible, 
since an edge chirality manifestly breaks time reversal.

On the other hand, the algebraic data defining the $3$-fermion UMTC 
(namely, the $F$ and $R$ matrices) 
can be chosen to all simultaneously be real~\cite{Kitaev_2005}.
Thus the fusion and braiding amplitudes of all braiding processes are real, 
and therefore so is the ground state of the Walker-Wang model 
built on the $3$-fermion theory (up to an overall phase).
In fact, the Hamiltonian itself can be taken to be real.
This is a commuting Pauli Hamiltonian, 
with two spin-$\frac{1}{2}$ degrees of freedom per link of a cubic lattice.
It was originally written down~\cite{BCFV} as a model of 
the `beyond-cohomology' SPT phase of bosons with $\Z_2$ time reversal symmetry in 3 spatial dimensions.

\subsection{Model and its properties}

\subsubsection{Model}

We work with the $3$-fermion Walker-Wang model Hamiltonian constructed in~\cite{BCFV}, which we review here.
There are two qubits on every link $\ell$ of the $3d$ cubic lattice (i.e., the ``links" are the ``sites" of the system), 
with Pauli algebra generated by $X^{\ell}_i$, $Z^{\ell}_i$, $i=1,2$.
The Hamiltonian is a sum of vertex $(A_V)$ and plaquette $(B_P)$ terms:
\begin{align}
H_{WW} = - \sum_V A_V - \sum_P B_P. \label{eq:WWHam}
\end{align}
The first sum above is over all vertices 
and the second over the square plaquettes.
We define the vertex term as:
\begin{align}
A_V = \prod_{\ell \sim V} X_1^\ell + \prod_{\ell \sim V} X_2^\ell \label{eq:vertexterm}
\end{align}
where $\ell \sim V$ means that $V$ is one of the endpoints of $\ell$.

\begin{figure*}
\centering
\includegraphics[width=0.8\textwidth,trim={1cm 4cm 6cm 1cm},clip]{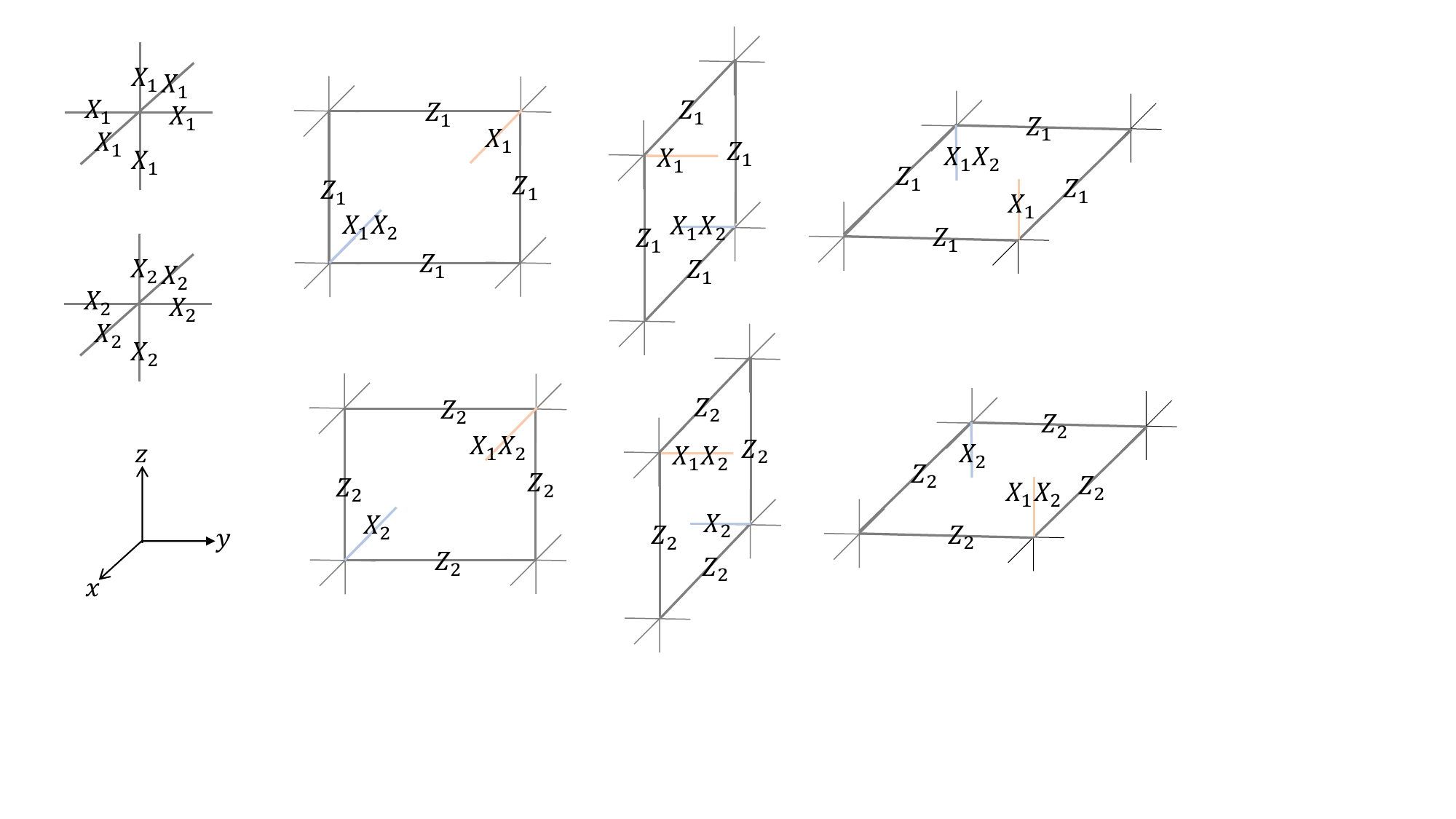}
\caption{Three-fermion Walker-Wang model.
There are two qubits per link
and the ground state is the common eigenstate of eigenvalue $+1$
of the vertex terms in \cref{eq:vertexterm} and the plaquette terms in \cref{eq:plaquetteterm}.  The orange and blue links lie `over' and 'under' the plaquette for this choice of projection, and are referred to as $O$ and $U$ links in the main text.
}
\label{fig:WW}
\includegraphics[width=0.7\textwidth, trim={0cm 6.5cm 15cm 0cm},clip]{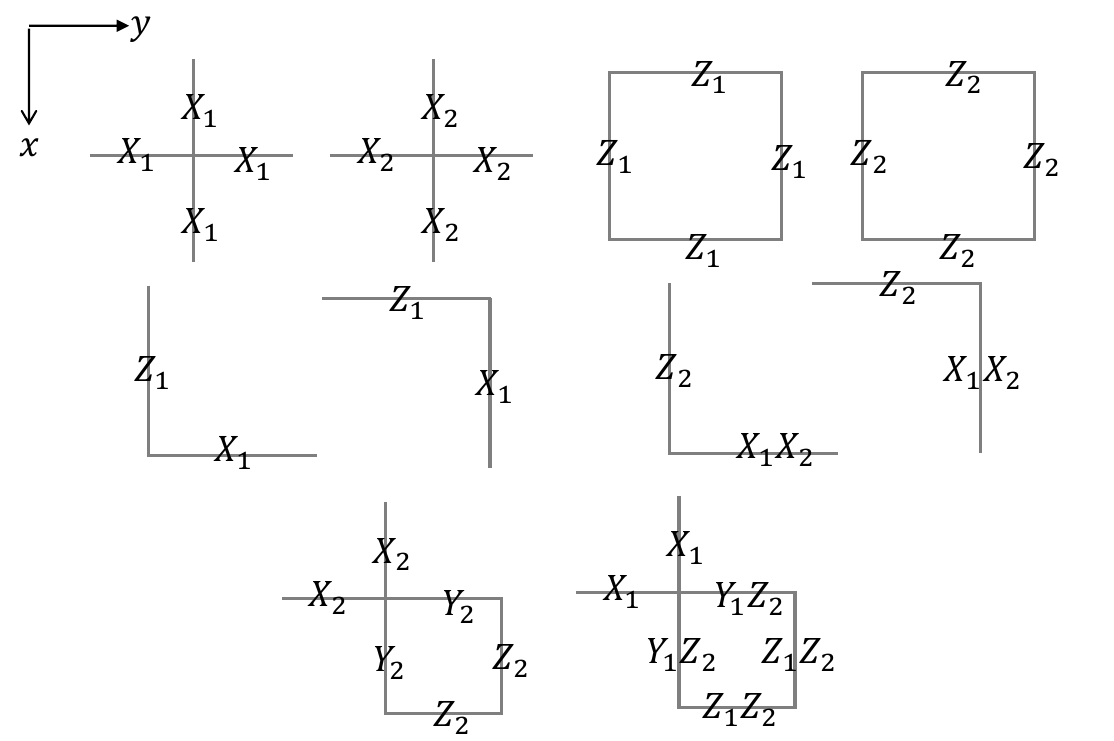}
\caption{Boundary terms of the three fermion Walker-Wang model.
Those on the top two rows are from \emph{truncated bulk terms},
and the two figures on the bottom, that we call \emph{small fermion loop operators},
are inserted for the theory with boundary to satisfy the local topological order condition.
In the top two rows, only the operator factors acting on the boundary ($z=0$) is shown;
e.g., the vertex operator on the top left
is in fact a product of five $X$'s, but one factor that acts 
on the link perpendicular to the $z=0$ plane is not shown.
They may look noncommuting, but they in fact commute with hidden factors included.
The small fermion loop operators are supported genuinely on the $z=0$ plane.
}
\label{fig:WWsurface}
\end{figure*}

To define the plaquette term $B_P$, 
we first fix a projection of the $3d$ cubic lattice into $2d$, 
in such a way that the projection of each plaquette 
is in one of the three forms shown in \cref{fig:WW}.
Out of all the links adjacent to vertices of a given plaquette,
there are exactly two that lie in the interior of its projection.
We label those as $O$ and $U$, 
depending on whether they lie `over' or `under' the plaquette.
Then we define $B_P = B_{P,1} + B_{P,2}$, where
\begin{align}
B_{P,1} &= X_1^O X_1^U X_2^U \prod_{\ell \in \partial P} Z_1^\ell , \nonumber\\
B_{P,2} &= X_1^O X_2^O X_2^U \prod_{\ell \in \partial P} Z_2^\ell . \label{eq:plaquetteterm}
\end{align}
The terms $A_V, B_{P,1}$ and $B_{P,2}$ all commute with each other.
The only non-trivial case to check is that the $B_{P,i}$ terms commute with each other.
This is due to the fact that given any two plaquettes $P$ and $P'$,
it is the case that either none of the $O$ and $U$ links of $P$ have any overlap with $\partial P'$ 
(which is equivalent to the condition with $P$ and $P'$ interchanged),
or it must be that the $O$ link of $P$ is in $\partial P'$ and the $U$ link of $P'$ is in $\partial P$,
or this is true with $P$ and $P'$ interchanged.
From this we see that any minus signs from anti-commutation of Pauli matrices in the two terms 
always come in pairs, and the terms commute. See \cref{fig:WW}.
More details can be found in~\cite{BCFV}.

The general construction of the Walker-Wang model
forbids any topologically nontrivial particle in the bulk,
whenever the input algebraic theory of anyons is modular.
Hence, we anticipate that the ground state should be unique under periodic boundary conditions,
since different ground states would be reached by
a particle-anti-particle pair that travels across the system.
Indeed, we can directly verify this for the Hamiltonian of \cref{eq:WWHam},
which we state as the following lemma.
\begin{lemma}\label{lem:nondegeneracy}
The Hamiltonian of \cref{eq:WWHam} has nondegenerate ground state
on which all $A_V, B_{P,1}$ and $B_{P,2}$ take eigenvalue $+1$
under the periodic boundary conditions ($3$-torus) of any system size.
\end{lemma}
\begin{proof} 
See page \pageref{pflem:nondegeneracy}.
\end{proof}

\subsubsection{Boundary and local topological order condition}

Given a Hamiltonian we may introduce a `boundary' by omitting terms
whose support lies outside the boundary.
If the Hamiltonian consists of pair-wise commuting terms,
this omission of terms produces a Hamitonian that is gapped,
but in an absurdly trivial manner.
A meaningful procedure or question is whether it is possible to define 
terms near the boundary in such a way that no more terms can be introduced
without being redundant.
We can formulate this intuitive requirement into a rigorous condition using
the concept of local topological order~\cite{BravyiHastingsMichalakis2010stability}.

A local Hamiltonian consisting of commuting projectors is said to exhibit
{\bf local topological quantum order (LTQO)}
at scale $L^\star$ if the reduced density matrix on a disk $D$ of radius $L^\star$
is unique for every state $\ket \psi$ that minimizes all Hamiltonian terms 
supported on the concentric disk $D^+ \supset D$ of radius $L^\star + O(1)$.
In other words, a Hamiltonian with LTQO determines the reduced density matrix uniquely
for any disk as long as the underlying state minimizes Hamiltonian terms around it.
Note that the LTQO condition does not ask about the boundary conditions of the system,
and hence it can be imposed or tested for systems with boundaries.

We define boundary terms for the three fermion Walker-Wang model in the following.
We will show that with our boundary terms the whole Hamiltonian obeys the LTQO condition.
Let us use the lattice as depicted in \cref{fig:WW}.
For clarity of presentation, let the vertices have integer coordinates, and the $z$-axis go upward.
The ``bulk'' occupies the region of coordinate $z < 0$,
and the ``boundary'' lies at $z=0$.
Any term that is supported in the bulk remains intact.
We define boundary vertex term as the product $\prod_{\ell \sim V} X_i^\ell$ ($i=1,2$)
of five $X$'s, instead of six, around a vertex $V$ with $z=0$,
omitting one factor on the link with $z > 0$.
The boundary plaquette terms are similarly defined:
Any term associated with a plaquette with a vertex of $z > 0$ is simply omitted.
Any term associated with a plaquette in the $zx$- or $zy$-plane are kept intact,
even if the plaquette has a vertex with $z=0$.
For a plaquette $P$ that lies in the $z=0$ plane,
we truncate the factor on the `over' link, similarly to the vertex terms:
$B_{P(z=0),1} = X_1^U X_2^U \prod_{\ell \in \partial P} Z_1^\ell$ that is a product of six Pauli matrices,
and 
$B_{P(z=0),2} = X_2^U \prod_{\ell \in \partial P} Z_2^\ell$ that is a product of five Pauli matrices.
The boundary term that are defined so far is depicted in
the top two rows of \cref{fig:WWsurface}.
Additionally, we define \emph{small fermion loop} terms, that live on the $z=0$ plane
and take eigenvalue $+1$ on any ground state, as in the third row of \cref{fig:WWsurface}.
The crucial reason we introduce the small fermion loop operators is the following.

\begin{lemma}\label{lem:enough-smallfermionloops}
For any operator $O$ supported on a bounded disk on the plane of $z=0$ (the surface),
if $O$ commutes with all truncated bulk terms in the top two rows of \cref{fig:WWsurface},
then $O$ is a $\CC$-linear combination of finite products of 
small fermion loop operators in the bottom of \cref{fig:WWsurface}.
Moreover, the small fermion loop operators that constitute $O$
can be chosen within an $O(1)$-neighborhood of the support of $O$.
\end{lemma}
\begin{proof}
Express $O$ as a $\CC$-linear combination of Pauli operators
since Pauli operators form an operator basis:
$O = \sum_j c_j P_j$ where $c_j \in \CC$.
Considering the assumption $O = T O T^\dagger$ for any term $T$ in the first two rows of \cref{fig:WWsurface},
the truncated bulk terms,
we see that nonzero $c_j$ is accompanied by a Pauli operator $P_j$ 
that commutes with every truncated bulk terms.
Thus, it remains to prove the lemma when $O$ is a Pauli operator of bounded support,
and overall phase factor of $O$ is immaterial.

Without loss of generality we may assume that $O$ is supported on a rectangle
as in the leftmost figure of \cref{fig:boundingRectangle}.
On the edge at the top left, the factor must commute with $Z_1,Z_2$ 
of the plaquette terms in the top right of \cref{fig:WWsurface},
and it also commutes with $X_1$ and $X_1 X_2$. Hence, the factor there is the identity.
The same argument (by the reflection symmetry of \cref{fig:WWsurface} about $x=y$ line)
shows that the factor on the horizontal edge on the top left is the identity;
see the second figure of \cref{fig:boundingRectangle}.
Inductively proceeding to the right, 
we see that the factor on the second vertical edge on the top left has to be one of $I, X_1, X_2, X_1 X_2$.
In every case, it is possible to eliminate it by multiplying the small fermion loop operators
within the bounding rectangle
--- 
since the small fermion loops are commuting with all the truncated bulk terms,
the claim that $O$ is a product of the small fermion loops is equivalent to
$O \prod (\text{small fermion loops})$ is a product of small fermion loops.
Moving to the right in a similar manner,
we come to a situation where the top row of the deformed rectangle
is now a single vertical edge;
see the third figure of \cref{fig:boundingRectangle}.
Then, the commutativity with ``L''-shaped operators forces the vertical edge to be the identity;
we have reduced the height of the bounding rectangle,
and can proceed similarly;
see the fourth figure of \cref{fig:boundingRectangle}.

\begin{figure*}
\includegraphics[width=\textwidth,trim={0cm 16cm 9cm 0cm},clip]{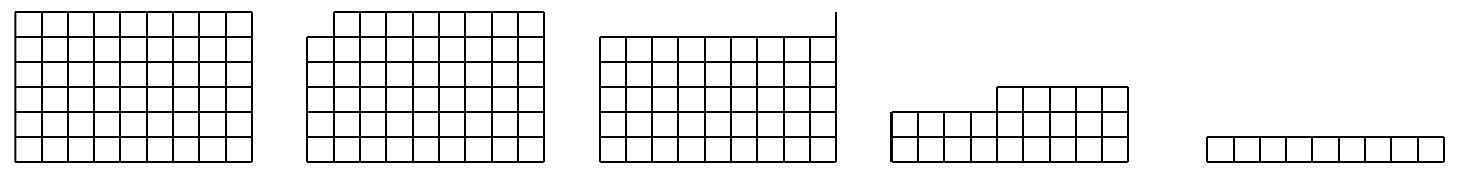}
\caption{
Bounding rectangle of an operator in the commutant of truncated boundary terms.
The rectangle can be deformed by small fermion loop operators so that it disappear eventually.
See the proof of \cref{lem:enough-smallfermionloops}.
}
\label{fig:boundingRectangle}
\end{figure*}

By induction, we may assume that the bounding rectangle has height $\le 1$;
see the last figure in \cref{fig:boundingRectangle}.
We can apply the commutativity that eliminated the top left corner in the beginning,
to conclude that the bounding rectangle of height at most 1 can only accommodate the identity operator.
This complete the proof.

See page~\pageref{pflem:enough-smallfermionloops} 
for an alternative proof that is perhaps more systematic.
\end{proof}

If one wants a finite system, 
one can define boundary terms on the bottom at, say $z = -L_z < 0$, similarly.
An easy way is to use the spatial inversion symmetry of our model
and impose periodic boundary conditions along $x$- and $y$-directions.
The symmetry is the inversion about any body center of the cubic lattice,
followed by the interchange of the qubits $1$ and $2$ within each link.
The inversion symmetry is not essential to define boundary terms,
and could be absent in some other model Hamiltonian of the same quantum phase.

\begin{lemma}\label{lem:ltqo}
The three fermion Walker-Wang model Hamiltonian on a system of linear size $L > 20$ 
with top and bottom boundaries open along $z$-direction
and periodic along $x$- and $y$-directions
that has the bulk and boundary terms as described above,
obeys the local topological order condition with $L^\star = L/2$.
\end{lemma}
The constant 20 is a sufficiently large but arbitrary constant,
which we do not optimize.
Similarly, the factor of 2 in $L^\star = L/2$ is unimportant.
\begin{proof}
For commuting Pauli Hamiltonians (Pauli stabilizer Hamiltonian),
it is shown~\cite[Lem.~2.1]{BravyiHastingsMichalakis2010stability}
that the LTQO condition is equivalent to the requirement
that any Pauli stabilizer $O$ for the ground state subspace 
that is supported on a disk of radius $ < L^\star$
must be a product of terms of the Hamiltonian supported on an $O(1)$-neighborhood of the disk.
We are going to prove this equivalent condition,
just assuming that $O$ is commuting with all Hamiltonian terms.

Before we present how to find a decomposition of $O$ in terms of the terms of the Hamiltonian,
let us make a geometric observation.
Given a Pauli operator that is a product of $Z_{1,2}$'s,
we imagine a link is `occupied' if there is a nonidentity $Z_{1,2}$ operator acting on the link.
We may refine the picture to say that the link is occupied by a type-1 or type-2 string segment,
depending on whether $Z_1$ acts or $Z_2$ does on the link.
Then, the commutativity with the vertex term demands 
that the occupied links must form a closed string of $Z$'s, separately for each string type.

First, we prove the claim when a Pauli operator $O$ is entirely within the bulk,
i.e., the support of $O$ is distance~5 away from the surface.
The operator $O$ can be written as $O_X O_Z$ up to an unimportant phase factor
where $O_X$ is a product of $X_1,X_2$'s and $O_Z$ is a product of $Z_1,Z_2$'s,
and we know $O$ commutes with all the bulk terms.
Since $O_Z$ has to commute with all the vertex terms,
it is a collection of closed strings.
But, since $O$ is within a ball in the bulk,
and any closed string can be expressed by a product of plaquettes.
Since the $B_{P,i}$ terms have these plaquettes in them,
we can remove the strings at the expense of introducing extra $X_i$ factors
from `over' and `under' links of $B_{P,i}$.
That is, we find $O_X'$ in the vicinity of $O_Z$ such that $O_Z O_X' = \prod B_{P,1} \prod B_{P,2}$
up to an unimportant phase factor.
Therefore, it suffices to prove the claim when $O_Z$ is the identity.
In such a case, the commutativity of $O_X$ with the plaquette terms $B_{P,i}$ is equivalent
to the commutativity of $O_X$ with small loop operators, the product of four $Z_i$'s around a plaquette.

Now the support of $O_X$ cannot have any link that meets a plaquette alone;
it must accompany another perpendicular adjacent link,
and such a pair of links can be removed by multiplying a vertex term.
Formally, we can consider a bounding box of $O_X$,
impose the commutativity condition,
and deform the bounding box by multiplying a vertex term,
eventually to eliminate the bounding box.
(The argument is identical to the 3D toric code's case,
and is a variant of the proof of \cref{lem:enough-smallfermionloops}.)
It is now clear that $O_X$ must be a product of vertex terms.

Second, we prove the claim when a Pauli operator $O$ has a factor near the surface.
The geometric interpretation of $Z$'s as a string segments is valid,
even with the truncated vertex terms at the boundary.
Hence, the reduction in the previous paragraph 
from a general $O$ to $O = O_X$ is still valid.
Now, if $O = O_X$ acts on a link with $z < 0$ by nonidentity,
the argument in the preceding paragraph
allows us to ``push'' $O_X$ to the surface.
Then, we can apply \cref{lem:enough-smallfermionloops}
to conclude that $O$ is indeed a product of the terms of the Hamiltonian
that are near the support of $O$.
\end{proof}

\begin{remark}\label{rem:ltqo}
We can repeat the proof of \cref{lem:ltqo} for a semi-infinite system 
on the half space $z \le 0$ with one boundary at $z = 0$.
The conclusion is that 
if any operator of finite support commutes with all terms of the Hamiltonian $H_{w/bd}$,
then it is a $\CC$-linear combination of products of terms in the Hamiltonian
in the $O(1)$-neighborhood of the operator.
In particular, we remark the following observation for later use.
Consider a two-dimensional slab at the boundary 
(i.e., $-L_z \le z \le 0$ form some $L_z \ge 0$),
and let $\stab$ be the multiplicative group generated by all terms of $H_{w/bd}$ (the stabilizer group).
We define two more associated groups.
Every $g \in S$ is a product $g = g_{in} g_{out}$ 
of two factors $g_{in}$ that is supported inside the slab, and $g_{out}$ that is outside the slab.
This decomposition is ambiguous on scalar phase factors, but that is the only ambiguity,
and we ignore such scalar phase factors.
Define $\stab(L_z) = \{ g \in \stab ~:~ g_{out} = I\}$
to be the subgroup of all $g \in \stab$ such that $g$ is supported on the slab.
Also, define $\stab|_{L_z} = \{ g_{in} ~:~ g \in \stab \}$
to be the multiplicative group of all truncated Pauli operators.
The group $\stab|_{L_z}$ is nonabelian, but includes the abelian group $\stab(L_z)$, and is supported on the slab.
Now, if a Pauli operator $g'$ supported on the slab commutes with every element of $\stab|_{L_z}$,
then it commutes with every element of $\stab$.
Then, the LTQO proof implies that $g'$ belongs to $\stab(L_z)$.
In other words, 
\emph{
$\stab(L_z)$ is the commutant of $\stab|_{L_z}$ within the group of all Pauli operators on the slab.
}
\end{remark}

\subsection{Separators and Clifford QCA}
\label{sec:modH}

In this subsection,
we will find a locally flippable separator whose common eigenstate of eigenvalue $+1$ 
is equal to the ground state of the Walker-Wang model \emph{without} boundary.
The separator that we find \emph{lacks} any vertex term,
but consists of the original plaquette terms $B_{P,1}$ (the right three in the first row of \cref{fig:WW}),
and modified plaquette terms $B'_{P,2}$.
This implies that the state admits a different set of stabilizer generators,
which can be used to define another Hamiltonian that share the same ground state as the original Hamiltonian.

The modified plaquette terms are not simple to draw on a piece of paper,
and is not instructive to present such a drawing;
we have to present it in a compact notation to be explicit.
The natural choice is to use the polynomial framework~\cite{Haah2013}
which we will elaborate on the next section.
In terms of polynomials, $B'_{P,2}$ is given by
\begin{align}
B'_{P,2} = 
\left(
\begin{array}{c|c|c}
 \frac{1}{x^2 y z}+\frac{1}{x y z} & \frac{z}{x}+z+\frac{1}{x y}+\frac{1}{x^2 y} & \frac{y z}{x}+\frac{1}{x^2}+\frac{1}{x^2 y z} \\
 \frac{1}{x y^2 z}+\frac{1}{x y z} & x z+\frac{z}{y}+z+\frac{1}{y}+\frac{1}{x y}+\frac{1}{x y^2}+\frac{1}{x y^2 z} & y z+z+\frac{1}{x}+\frac{1}{x y} \\
 x y+\frac{1}{z}+\frac{1}{x y z} & z+\frac{1}{x y}+\frac{1}{x y z}+1 & z y+y+\frac{1}{x}+\frac{1}{x z} \\
 \frac{1}{x^2 y z}+\frac{1}{x y z} & \frac{z}{x}+z+\frac{1}{x y}+\frac{1}{x^2 y} & \frac{y z}{x}+\frac{1}{x^2}+\frac{1}{x^2 y z} \\
 \frac{1}{x y^2 z}+\frac{1}{x y z} & x z+\frac{z}{y}+z+\frac{1}{y}+\frac{1}{x y}+\frac{1}{x y^2}+\frac{1}{x y^2 z} & y z+z+\frac{1}{x}+\frac{1}{x y} \\
 x y+\frac{1}{z}+\frac{1}{x y z} & z+\frac{1}{x y}+\frac{1}{x y z}+1 & z y+y+\frac{1}{x}+\frac{1}{x z} \\
\hline
 \frac{1}{x y z}+\frac{1}{x z} & \frac{1}{x y z}+\frac{1}{x y} & 0 \\
 \frac{1}{x y z}+\frac{1}{y z} & 0 & \frac{1}{x y z}+\frac{1}{x y} \\
 0 & \frac{1}{x y z}+\frac{1}{y z} & \frac{1}{x y z}+\frac{1}{x z} \\
 y+1 & z+1 & 0 \\
 x+1 & 0 & z+1 \\
 0 & x+1 & y+1 \\
\end{array}
\right).
\end{align}
Briefly, each column represents one operator up to translations in the lattice.
A $j$-th row in the upper half block represents factors of $X$ in an operator at $j$-th qubit within a unit cell.
The ordering of the qubits within a unit cell is
``1''-qubits on the edge along $x$-, $y$-, $z$-axes, and then ``2''-qubits.
A monomial $x^a y^b z^c$ with $a,b,c \in \ZZ$ in the upper half block
represents a nonidentity factor of $X$ at the unit cell of coordinate $(a,b,c)$,
and that in the lower half block represents $Z$, respectively.
Any $Y = i X Z$ factor is represented by repeating a monomial on the both upper and lower half blocks;
e.g., the operator represented by the third column has a factor of $Y_2$ at the edge along $z$-axis.
For further details, see \cref{sec:pauli}.

We claim that this is a locally flippable separator and
that the group generated by our separator is the same as the group generated by terms of our Hamiltonian.
Importantly, there exists a Clifford QCA%
\footnote{
A Clifford QCA is by definition a QCA that maps a Pauli operator to a Pauli operator.
In one dimension, any translation invariant QCA on the system of \emph{one} qubit per site
is known to be a Clifford circuit up to a shift~\cite{clifQCA}.
In fact, it is not difficult to generalize the result of \cite{clifQCA}
to more general one-dimensional translation invariant Clifford QCAs
on a system of $q \ge 1$ qudits of prime dimension $p$ per site,
following the calculation in \cite[Sec.~6]{Haah2013}.
The result is that
\emph{any such one-dimensional Clifford QCA is a Clifford circuit up to a shift.}
}
that disentangles the ground state of the three fermion Walker-Wang model.
The disentangling Clifford QCA is too complicated to write in this page, 
even with the compact polynomial notation.
We present calculation in the accompanying computer algebra script
to establish these claims.
A noteworthy feature of our QCA is that it maps any complex-conjugation-invariant operator 
to a complex-conjugation-invariant operator.
That is, our QCA does not break time-reversal symmetry,
but manifestly disentangles the ground state of the Walker-Wang model.
A reader might want to recall that
this Walker-Wang model with three fermion theory at the boundary
is proposed to represent a symmetry-protected topological phase under the time-reversal symmetry~\cite{BCFV}.

We denote by $\alpha_{WW}$ this disentangling QCA.

Since the separator generates the same group as the terms of $H_{WW}$ in \cref{eq:WWHam},
we may say that the separator defines another local Hamiltonian 
that represents the same quantum phase as $H_{WW}$.
A feature of this new Hamiltonian is that the terms are nonredundant.
\begin{definition}\label{def:nonredundant}
A set of operators on a finite system is {\bf nonredundant} 
if any nonempty product is never proportional to the identity operator.
A translation invariant set of operators on an infinite system is {\bf locally nonredundant}
if it is nonredundant.
A locally nonredundant set may become redundant in a periodic finite system.
\end{definition}
One may then wonder if we could choose a nonredundant boundary terms as well.
We give a partially affirmative answer:
\begin{lemma}\label{lem:nonredudantBd}
Let $r \ge 1$ be a sufficiently large number
such that every term of $H_{WW}$ can be written as a product of the elements of the separator 
within a ball of diameter $r$ around the term.
Then, for any $ t \ge r$,
there exists a locally nonredundant commuting Pauli Hamiltonian $H(t)$ 
(translation invariant in $x,y$-directions)
supported on the slab $(-t \le z \le 0)$ of the semi-infinite system ($z \le 0$)
such that the union of the terms of $H(t)$ and the elements of the separator below the plane of $z = - t + r$
generate the same multiplicative group as the $H_{w/bd}$ of \cref{rem:ltqo}.
\end{lemma}
Note that this does not imply that the boundary terms form a locally flippable separator
on the boundary.
Also, even though $H(t)$ and the separator are nonredundant on their own, 
there may be redundancy in the overlapping region where $-t \le z \le -t + r$.
\begin{proof}
The claim will be proved by combining \cref{rem:ltqo} and \cref{lem:removingRelations2D}.
Since $S(t)$ is a commutant of $S|_t$ that is a translation-invariant group
on a 2-dimensional lattice of qubits, 
\cref{lem:removingRelations2D} implies that $S(t)$ has a locally nonredundant translation-invariant
generating set. We define $H(t)$ to be the negative sum of all these generators.
Our boundary terms are contained in the slab of thickness 1.
Hence, the subgroup $S(t)$ for $t \ge 1$ as defined in \cref{rem:ltqo}
contains all of our boundary terms.
Also, by the choice of $r$ any term of $H_{WW}$ either belongs to $S(t)$ or
is a product of elements of the separator below the plane with $z = -t +r$.
Therefore, the terms of $H(t)$ and the chosen separator generate the full stabilizer group.
\end{proof}

Given the complicated separator above,
a reader might be curious how we find the separator for this model.
In fact, in \cref{sec:pauli} we develop a more general theory 
in which we give a constructive proof of the existence of a separator
for translation-invariant commuting Pauli Hamiltonians (Pauli stabilizer Hamiltonians).
The separator above is an example of our constructive proof
applied to the three fermion Walker-Wang model.
Within the class of translation-invariant commuting Pauli Hamiltonians,
the sole assumption that guarantees a separator 
is the nondegeneracy of the ground state. See \cref{thm:ticp-disentangler}.
This is a very mild condition, and is certainly a necessary one;
however, a Hamiltonian being composed of commuting Pauli operators is special.

\subsection{Three-fermion theory by a hypothetical commuting Hamiltonian in 2D}

Our next question is whether the found QCA can be a quantum circuit.
We conjecture that the answer is no, even if we allow translations in the lattice and ancillary qudits.
In the next two subsections we present evidences in favor of our conjecture.
Before we begin elaboration, here we note an important fact:
\emph{
If the disentangling QCA $\alpha_{WW}$ were a circuit,
then the three fermion theory can be realized in a two-dimensional lattice
of finite dimensional qudit degrees of freedom with a commuting projector Hamiltonian.
More generally, the same claim holds if the QCA $\alpha_{WW}\otimes \Id$ were a circuit,
where $\Id$ is the identity QCA acting on additional qudit degrees of freedom.}
We will call this hypothetical commuting projector Hamiltonian $H_{3F}$.

To show this fact, of course, we must define what we mean by ``realizing'' the three fermion theory.  
One might imagine various possible ways of ``realizing'' the three fermion theory, 
by constructing different microscopic models whose long distance behavior in some way has the correct anyons,
but it might be quite a difficult task to even define, for an arbitrary microscopic theory, 
what quantum phase it realizes at long distances.  
In the present case, however, we will be able to demonstrate that 
there are only three types of nontrivial topological charges in $H_{3F}$
and they have mutual and self statistics that are identical to those of 
the unitary modular tensor category of three fermions.
In this sense, we say $H_{3F}$ realizes the three fermion theory.

First, we must explain how to construct $H_{3F}$.
Roughly speaking, we consider the Walker-Wang model with boundaries 
and we use the assumption that $\alpha_{WW}$ is a quantum circuit 
to construct another QCA $\beta$ which acts like $\alpha_{WW}$ deep in the bulk 
while acting like the identity near the top boundary.  
Acting with this QCA $\beta$ on the model with boundaries, 
the terms near the top boundary are left unchanged, 
while the bulk becomes disentangled, 
allowing us to consider the theory on the top boundary alone, separate from that of the bottom boundary.
One technical detail is that the $\alpha_{WW}$ is defined using the modified terms (the separator) of \cref{sec:modH},
rather than using the original terms of the Walker-Wang model; 
to embrace this technicality, 
we modify the Walker-Wang model with boundaries to define a Hamiltonian $\tilde H$ 
so that near the boundaries we use the terms of the original Walker-Wang model, 
slightly further from the boundaries we use both original and modified terms, 
and in the bulk we just use the modified terms.  

Let us now describe the procedure in detail.
\begin{lemma}\label{lem:betabulk}
Assume that $\alpha_{WW}\otimes \Id$ is a quantum circuit of depth $O(1)$;
here the factor $\Id$ represents the identity QCA on any added degrees of freedom.
Then there exists $m_0 = O(1)$ such that for any $m \ge m_0$
there exists a QCA $\beta$ with the following properties.

For any operator $O$ supported within distance $m$ of the plane $z=0$ we have $\beta(O)=O$,
while for any operator $O$ supported far from the plane $z=0$ by distance $>2m$ 
we have $\beta(O)=\alpha(O)$.
\end{lemma}

\begin{proof}
Write $\alpha_{WW} \otimes \Id $ as a quantum circuit
$\alpha=\alpha_d \circ \alpha_{d-1} \circ \cdots \circ \alpha_{1},$
where
each $\alpha_a$ is some QCA that can be written in the form
$\alpha_a(O) = \prod_{S \in G_a} U_{S,a}^\dagger O U_{S,a}$.
By assumption, $d = O(1)$.
Define
$\beta=\beta_d \circ \beta_{d-1} \circ \cdots \circ \beta_{1}$,
where
each $\beta_a$ is some QCA that can be written in the form
$\beta_a(O) = \prod_{S \in \tilde G_a} U_{S,a}^\dagger O U_{S,a}$, 
where $\tilde G_a$ includes only the set of
gates whose support is far from the $z=0$ plane by distance $>m$.

By construction, then, $\beta(O) = O$ 
whenever $O$ is supported within distance $m$ of the plane $z=0$ 
since $\beta$ acts by identity there.
If $O$ is supported on the set of links of distance more than $2m$ from the plane $z=0$,
the gates that determine $\alpha(O)$ are supported 
on links that are far from the $z=0$ plane by distance $2m - O(1)$.
For any sufficiently large $m$, we have $2m - O(1) > m$,
and $\beta$ has all the necessary gates to have $\beta(O) = \alpha(O)$.
\end{proof}

Our construction of the Hamiltonian $H_{3F}$ will use only the properties of $\beta$ in \cref{lem:betabulk}.  So even if $\alpha_{WW}$ is not a circuit but such a $\beta$ can still be constructed, then the properties of $H_{3F}$ will still follow.

To construct $H_{3F}$ by $\beta$ of \cref{lem:betabulk},
we consider a system periodic in the $x,y$ coordinates, but with $z$ coordinate ranging from $0$ to $ - 10 L_z < 0$,
where $L_z$ is larger than 
any diameter of elements of the separator, the ranges of $\beta$ and $\beta^{-1}$,
and the constant $m_0$ of \cref{lem:betabulk}.
% Consider the Walker-Wang model with boundary on a system of linear size $L$ with given $L_z$.
We define an intermediate Hamiltonian $\tilde H$ as follows. 
$\tilde H$ includes all terms in the Walker-Wang model with boundary 
which are supported \emph{within} distance $2L_z$ from the top boundary $z=0$.
$\tilde H$ also includes all modified terms of \cref{sec:modH} which are
supported \emph{below} the plane $z = -L_z$.
Here, the ``modified terms'' are equal to $-1$ times the elements of the separator, 
so that the $+1$ eigenstate of the elements of the separator minimizes the energy.
Note that in the middle slab specified by $-2L_z \le z \le -L_z$,
there coexist the original terms of $H_{WW}$ and the elements of the separator.
Finally, if we have added additional degrees of freedom (stabilization),
then $\tilde H$ include a term that fixes the ground state of each such qudit
--- we choose such a term to be the generalized Pauli operator $-Z$ on that qudit.
If the reader objects to a term such as $-Z$ for qudits since that term is non-Hermitian in general, one can choose instead any
term that is diagonal in the $Z$ basis for that degree of freedom, 
with all eigenvalues distinct and such that the ground state has $Z=1$ for that degree of freedom.

We have not specified what to do near the bottom boundary $z \le -9L_z$; an
arbitrary choice can be made there, as we will drop that region anyway.

We then construct the QCA $\beta$ of \cref{lem:betabulk} with $2m = L_z$ 
and consider the Hamiltonian $\beta(\tilde H)$.
(\cref{lem:betabulk} does not specify $\beta$ near the bottom boundary $z \le -9L_z$,
but one can just drop the gates of $\alpha_{WW} \otimes \Id$ near the bottom boundary.)
This Hamiltonian $\beta(\tilde H)$ is a sum of local commuting terms.
Since $\beta(O)=\alpha(O)$ for $O$ supported far from the plane $z=0$ by distance $> 2m = L_z$,
each term of $\tilde H$ below the plane $z = -L_z$ that is an element of the separator of \cref{sec:modH},
is mapped to (up to a sign) a single-qubit Pauli $Z$ operator.
Every generalized Pauli $Z$ term of $\tilde H$ below the $(z=-L_z)$-plane, 
that fixes the ground state of additional qudits, remains unchanged under $\beta$,
since $\alpha \otimes \Id$ acts by the identity there.
So, restricting to the $-1$ eigenspace of these terms in Hamiltonian $\beta(H)$ 
(i.e., the $+1$ eigenspace of the corresponding stabilizer), 
we can remove the corresponding qubit and qudit, 
replacing all occurrences of Pauli $Z$ operators on that qudit (qubit) in $\beta(\tilde H)$ 
with a scalar $\pm 1$.

Furthermore, every qubit or qudit with $z$ coordinate being $-8 L_z \le z \le -2L_z$,
has a term which maps to a Pauli $Z$ operator on that qubit or qudit.
This holds for the original qubit degrees of freedom 
since $\alpha_{WW}$ disentangles the separator for the Walker-Wang model.
For any added qudit, there is a $Z$ term that remains intact under $\beta$ in this region.

%Let $I_{mid}$ be this slab of $-8L_z\le z \le -2L_z$.
%We can remove all degrees of freedom in $I_{mid}$.
%For $L_z$ large enough, there are no terms in $\beta(\tilde H)$ 
%whose support includes sites with $z$ coordinates in the interval $[0,-L_z/5-O(1)]$ and the interval $[-4L_z/5+O(1),-L_z]$.
%We will say that a term in $\tilde H$ is supported ``above" $I_{mid}$ if it is supported on the set of links with $z$ coordinates greater than those in $I_{mid}$ and we will say that it is supported ``below" $I_{mid}$ if it is supported on the set of links with $z$ coordinates more negative than those in $I_{mid}$.
\begin{definition}
\label{H3fdef}
We define $H_{3F}$ as
the sum of terms of $\beta( \tilde H)$ that are supported within distance $3L_z$ of the top boundary $z=0$.  The system
for $H_{3F}$ includes only the sites supported within distance $3L_z$ of the top boundary $z=0$.
%
%Pick $L_z=O(1)$ sufficiently large that the conditions of lemma \ref{betabulk} hold and such that the LTQO condition is obeyed for the Walker-Wang model with boundaries.
%We will define $H_{3F}$ by including only the links above $I_{mid}$ in the system and including the terms in $\beta(\tilde H)$ which are supported above $I_{mid}$.
\end{definition}

Hamiltonian $H_{3F}$ can be regarded as a two-dimensional lattice Hamiltonian by ignoring the $z$ coordinates, regarding all sites in the three-dimensional lattice with the same $x,y$ coordinate but different $z$ coordinates as corresponding to the same site in some two-dimensional lattice.  Since the $z$ coordinate
takes only $O(1)$ possible values, if we have stabilized by adding only $O(1)$ qudits on each site, then this
two-dimensional system has only $O(1)$ qudits on each site.

\begin{remark}
This definition of $H_{3F}$ may have some redundancies among the terms, 
i.e., some products of terms may be proportional to the identity operator.
However, if we instead use the terms of \cref{lem:nonredudantBd} to define $H_{3F}$, 
we arrive at a locally nonredundant Hamiltonian 
whose terms generate the same multiplicative group as the present $H_{3F}$.
%which also realizes the three fermion phase (the group generated by the terms is the same as in $H_{3F}$) but whose terms are non-redundant.
Indeed, \cref{lem:nonredudantBd} gives an alternative Hamiltonian $H(3L_z)$ near the top boundary 
that is locally nonredundant.
If we use the terms of $H(2L_z)$ in place of the original terms of $H_{WW}$ near the top boundary,
then $\beta(H(3L_z))$ has locally nonredundant terms. Additional terms to $H_{3F}$
are from the separator, but they are all single-qubit or single-qudit $Z$'s after $\beta$.
So, even though these additional terms may give some redundancy,
they are easily removed by setting these single-qubit $Z$ factors in the terms of $\beta(H(3L_z))$ by scalars.
We have chosen to define $H_{3F}$ without reference to \cref{lem:nonredudantBd} for more explicit presentation.
We emphasize that for any finite system, the Hamiltonian will not be nonredundant, since, in particular, the product of all translates of any given small fermion
loop will equal the identity. 
\end{remark}

\begin{remark}
The importance of the multiplicative abelian group $S$ generated by the terms of 
our commuting Hamiltonian is due to the following observation:
Let $A$ be a set of local generators for $S$ and $B$ be another such set.
Then, an obvious Hamiltonian path
\begin{align}
H_\zeta = -(1-\zeta) \sum_{a \in A} a - \zeta \sum_{b \in B} b
\end{align}
is gapped with an invariant ground state subspace throughout the path where $\zeta$ ranges 
from~$0$ to~$1$.
\end{remark}

This motivates the following definition of topological charges for a Hamiltonian on an infinite lattice.
\begin{definition}[\cite{Haah2013}]
An excitation --- a finite set of flipped terms --- of a commuting Pauli Hamiltonian (Pauli stabilizer Hamiltonian)
is {\bf physical} if it can be created from a ground state by a Pauli operator of possibly infinite support.
Two Pauli operators of possibly infinite support can be multiplied,
and this induces abelian group structure on the set of all physical excitations.
A physical excitation has {\bf trivial charge} if it is created by a Pauli operator of finite support.
A {\bf (topological) charge} is an equivalence class of physical excitations modulo trivial charges.
\end{definition}
The definition is narrow, since it cannot handle nonabelian anyons,
but will be sufficient for our purposes.
For a family of finite systems of increasing system size,
we can define a topological charge as follows.
\begin{definition}\label{def:charge}
An excitation $e$ --- a set of $O(1)$ flipped terms --- of a commuting Pauli Hamiltonian
is {\bf physical} if $e$ and possibly some other excitation a distance $\omega(1)$ from $e$
can be created from a ground state by a Pauli operator.
A physical excitation has {\bf trivial charge} if it can be created without creating any other excitation.
A {\bf (topological) charge} is a equivalence class of physical excitations modulo trivial charges.
\end{definition}

We now investigate the topological charges of $H_{3F}$.
Strictly speaking, $H_{3F}$ does \emph{not} consist of Pauli operators,
since $\beta$ does not transform every Pauli operator to a Pauli operator.
However, the definition of topological charges will carry over to $H_{3F}$ via $\beta$;
if desired, we can always pull an operator on the system of $H_{3F}$ 
back to the original Walker-Wang system, examine properties, and push it forward.
The fact that $\beta$ does not preserve the ``cut'' at the plane $z = - 3L_z$
will not affect the analysis of charges;
any Pauli operator $P$ that inserts a topological charge near the top surface
in the original Walker-Wang system with boundary,
can be chosen such that $\beta(P)$ is supported within distance $2L_z$ from the top surface,
because every element of the locally flippable separator represents a trivial charge,
and thus can be ignored from a representative excitation for the charge.

Our goal is to show that any topological charge of $H_{3F}$ 
is inserted to the system by an open fermion string operator.
Of course, we have to define what open fermion string operators are.
\begin{figure}
\centering
\includegraphics[width=0.7\textwidth,trim={0cm 14cm 16cm 0cm},clip]{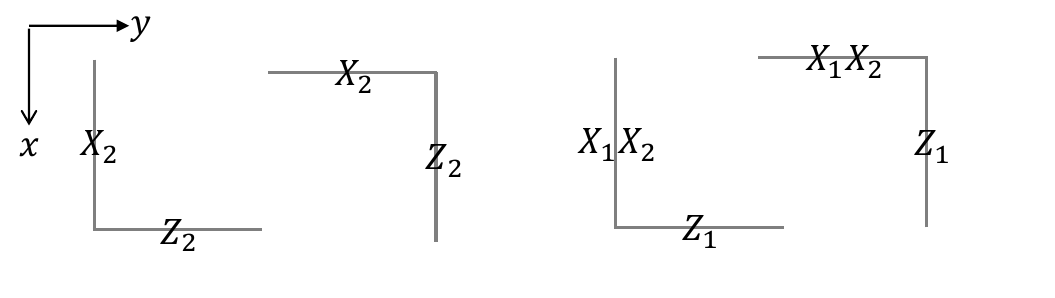}
\caption{Segments of fermion string operators.
We call any product of these as fermion string operators.
Depending on whether the links acted by $Z_j$ are a closed chain or not,
we call a fermion string operator closed or open.
}
\label{fig:open}
\end{figure}
\begin{definition}
A {\bf fermion string operator} is an arbitrary product of operators in \cref{fig:open} on $z=0$ plane.
We associate a {\bf chain} (one-dimensional chain as in cellular homology) of type-$j$ ($j=1,2$)
to a fermion string operator as the set of links on which $Z_j$'s act.
The set of {\bf end points} of a fermion string operator is the set of vertices 
at which there are an odd number of incident links of the associated chain.
A fermion string operator is {\bf closed} if it is free of end points,
and {\bf open} otherwise.
\end{definition}

Our small fermion loop operators in the bottom of \cref{fig:WWsurface}
are the smallest closed fermion string operators.
In fact, any closed fermion string operator whose chain is of null homology
is a product of the small fermion loop operators.
Note that a closed fermion string operator that is supported along the circumference of a disk,
is a product of all small fermion loop operators within the disk.

\begin{lemma}\label{lem:chargesOfH3F}
Let $B$ be a set of points.
Suppose that an operator $O$ commutes with all terms in $H_{3F}$ 
except those terms whose support overlaps with $B$; here we regard $H_{3F}$ as a two-dimensional lattice Hamiltonian so that the points in $B$ are specified by $x,y$ coordinates.
Then, $O$ can be written as
a $\CC$-linear combination of operators, each of which is a product of form $h_{3F} O_f O_B$,
where $h_{3F}$ is a product of terms of $H_{3F}$,
$O_f$ is a fermion string operator whose end points are within distance $O(1)$ from $B$,
and $O_B$ is an operator supported within distance $O(1)$ from $B$.
If $B$ is empty, then $O_f$ must be a closed fermion string operator and $O_B = I$.
\end{lemma}
\begin{proof}
Although $O$ is an operator on the system of $H_{3F}$ (which excludes the deep bulk),
we identify $O$ with $O \otimes I$ that is an operator on the system of $\beta(\tilde H)$ 
(which includes the deep bulk).
The extended operator $O \otimes I$ obviously commutes with every term deep in the bulk.
Write $\beta^{-1}(O \otimes I) = \sum_k c_k P_k$ as a $\CC$-linear combination of Pauli operators $P_k$.
Since every term of $\tilde H$ is a Pauli operator,
it suffices to prove the lemma in the case where the sum $\sum_k c_k P_k$ had only one summand,
i.e., $\beta^{-1}(O \otimes I) = P$ is a Pauli operator,
which we assume for the rest of the proof.
The supposition of the lemma implies that $P$ commutes with all terms of $\tilde H$ except those near $B$, where ``near $B$" means that
the term is supported on some site whose projection to the $z=0$ plane by $(x,y,z) \rightarrow (x,y,0)$ is within distance $O(1)$ of $B$.

If $P$ acts on the qudits that are inserted by stabilization,
then we may immediately forget about it, by modifying $P$ by single-qudit Pauli operators.
Also, since $H_{3F}$ has single-qubit $Z$ terms near the plane $z = -3L_z$,
we may assume that $\beta(P)$ is supported within distance $2L_z$ from the top boundary.
(Any excitation in the region $-3L_z \le z \le -2 L_z$ can be eliminated by a single-qubit $X$'s near $B$,
and any remaining factor of $Z$ can be eliminated by a term of $H_{3F}$ that is a single-qubit~$Z$.)

The important fact is that we can ``push'' $P$ to the top boundary, for the part that is far from $B$.
Recall that a local product of terms of the original Hamiltonian $H_{WW}$ is an element of the separator,
and conversely the elements of the separator locally generate every term of $H_{WW}$.
Thus, the geometric interpretation in the proof of \cref{lem:ltqo} ---
that $Z$'s form strings if it commutes with the vertex terms 
and that $X$'s form a dual cube if factors of $Z$ are absent --- is applicable
in the region where $P$ commutes with terms of $\tilde H$.
As in the proof of \cref{lem:ltqo}, we first push the factors of $Z$'s to the top boundary
from the region where $P$ does not create excitation,
by multiplying plaquette operators of $H_{WW}$,
which in turn is equivalent to multiplying terms of $\tilde H$.
Next, similarly, we can push $X$-factors of $P$ to the top boundary by multiplying terms of $\tilde H$.
Note that on the operator $\beta(P) = O$, 
this amounts to pushing $O$ to the top boundary by the terms of $H_{3F}$ only,
without single-qubit $Z$'s in the deep bulk ($z \le -3L_z$)
--- this is a reason we have defined $H_{3F}$ to include the region $-3L_z \le z  \le -2L_z$
where only single-qubit $Z$'s terms are present.

Now, we are left with a Pauli operator whose factors $F$
that act far from $B$ are strictly on the top surface.
We claim $F$ must form fermion strings with the end points near $B$.
This is by inspection of \cref{fig:WWsurface}:
The factor $F$ must commute with every surface terms.
In particular, the loop and dual loop operators on the first row of \cref{fig:WWsurface}
forces $F$ to be some strings of $Z$ and dual strings of $X$.
The operators on the second row of \cref{fig:WWsurface} 
dictates that these strings must align to become fermion string operators.
Since $\beta$ acts like the identity at the top surface,
we conclude the proof.
\end{proof}

By \cref{lem:chargesOfH3F} we conclude that there are only three topologically nontrivial charges:
The end point of a fermion string operator of type~1, or that of type~2, 
or the end point of the product of the two.
Since an end point depends on the parity of the incident links of 
the associated chain of a fermion string operator,
we may say they are $\ZZ_2$ charges.
They deserve the name ``fermion'' because of the following definition of self-statistics.
\begin{definition}[\cite{LevinWen2003Fermions}]\label{def:topologicalSpin}
Suppose a topological charge $a$ under \cref{def:charge}
is created at the end point of a string operator 
--- an operator supported on an $O(1)$-neighborhood 
of a path of length $\omega(1)$ without self-intersection.
We define the {\bf topological spin} $\theta_a$ of $a$ 
as the phase factor in the commutation relation
\begin{align}
t_{1} t_{2}^\dagger t_{3} = \theta_a t_{3} t_{2}^\dagger t_{1}.
\end{align}
where $t_j$ ($j=1,2,3$) 
is a string operator that moves a charge from a point $p_j$ to a common point $p_0$
where the three string operators are arranged counterclockwise around $p_0$.
\end{definition}
The definition is applicable only for charges with string operators;
e.g., we cannot apply it for charges in the cubic code model of~\cite{Haah_2011}.
Note also that we are relying on the fact that $t_j$'s obey commutation relation up to a phase factor,
which may not at all be true for general Levin-Wen models~\cite{Levin_2005},
but is true at least if the string operators are (generalized) Pauli operators.
One can easily see that the precise choices of hopping operators $t_j$ are not important;
a string operator may be modified by multiplying terms of the Hamitonian,
but as long as the modifying terms are far from points $p_1,P_2,p_3$
the commutation relation among the string operators remains unchanged.

The standard terminology is that $a$ is fermion if $\theta_a = -1$, or boson if $\theta_a = +1$.
The trivial topological charge ``$1$'' always has $\theta_1 = 1$,
but we generally refer to the trivial topological charge as the vacuum, rather than as a boson.
We leave it to the reader to compute the topological spins of the charges of $H_{3F}$
from the fermion string operators, to show that there are only fermions other than the vacuum;
see~\cite{BCFV}.

\subsection{The disentangling QCA is not a Clifford Circuit}
\label{ntcliff}

In this subsection, we prove the following:
\begin{theorem}
The disentangling Clifford QCA $\alpha_{WW}$ of \cref{sec:modH}
is not a product of a Clifford circuit (of $O(1)$ depth) and a shift.
\end{theorem}
\begin{proof}
The conclusion from the previous subsection is that $H_{3F}$ contains only fermions other than vacuum.
However, any translation-invariant two-dimensional commuting Pauli Hamiltonian 
that satisfies the local topological order condition has always a boson by \cref{cor:2dBoson} below.
Since a Clifford circuit maps a Pauli operator to a Pauli operator,
these two facts are contradictory \emph{if} $\alpha_{WW}$ consisted of Clifford gates.

Here, the hypothetical Clifford circuit may not be translation-invariant.
However, if we consider a family of Clifford circuits on finite systems with periodic boundary conditions,
then we can pick an instance in the family, and using the periodic boundary conditions
promote the instance to define another family that has coarser translation symmetry.
\end{proof}

The statement we will prove is suggested in~\cite{Bombin2011Structure},
where Bombin has demonstrated how to extract copies of toric codes
from a translation-invariant commuting Pauli Hamiltonians (Pauli stabilizer Hamiltonians) 
on two-dimensional lattice of qubits.
An obstruction was identified in such extraction, and it was termed ``chirality.''
This chirality is revealed only after one examines all topological charges of the model,
and Ref.~\cite{Bombin2011Structure} suggests that it is the same notion as
in Appendix~D of~\cite{Kitaev_2005}.
The trouble is that the two notions have different definitions,
between which rigorous connection has not been established.

Our new ingredient here is, roughly speaking, to show 
that the one-dimensional boundary of a two-dimensional commuting Pauli Hamiltonian can be gapped out.
We state our result only for qubit systems,
but a parallel argument will prove an analogue for systems of prime $p$-dimensional qudits.
We are unaware of any previous result that constructs gappable boundaries generally,
though the boundaries for the toric code is understood very well.

Then, we combine this result with an idea of Levin~\cite{Levin2013}
that a gapped boundary implies that there is a boson in the bulk.
We continue to use \cref{def:topologicalSpin} whenever we refer to a fermion, a boson, or topological spin.

To avoid lengthy phrases,
we adopt the terminology and convention of~\cite{Haah2013}:
An {\bf exact code Hamiltonian} is a locally topologically ordered
frustration-free translation-invariant local Hamiltonian 
on an {\em infinite} system of qubits without boundary that consists of commuting products of Pauli matrices.
If a finite system or a boundary is needed, we will explicitly mention it.

\begin{lemma}[Thm.~4 in Sec.~7 of~\cite{Haah2013}]\label{lem:2dbulk}
For any exact code Hamiltonian in two dimensions,
there exists another exact code Hamiltonian $H$ that is locally nonredundant
and whose terms generate the same multiplicative group.
Any term $h$ of $H$ can be flipped alone by a local Pauli operator,
or there are two terms $h_x$ and $h_y$ that are translates of $h$ along $x$- and $y$-directions,
respectively,
such that the pair $h,h_x$ are simultaneously flipped by a local Pauli operator,
and so are the pair $h,h_y$.
\end{lemma}

This means that \emph{every topological charge is attached to a string operator},
and we may speak of its topological spin.
The local nonredundancy is not important and will not be used for the rest of this subsection,
but is included for clarity of the presentation of the statement.
This is a special result tailored to two dimensions.

\begin{lemma}\label{lem:2dbd-ltqo}
For any exact code Hamiltonian $H$ in two dimensions,
let $H_{bulk}$ be the sum of all terms of $H$ that is supported on the half-plane $y \le 0$.
Then, there exists a local Hamiltonian $H_{bd}$ within distance $O(1)$ of the boundary $y=0$
that is translation-invariant along the boundary 
(where the translation group of $H_{bd}$ may be smaller, i.e., coarser,
than that of $H$ along $x$-direction)
such that $H_{w/bd} = H_{bulk} + H_{bd}$ is locally topologically ordered
and consists of commuting Pauli operators.
%Moreover, for any term $h$ of $H_{w/bd}$,
%either $h$ can be flipped alone by a local Pauli operator,
%or there exists a translate $h'$ of $h$ along $x$-direction 
%such that $h$ and $h'$ are simultaneously flipped by a local Pauli operator.
\end{lemma}
The construction of $H_{w/bd}$ is very similar to that of $\tilde H$ in the previous subsection,
except that we have to find a boundary Hamiltonian.
\begin{proof}
Since $H$ obeys the LTQO,
we know any stabilizer (a product of terms of the Hamiltonian) on a rectangle
is a product of terms within $r$-neighborhood of the rectangle for some $r = O(1)$.
In fact, an analogous fact for excitations is also true:
\emph{
Any excitation of $H$ on a rectangle, 
which is created by a finitely supported operator that may or may not be supported on the rectangle,
is created by an operator within $r'$-neighborhood of the rectangle for some $r' = O(1)$.
}
This can be understood from \cref{lem:2dbulk}
since any charge is moved by string operators along any direction.
Let $L_y$ be a positive constant that is larger than $r$, $r'$, and the interaction range of $H$.

As in \cref{rem:ltqo}, consider the multiplicative group $S_{bulk}$
of all terms of $H_{bulk}$, and define $S|_{3L_y}$ as the group of factors of $s \in S_{bulk}$
that lie in the strip $-3L_y \le y \le 0$.
The group $S|_{3L_y}$ is translation invariant along $x$-direction.
Setting $G = S|_{3L_y}$ in \cref{thm:maximal-commutative-algebra},
we obtain a subgroup $S$ of the commutant of $G$
such that (i) $S$ admits a generating set that is translation-invariant along $x$-direction 
(with possible spontaneous translation symmetry breaking), and
(ii) if a finitely supported Pauli operator $P$ commutes with everything in $G$ and $S$, then $P$ belongs to $S$.
We declare the sum of operators in the generating set of $S$ to be $H_{bd}$.

We have to show that $H_{w/bd} = H_{bulk} + H_{bd}$ obeys the LTQO.
Let $P$ be any Pauli operator on a finite rectangle that commutes with every term of $H_{w/bd}$.

First, if $P$ is supported far from the boundary by distance $> L_y$,
then $P$ commutes with every term of $H$,
and by the LTQO of $H$, we know $P$ can be written as a product of terms
in the half-plane $y \le 0$.
Since $H_{w/bd}$ includes all such terms, $P$ is a product of terms of $H_{w/bd}$.

Second, if the support of $P$ overlaps the strip $-L_y \le y \le 0$,
then we temporarily regard $P$ as an operator in the system of $H$ on the full infinite plane.
By assumption, $P$ commutes with every term of $H$ in the lower half-plane,
but $P$ may not commute with some terms of $H$ that are not in the lower half-plane.
So, some excitation may be created by $P$ from a ground state of $H$,
but the excitation must reside in the strip $-L_y \le y \le + L_y$.
By the remark in the first part of this proof,
we know there exists a Pauli operator $Q$ within the strip $-2L_y \le y \le 2L_y$
such that $PQ$ commutes with every term of $H$.
The product $PQ$ is supported on the rectangle $A$
that is enlarged from the support of $P$ by a rectangle in the strip $-2L_y \le y \le 2L_y$.
Then, by the LTQO of $H$, $PQ$ is a product of terms on the $L_y$-neighborhood of $A$.
Collect the terms that participate in $PQ$ but are supported on the region $y < -2L_y$;
call the product of these as $T$. Note that $T$ does not overlap with $Q$.
The role of $T$ is to ``push'' $P$ to the boundary.
Indeed, $PQT$ is supported on the region $ y \ge -3L_y$, so is $PT$.

Now we go back to the system of $H_{w/bd}$ with the operator $PT$ 
which commutes with every term of $H_{w/bd}$ and is supported on the region $y \ge -3L_y$.
In particular, $PT$ commutes with $G$ and $S$,
and therefore, $PT$ must be a product of terms of $H_{bd}$.
\end{proof}

A similar argument can be repeated to a bottom boundary at, say $y = -10 L_y$.
The Hamiltonian with LTQO on a system with two boundaries at $y = 0, -10L_y$ is our object in the next lemma.
Let us say that a boundary of an exact code Hamiltonian $H$ is {\bf gapped out}
by a commuting Pauli Hamitonian $H_{bd}$ near the boundary
if $H_{w/bd}$ that includes all terms of $H$ on one side of the boundary and all terms of $H_{bd}$,
obey the local topological order condition.

\begin{lemma}\label{lem:2dChargeAbsorb}
Suppose that $H$ is a two-dimensional exact code Hamiltonian that admits a nontrivial topological charge.
Then, in a quasi-1D system obtained by gapping out two parallel boundaries $y = 0$ and $y = -10L_y$
for any sufficiently large $L_y$,
an excitation that lies in the region $-6L_y \le y \le -4 L_y$ 
and represents some nontrivial topological charge of $H$,
is created by a finite Pauli operator.
\end{lemma}

That is, some nontrivial charge is ``absorbed'' by the gapped boundary.

\begin{proof}
We choose $L_y$ as in \cref{lem:2dbd-ltqo},
so that the resulting quasi-1D Hamiltonian $H_1$ with two boundaries obeys the LTQO at scale $5L_y$.
The system size is infinite along $x$-direction, but finite along $y$-direction.

First, suppose that under a periodic boundary condition along $x$-direction with period $L_x$ much larger
than the interaction range $C$ of $H_1$,
the quasi-1D system has a degenerate ground state subspace.
The range $C$ may be larger than the interaction range of the bulk terms,
since that of the boundary terms may be larger,
but is independent of $L_x, L_y$.
The argument of~\cite{BravyiTerhal2009no-go} proves that 
there exists a Pauli operator $P$
supported on a rectangle $0 \le x \le 2C$
such that $P$ induces a nontrivial transformation on the ground space.
That is, there must be a nontrivial logical operator that traverses two boundaries.
Write $P = P_{low} P_{high}$ as a product of two Pauli operators $P_{low}$ and $P_{high}$
where $P_{low}$ is the factor in the region $-10 L_y \le y \le -5 L_y$, 
and $P_{high}$ in the region $-5 L_y < y \le 0$.

The excitation created by $P_{low}$ in the bulk must represent a nontrivial topological charge;
otherwise, there would be an operator $Q$ around the excitation such that 
$P_{low}Q$ (and also $P_{high}Q^\dagger$) do not create any excitation 
by the remark in the beginning of the proof of \cref{lem:2dbd-ltqo},
and therefore each of $P_{low} Q$ and $P_{high} Q^\dagger$ would be a product of terms of $H_1$ by the LTQO,
contradicting the fact that $P$ induces nontrivial transformation on the ground space.
The excitation created by $P_{low}$ is the promised one.

Second, suppose that under a periodic boundary condition along $x$-direction with period $L_x$ much larger
than the interaction range $C$ of $H_1$,
the quasi-1D system has a nondegenerate ground state.
The classification of translation-invariant 1D commuting Pauli Hamiltonians~\cite[Thm.~3]{Haah2013},
basically says that up to a local Clifford circuit, 
there are only the Ising model, the trivial Hamiltonian $-\sum_j Z_j$
where qubits are frozen, and non-interacting qubits.%
\footnote{We could use this result in the first case, but we cite an earlier result that suffices.}
Applied to the present case, this implies that up to a Clifford circuit of depth that depends on $L_y$ but not on $L_x$,
the system must be equivalent to a trivial Hamiltonian.
(The Ising model is ruled out by the nondegeneracy assumption, 
and non-interacting qubits are ruled out by the LTQO.)
In particular, $H_1$ that is infinite along $x$-direction
has no nontrivial topological charge.
Since we assume some flip of bulk terms of $H_1$ represents a nontrivial charge in the infinite 2D system,
we conclude that an excitation of nontrivial charge in the bulk
can be created by a finite operator in the quasi-1D system.
\end{proof}

\begin{corollary}\label{cor:2dBoson}
For any two-dimensional exact code Hamiltonian,
either every topological charge is trivial,
or there exists a nontrivial boson.
\end{corollary}
\begin{proof}
Assuming there is a nontrivial topological charge in the bulk,
by \cref{lem:2dChargeAbsorb} we have some topological charge 
that can be inserted in the quasi-1D system on $-10L_y \le y \le 0$ by a finitely supported operator $P$.
This finite operator can in fact be chosen to live on $ -5L_y \le y \le 0$ 
so that it does not touch the bottom boundary:
The reason is similar to the first case in the proof of \cref{lem:2dChargeAbsorb}.
We divide any given $P$ into two pieces $P_{low} P_{high}$ where $P_{low}$ is supported on 
$-10L_y \le y \le -5L_y$ and $P_{high}$ on $-5 L_y < y \le 0$.
Either of $P_{low}$ or $P_{high}$ must insert a nontrivial charge due to the LTQO condition.

\begin{figure}
\centering
\includegraphics[width=0.6\textwidth,trim={0cm 14cm 23cm 0cm},clip]{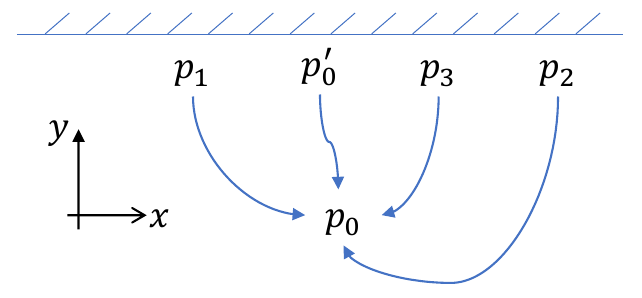}
\caption{A charge inserted from a boundary is a boson.}
\label{fig:bosonAtBd}
\end{figure}

We thus have shown that there is some nontrivial charge $b$
that can be inserted to the bulk from one boundary at, say $y=0$.
We no longer consider the quasi-1D system, but the semi-infinite system on the half-plane $y \le 0$.
Let $P_0$ be the charge insertion operator into an $O(1)$-neighborhood of a point $p'_0$ near the boundary.
Place points $p_1, p_2, p_3$ near the top boundary
so that translates $P_1,P_2,P_3$ of the operator $P_0$ 
insert $b$ into the respective neighborhoods of points $p_j$.
We order the points $p_j$ as in \cref{fig:bosonAtBd}.
The points are sufficiently far apart so that the translates of $P_0$ do not overlap.
By \cref{lem:2dbulk} we know every charge can be moved by a string operator.
Let $t'$ be the string operator that moves $b$ from $p'_0$ to $p_0$
where $p_0$ is a point deep in the bulk.
Similarly, let $t_j$ be the string operator that moves $b$ from $p_j$ to $p_0$ as in \cref{fig:bosonAtBd}.
By the LTQO condition, we have identity actions on a ground state $\ket \psi$
(up to an unimportant scalar):
\begin{align}
(t' P_0)^\dagger  t_j P_j \ket \psi &=\ket \psi, & j = 1,2,3. \label{eq:idaction}
\end{align}
Since $P_j$'s all commute as they are nonoverlapping, we see
\begin{align*}
\ket \psi 
&= (t' P_0)^\dagger  t_1 P_1 \cdot P_2^\dagger t_2^\dagger t' P_0 \cdot (t' P_0)^\dagger  t_3 P_3 \ket \psi 
& (\text{by \cref{eq:idaction}}) \\
&= P_2^\dagger (t' P_0)^\dagger  t_1   t_2^\dagger t' P_0 (t' P_0)^\dagger  t_3 P_1 P_3  \ket \psi 
& (\text{moving }P_1,P_2) \\
&= P_2^\dagger (t' P_0)^\dagger  t_1   t_2^\dagger  t_3 P_1 P_3  \ket \psi 
& (\text{cancelling } t' P_0) \\
&= \theta_b P_2^\dagger (t' P_0)^\dagger  t_3   t_2^\dagger  t_1 P_1 P_3  \ket \psi 
& (\text{definition of }\theta_b) \\
&= \theta_b P_2^\dagger (t' P_0)^\dagger  t_3   t_2^\dagger  t' P_0 (t' P_0)^\dagger t_1 P_1 P_3  \ket \psi 
& (\text{uncancelling } t' P_0) \\
&= \theta_b (t' P_0)^\dagger  t_3 P_3 \cdot P_2^\dagger t_2^\dagger t' P_0 \cdot (t' P_0)^\dagger  t_1 P_1 \ket \psi 
& (\text{moving }P_1,P_2)\\
&= \theta_b \ket \psi 
& (\text{by \cref{eq:idaction}}).
\end{align*}
Therefore, $\theta_b = +1$ for $\ket \psi$ to be nonzero.
\end{proof}

\subsection{Nontrivial 3D QCA or Nontrivial 2D Fermionic QCA}

In \cref{ntcliff},  we have shown that the disentangling QCA $\alpha_{WW}$ 
cannot be decomposed as a quantum circuit using Clifford gates.
In this subsection, we consider the general question of 
whether $\alpha_{WW}$ is trivial using arbitrary gates, 
i.e., whether $\alpha_{WW} \otimes \Id$ is in $\Circshift$ for some choice of stabilization.
Here we only allow stabilization with additional qudit degrees of freedom, not additional fermionic degrees of freedom.
That is, we are only interested in whether $\alpha_{WW}$ is nontrivial as a qudit QCA.

If $\alpha_{WW}$ is indeed trivial, then we can construct $H_{3F}$ as above.
So, any Hamiltonian term of the Walker-Wang model supported on the top layer 
will also be in the group generated by the terms of $H_{3F}$.
These terms supported on the top layer are those shown pictorially in the last row of \cref{fig:WWsurface},
as well as their translates.  
% The operators $X_1,X_2,Y_1,Y_2,Z_1,Z_2$ represent Pauli operators on the first or second qubit on the given edge.
These operators can be thought of as describing ``small loops'' of fermions, 
and they have expectation value $+1$ in the ground state of $H_{3F}$.
We call the two different small loops shown the $\ftwo$ loop and the $\fone$ loop respectively, reading from left to right.  
We define an $\fthree$ loop to be the product of the $\ftwo$ loop and the $\fone$ loop.

All the other terms in the Hamiltonian commute with these, 
i.e., they are in the commutant algebra of the algebra generated by these small fermion loops.
We will find it useful to consider the commutant algebra of the first of the two small fermion loops, 
i.e., the $\ftwo$ loop, as this algebra has a nice representation in terms of Majorana operators. 
See Ref.~\cite{Kapustin2d}.
This commutant algebra is generated by arbitrary operators acting on the first qubit on each edge
and arbitrary operators acting on the additional degrees of freedom (those other than the first or second qubit),
and by a certain algebra, which we denote $\commA$, of operators acting on the second qubit.
This algebra $\commA$ is generated by the operators shown in the top row of \cref{fig:majorana}.
In this figure we omit the subscript $2$; symbols $Z,X$ denote Pauli operators.

The bottom row of \cref{fig:majorana} gives a representation of this commutant algebra in terms of Majorana operators, 
with two Majorana operators $\gamma,\gamma'$ on each vertex.
This representation is obtained by mapping each operator in the top row 
to the operator in terms of Majorana operators directly below it.  
We show the edges on the bottom row to help the reader understand the geometry,
but we emphasize that the Majorana operators are on vertices, rather than edges.  
Thus, for example, the operator on the left-hand side of the figure 
with a $Z$ on a vertical edge and an $X$ on a horizontal edge as shown,
maps to an operator which is a product of $\gamma'$ 
on the vertex at the top of the vertical edge and $\gamma$ at the bottom of the edge.
One may verify that this mapping preserves the commutation relations.

This representation using Majorana operators is 
{\it not} a faithful representation of the commutant algebra ${\commA}$.  
Rather, one may show that the $\ftwo$ loop operators are equal to $+1$.
We will call the system on which this representation using Majorana operators acts 
(i.e., the system with two Majorana operators $\gamma,\gamma'$ per vertex, 
as well as degrees of freedom of $H_{3F}$ other than the second qubit) 
the {\bf Majorana system}, 
while calling the system of $H_{3F}$ the {\bf qudit system}.

\begin{figure}[t!]
\includegraphics[width=3in]{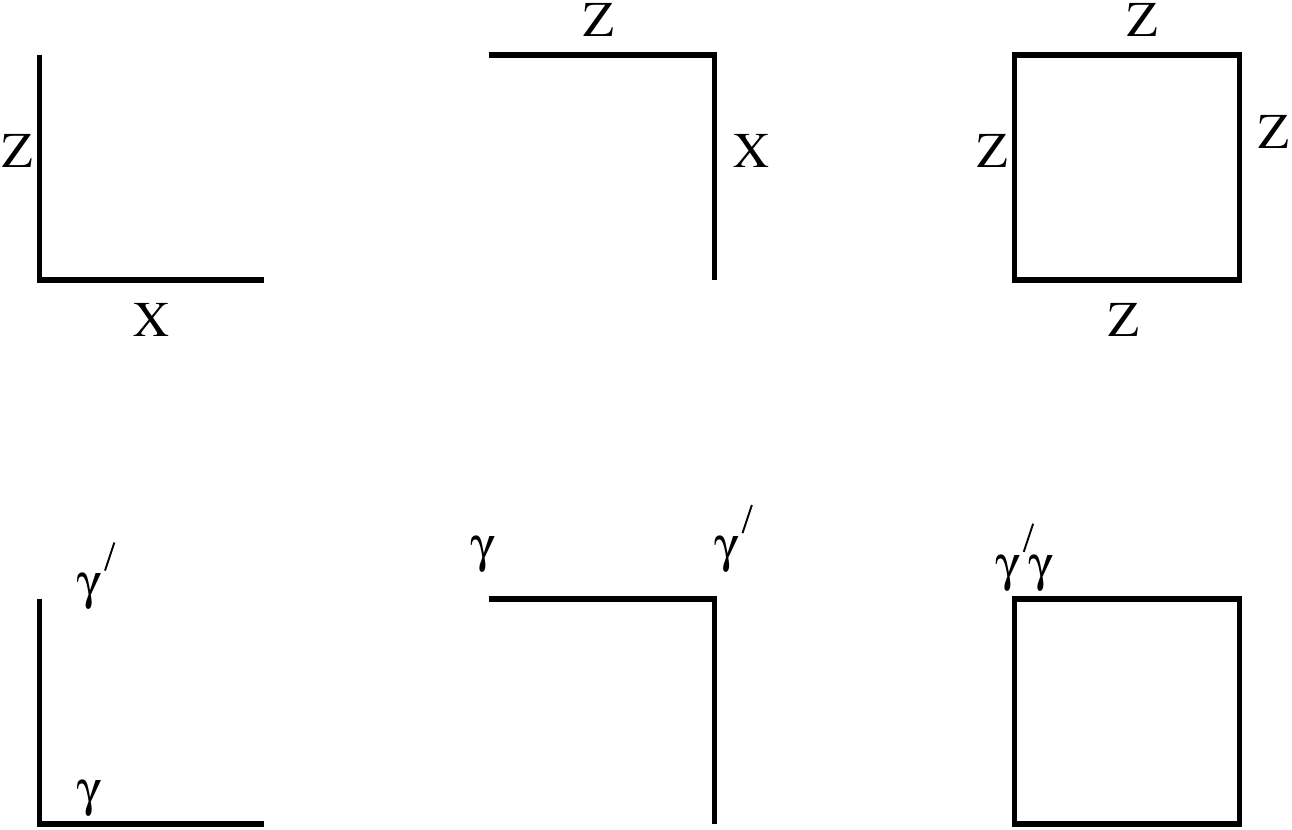}
\caption{Majorana representation of algebra $\commA$.  The three operators on the top row (and their translates) generate the commutant, as may be verified by the polynomial method.  The bottom row gives a representation of this algebra using Majorana operators.}
\label{fig:majorana}
\end{figure}

We emphasize that {\it any} fermion bilinear is the image of some operator in the commutant algebra.
An open string of $\ftwo$ fermions (more precisely, such an open string modified at the endpoints) 
gives such a fermion bilinear.
For an example, see \cref{fig:bilinear}
where we have drawn an open string of $\ftwo$ fermions, 
by multiplying five string segments from \cref{fig:open}, 
and then further multiplying by an additional two $X_2$ operators at each end of the string.  
The multiplication by the $X_2$ operators is necessary to have an operators in the commutant algebra;
without these operators the open string would fail to commute at the ends.
One may verify that this operator indeed is a fermion bilinear as follows.
First, remove the $X_2$ operators by multiplying by operators from \cref{fig:majorana}, 
i.e., multiply by operators from \cref{fig:majorana} 
to replace each $X_2$ on some segment with $Z_2$ on another segment.  
Each of these operators gives a fermion bilinear.  
Then, what is left is a homologically trivial closed loop of $Z_2$,
which is a product of loops of $Z_2$ around plaquettes, 
which can be mapped using \cref{fig:majorana} into a product of $\gamma' \gamma$ on vertices.
After some algebra, one may verify that it indeed is a fermion bilinear.

Thus we have
\begin{lemma}
\label{lem:image}
Let $\varphi(\cdot)$ be the mapping from the qudit representation to the Majorana repsentation.
Let $O$ be any operator in the Majorana representation with even fermion parity.
Suppose that $O$ is supported on some set $S$ such that given any pair of vertices in $S$, 
there is a path of edges in $S$ connecting those two vertices.
Then, $O$ is the image under $\varphi$ of some operator in the qudit representation supported within
distance $O(1)$ of $S$.
\end{lemma}
\begin{proof}
Decompose $O$ as a sum of products of Majorana operators and qudit operators supported on $S$.
By assumption, only even products of Majorana operators appear in the sum.
Any product of a pair Majorana operators, supported on some pair of sites $i,j$ respectively,
is the image under $\varphi$ of some operator in the qudit representation 
on an (arbitrary) path of edges connecting those two sites $i,j$, 
multiplied by some operator supported within distance $O(1)$ of the endpoints.
\end{proof}

\begin{figure}[t!]
\centering
\includegraphics[width=0.7\textwidth,trim={0cm 12cm 13cm 0cm},clip]{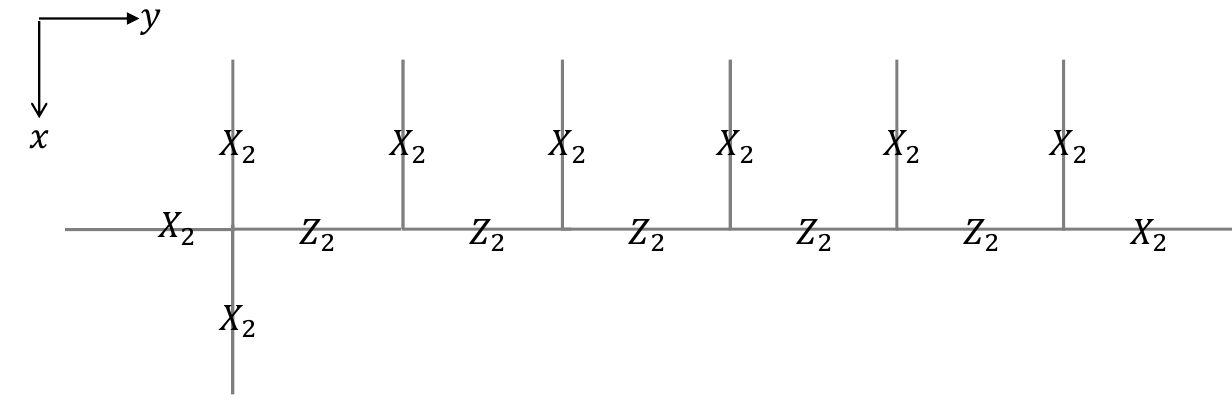}
\caption{Open string of $\ftwo$ fermions, modified by multiplying by two extra $X_2$ operators at both left and right ends.  One may verify that this maps to a fermion bilinear, with one fermion operator at each end of the string.}
\label{fig:bilinear}
\end{figure}

Thus, using this representation of the commutant $\commA$,
the Hamiltonian $H_{3F}$ includes the $\ftwo$ loop operators,
arbitrary terms which have even fermion parity in the Majorana operators $\gamma,\gamma'$,
and arbitrary operators on the additional degrees of freedom.

Let us heuristically explain why this representation will be useful: 
roughly speaking, it turns a model with three fermion topological order into a model with no topological order.
The three fermion model has string operators which create defects at their ends.
One such operator is an open string of an $\ftwo$ fermion (modified at the endpoints).  
This operator maps to a fermion bilinear as noted.
That is, a string operator maps to an operator which acts only at the ends of the string.  
We will enlarge the Hilbert space of this fermionic model by allowing states with either even or odd fermion parity;
that is, we will consider operators with odd fermion parity,
so an open string of $\ftwo$ fermions becomes just a product of two local operators,
one at each end of the string.
The  $\fone$ loop is a product of operators acting on the first qubit 
times a product of $Z_2$ around a closed loop; 
from \cref{fig:majorana}, the product of $Z_2$ 
corresponds to the fermion parity $\gamma' \gamma$ at a vertex of the loop.  
So, if we take the product of all small $\fone$ loops inside some disk, 
in the Majorana representation the result is equal to 
some operator on the second qubits around the boundary of the disk 
times the product of all $\gamma' \gamma$ inside the disk.
At the same time, in the qubit representation this product of all small $\fone$ loops inside a disk 
is an operator supported on the boundary of the disk that is a ``large $\fone$ loop.''
Thus, while in the qubit representation, there are stringlike operators 
which are large loops of $\ftwo$ or $\fone$ fermions (or their product, an $\fthree$ loop),
in the Majorana representation, the corresponding operators are 
either pointlike (for open $\ftwo$ loops) or disk-like (the disk being the interior of an $\fone$ loop).

To make this precise we show that:
\begin{lemma} \label{localflipMajorana}
If $H_{3F}$ exists, then for sufficiently large system size, there exists a locally flippable separator for the Majorana system, whose
terms generate the same multiplicative group as the terms of $H_{3F}$ in this representation.
\end{lemma}
\begin{proof}
Write $H_{w/bd}$ in the Majorana representation.  That is, introduce (just at the top boundary $z=0$) a pair of Majoranas $\gamma,\gamma'$ on each vertex, and write all terms in the $z=0$ plane which act on the second qubit in terms of these Majoranas $\gamma,\gamma'$.

Take the resulting Hamiltonian and from this Hamiltonian construct a new Hamiltonian, $H_{Maj}$, following the same steps that we used to construct $H_{3F}$ from $H_{WW}$.  In detail: the QCA $\beta$ acts on the system of qubits (as well as any additional degrees of freedom) rather than on the Majorana system but we since $\beta$ acts as the identity on operators supported on the $z=0$ plane, we can define a QCA $\beta_{Maj}$ such that $\beta_{Maj}(O)=\beta(O)$ for any operator $O$ whose support does not include the $z=0$ plane and such that
$\beta_{Maj}(O)=O$ is $O$ is supported on the $z=0$ plane.
Define $\tilde H_{Maj}$ by writing the terms of $\tilde H$ in the Majorana representation.
Define $H_{Maj}$ as
the sum of terms of $\beta_{Maj}( \tilde H_{Maj})$ that are supported within distance $3L_z$ of the top boundary $z=0$. 
Regard $H_{Maj}$  as a two-dimensional lattice Hamiltonian by ignoring the $z$ coordinates, regarding all sites in the three-dimensional lattice with the same $x,y$ coordinate but different $z$ coordinates as corresponding to the same site in some two-dimensional lattice.

Just as it was possible to construct $H_{3F}$ so that the terms are locally nonredundant using the terms
\cref{lem:nonredudantBd} to define $H_{3F}$, it is also possible to
construct $H_{Maj}$ so that its terms are locally non-redundant, 
since, as remarked in \cref{rem:includingFermions}, 
the same result applies for translationally invariant Hamiltonians 
which are sums of products of Pauli or Majorana operators.

Consider the Hamiltonian $\beta_{Maj}^{-1}(H_{Maj})$.  
Here, while the system of $H_{Maj}$ includes only sites within distance $3L_z$ of the top boundary $z=0$,
we identify each term in $H_{Maj}$
operator on the system of $H_{3F}$ (which excludes the deep bulk),
we identify $O$ with $O \otimes I$ that is an operator on the system of $\beta(\tilde H_{Maj})$ 
(which includes the deep bulk), so in that way we define $\beta_{Maj}^{-1}(H_{Maj})$.
This Hamiltonian is translation invariant and is a sum of products of Paulis and Majoranas.  
Note that $\beta_{Maj}^{-1}(H_{Maj})$ will necessarily have a degenerate ground state 
because it has no terms supported in the bulk, far from the top boundary.

We will show later that $\beta_{Maj}^{-1}(H_{Maj})$ has no nontrivial charges.  
Then, as remarked in \cref{rem:includingFermions},
\cref{lem:NoredunNochargeFlippable-qudit} also applies to Hamiltonians 
which are sums of products of Paulis and Majoranas; 
applying this lemma to $\beta_{Maj}^{-1}(H_{Maj})$, 
it follows that every term of $\beta^{-1}(H_{3F})$ 
can be flipped by a product of Paulis and Majoranas of bounded support.  
Hence, on any sufficiently large finite system (larger than the support of the flippers),
the terms of $\beta_{Maj}^{-1}(H_{Maj})$ are also nonredundant:
to see this, note that any flipper of bounded support can only flip terms in its support,
and so the result of the infinite system implies 
that the flipper will flip exactly one term in the finite system, 
and since each term can be flipped independently, they must be nonredundant.

Thus, each term of $H_{Maj}$ can be flipped by an operator of bounded support, 
and the terms of $H_{Maj}$ are nonredundant.  
While $\beta_{Maj}^{-1}(H_{Maj})$ has a degenerate ground state, 
$H_{Maj}$ has a unique ground state as can be verified using nonredundancy of terms and counting the number of terms.%
\footnote{
Another way to see the nondegeneracy is to consider a logical operator
that must exist on a thin subsystem of codimension one (in this case a string-like region),
break up the logical operator into halves,
and use the fact that there is no nontrivial charge 
and that the Hamiltonian obeys the local topological order condition.
}
Hence, the ability to locally flip the terms and the nonredundancy of terms means that the
terms of $H_{Maj}$ define a locally flippable separator, completing the proof.

It remains to show that $\beta_{Maj}^{-1}(H_{Maj})$ has no nontrivial charges.
Let $B$ be a set of points.
Let $O$ be an operator acting on the Majorana system that commutes with all terms in $H_{Maj}$ 
except those whose support overlap with~$B$.
Then, $\beta^{-1}$ commutes with all terms in $H_{3F}$ except those whose support is within $O(1)$ of $B$.
So, by \cref{lem:chargesOfH3F}
Then, $\beta^{-1}(O)$ can be written as
a $\CC$-linear combination of operators, each of which is a product of form $h_{3F} O_f O_B$,
where $h_{3F}$ is a product of terms of $H_{3F}$,
$O_f$ is a fermion string operator whose end points are within distance $O(1)$ from $B$,
and $O_B$ is an operator supported within distance $O(1)$ from $B$.
If $B$ is empty, then $O_f$ must be a closed fermion string operator and $O_B = I$.

However, we also have that $\beta^{-1}(O)$ commutes with {\it all} $\ftwo$ loop operators, 
even those supported near $B$.  
This allows us to show that $O_f$ must be a type-$2$ string operator.  
That is, the type-$1$ chain associated with $O_f$ must be trivial.  
The reason is that any nonvanishing type-$1$ chain would anti-commute with the $\ftwo$ loop.  
There is a subtlety: one might wonder whether it is possible to have a nonvanishing type-$1$ chain 
but also choose $O_B$ so that the operator commutes with the $\ftwo$ loops.  
(In fact, we did this for type-$2$ strings, 
multiplying for example by extra factors of $X_2$ in as in \cref{fig:bilinear}.)
To show that this is not possible, 
consider the product of $\ftwo$ loop operators in a disc of radius $O(1)$ near each point in $B$, 
choosing the disc large enough that its boundary is disjoint from the support of $O_B$.  
Then, any nonvanishing type-$1$ chain does not commute with this product, 
and this product commutes with $O_B$. 
(Equivalently, note that any type-$1$ string 
anti-commutes with an {\it odd} number of $f_2$ loop operators at its endpoints, 
which cannot be corrected by a local $O_B$.)

Now, $O_f$ is a type-$2$ string operator, 
and hence $\beta(O_B O_f)$ is a product of Majorana operators at the ends of the string, 
multiplied by some local operators, 
i.e., it is a product of operators supported near points in $B$ so creates trivial charge.
\end{proof}

By \cref{thm:qcadisent} and \cref{localflipMajorana}, 
if $\alpha_{WW}$ is trivial, 
then there exists a two-dimensional fermionic QCA $\alpha_{Maj}$ 
which disentangles the locally flippable separator defined by \cref{localflipMajorana}.
There is one detail here: in the fermionic case, 
\cref{thm:qcadisent} requires a bijection $g(\cdot)$ 
from elements $a$ of the separator with $D_a=2$ to degrees of freedom $g(a)$ with $D_a=2$ 
such that $\dist(a,g(a))=O(R)$ 
and such that, if $\tilde \opX_a$ has even fermion parity,
then $g(a)$ is a qudit and if $\tilde \opX_a$ has odd fermion parity then $g(a)$ is a fermion.
However, $\beta_{Maj}^{-1}(H_{Maj})$ is translation invariant,
and so in each unit cell we have some number of flippers which are fermion parity odd
and some which are fermion parity even.
There must be at least one fermion parity odd flipper;
otherwise, all states would have the same fermion parity.
Once we have at least one fermion parity odd flipper in each cell,
so that we have some number $n_{odd}>0$ of odd parity flippers 
and $n_{even}$ even parity flippers 
(just considering the flippers on degrees of freedom with dimension $2$, 
ignoring any added degrees of freedom for stabilization),
we can find another set of terms with the same cardinality 
which generates the same algebra with any desired nonzero number of odd flippers;
this is a consequence of the fact that the 
(graded) operator algebra on a system of $n > 0$ Majorana modes and $m$ qubits 
is equivalent to the (graded) operator algebra on any other system of
$n'>0$ Majorana modes and $m'$ qubits with $n+2m=n'+2m'$.
(Proof of fact: Consider a Jordan-Wigner transformation, 
which is a faithful representation of the Majorana operators. 
Tensor in qubits, and interpret all the Pauli operators as a result of Jordan-Wigner on $n+2m$ Majorana operators.)
In this way, we ensure that such a bijection $g$ exists.

We now show that
\begin{theorem}
If $\alpha_{WW}$ is trivial, then $\alpha_{Maj}$ is nontrivial.
\end{theorem}

\begin{proof}
The QCA $\alpha_{Maj}$ acts on the Majorana system.  
Note that if $\alpha_{Maj}$ is trivial, 
it can be decomposed as a circuit QCA $\alpha_C$ followed by a shift QCA $\alpha_S$.
Hence, the circuit $\alpha_C$ acting on the separator defined by \cref{localflipMajorana} 
is equal to $\alpha_S^{-1}$ acting on the trivial separator.
A shift leaves all qudit elements of the trivial separator invariant, 
so we may assume without loss of generality that the shift $\alpha_S$ acts only on the fermionic degrees of freedom.
In this case, the image of each trivial separator elements $i \gamma' \gamma$ under the shift $\alpha_S^{-1}$ 
is a Majorana bilinear, with the two operators in the bilinear separated by a distance $O(1)$.
By a standard argument (see \cref{alphapC}), on a two-torus there exists a circuit $\alpha'_C$ 
which maps such a separator to the trivial separator, except possibly 
on a small subsystem of the two-torus consisting of two (long) lines,
one horizontal and one vertical, 
where on such long lines the system may be in either the trivial separator
or in the trivial separator shifted by one Majorana mode, i.e., in a Majorana chain.
We absorb this circuit $\alpha'_C$ into $\alpha_C$, redefining $\alpha_C\rightarrow \alpha'_C \circ \alpha_C$, 
so that the circuit QCA $\alpha_C$ maps the separator defined by \cref{localflipMajorana} 
to the trivial separator except possibly on these two long lines.

We will show that
there exists some other QCA $\alpha_{qud}$ 
acting on the qudit system such that the representation map $\varphi$ converts $\alpha_{C}$ into $\alpha_{qud}$:
let $\varphi$ be the map from operators $O$ on the qudit system which commute with the $\ftwo$ loop, 
to operators in the Majorana representation.
We wish that for any such $O$, we have 
\[
\alpha_{C}(\varphi(O))=\varphi(\alpha_{qud}(O)).
\]
Further, we wish that $\alpha_{qud}$ maps any small $\ftwo$ loop to itself.

Given these properties of $\alpha_{qud}$, applying $\alpha_{qud}$ to the terms of $H_{3F}$ 
gives a set of terms which includes the $\ftwo$ loop 
and includes for each vertex a term $O$ 
such that $\varphi(\alpha_{qud}(O))$ is equal to $\gamma' \gamma$ on that vertex,
at least in the region far away from the two long lines.
In what follows, we will ignore this small subsystem.
However, then this term $O$ is equal to a product of $Z_2$ around a plaquette 
(since that is the operator whose image under $\varphi$ is $\gamma' \gamma$), 
possibly multiplied by a product of $\ftwo$ loops.

Thus, the terms of $\alpha_{qud}(H_{3F})$ generate a multiplicative group 
including the $\ftwo$ loop and the product of $Z_2$ around a plaquette.
These two types of terms then generate also the product of $X_2$ on all links attached to a given vertex.
Thus, the terms of $\alpha_{qud}(H_{3F})$ generate a multiplicative group 
which includes the terms of the toric code on the second qubit, 
i.e., the product of $Z_2$ around a plaquette and the product of $X_2$ around links attached to a vertex.

For the toric code, there are topological charges at the end points of open strings of $Z_2$ 
or open strings of $X_2$ on the dual lattice. 
These topological charges are bosons $a$,
i.e., following \cref{def:topologicalSpin} they have $\theta_a=+1$.
Let $O$ be an open string operator which creates such bosonic charges in the toric code;
then, $\alpha_{qud}^{-1}(O)$ creates bosonic charges in $H_{3F}$, 
inconsistent with \cref{lem:chargesOfH3F} which shows that all topological charges in $H_{3F}$ are fermionic.

It remains to construct $\alpha_{qud}$.  
The QCA $\alpha_{qud}$ will also be a circuit.
For each gate $U$ in the circuit for $\alpha_{C}$, 
there will be a corresponding gate $V$ in the circuit for $\alpha_{qud}$ acting on the qudit system 
such that $U^\dagger \varphi(O) U = \varphi(V^\dagger O V)$ for any $V$ which commutes with $\ftwo$.
So, it suffices to find $V$ with the correct support (i.e., on a set of bounded diameter) such that $U=\varphi(V)$.

Suppose $U$ is supported on some set $S$.
To construct $V$, we will apply \cref{lem:image}.
This lemma requires that $S$ contain some path of edges connecting any pair of vertices in $S$,
and perhaps $S$ may not have this property.
However, given set an arbitrary set $S$,
we can add edges to this set so that $S$ obeys this requirement of \cref{lem:image} 
{\it without increasing the diameter of $S$}, 
i.e., simply add edges on any shortest path between each pair of vertices.
So, adding edges to $S$ in this way, we apply \cref{lem:image},
and this gives a bound on the diameter for the support of $V$.

Note that here we use the locality of the gates to ensure that $\alpha_{qud}$ will also be local;
given an arbitrary QCA $\alpha$ acting on the fermionic system,
it is not obvious how to construct a QCA $\alpha'$ acting on the qudit system 
such that $\alpha(\varphi(O))=\varphi(\alpha'(O))$.
\end{proof}

\begin{lemma}
\label{alphapC}
Consider a system of fermionic degrees of freedom with sites on the two-torus.  
Consider the image of the trivial separator under a shift QCA of range $O(1)$.
Then, there is a circuit of depth $O(1)$ and range $O(1)$ such that, acting on that image, the result is
a separator whose elements are the same as the elements $i\gamma_j \gamma'_j$ of the trivial separator, except possibly on a small subsystem of the two-torus consisting of two (long) lines,
one horizontal and one vertical, 
where on such long lines the system may be in either the trivial separator
or in the trivial separator shifted by one Majorana mode, i.e., in a Majorana chain.
\begin{proof}
 Consider the image of the trivial separator under the shift.  Tile the two-torus with hexagons, each of linear size which is much larger than the shift but still $O(1)$, so that at most three hexagons are within distance $O(1)$ of any point.  Given any pair of neighboring hexagons, there may be some number of separator elements which are a product of one Majorana operator in each hexagon.  By a unitary supported within those two hexagons, we can change this number by any even amount, changing it to equal $0$ or $1$.  For example, if $i\gamma_i\gamma_j$ and $i\gamma_k\gamma_l$ are separator elements with $i,k$ in one hexagon and $j,l$ in another, then swapping $\gamma_i$ and $\gamma_l$ reduces the number by two.  This gives one unitary for each pair of hexagons but these unitaries can all be done in parallel
since they have disjoint support; since the diameter of the hexagons is $O(1)$, this product of unitaries is a quantum circuit.
Note that there are no separator elements which are a product of two Majorana operators in non-neighboring hexagons.

 Define a chain complex, with hexagons corresponding to $0$-cells, a $1$-cell attached to any pair of neighboring hexagons, and a $2$-cell for each triple of hexagons meeting at a point.  Define a $1$-chain with ${\mathbb Z}_2$ coefficients where the coefficient on an edge is equal to the number of
separator elements which are a product of one Majorana operator in each hexagon.  By a gate supported on any triple of hexagons corresponding to a $2$-cell, we can change this chain by a boundary, turning it into any homologous chain.  In particular, the chain can be made equal to zero everywhere, except possibly on two long one-dimensional lines if the chain is homologically nontrivial.  These gates can all be done in parallel in depth $O(1)$ since we can color the $2$-cells by $O(1)$ colors (indeed, by $2$ colors) and in each round we only apply the gates corresponding to a given color.  

Finally, apply a further circuit so that on any hexagon with all $0$ coefficients on the edges attached to that hexagon, we map all separator elements supported on that hexagon to elements of the trivial separator.  
\end{proof}
\end{lemma}

\section{Translation-invariant commuting Pauli Hamiltonians}\label{sec:pauli}

In this section we prove three main results for translation-invariant commuting Pauli Hamiltonians.
All our results will be algorithmic
where the termination of the algorithm will be guaranteed by results from commutative algebra.
We will use definitions from \cref{sec:introduction,separatorsection},
but the content here is independent of the previous sections.
In the last subsection here we prove lemmas on which the results of the previous sections rely.

The first main result is that there exists a disentangling QCA for a commuting Pauli Hamiltonian
if and only if the ground state is nondegenerate.
The forward implication is obvious,
but the backward implication is not.
Applying our general solution to a specific example,
we obtain the locally flippable separator for the three fermion Walker-Wang model in \cref{WWsection}.
The second main result is that any translation-invariant Clifford QCA $\alpha$ on qubits
squares to be trivial after stabilization,
i.e., $\alpha^2 \otimes \Id$ is a product of a Clifford circuit and a shift.
The third result is that in any one-dimensional translation-invariant
group of Pauli operators,
there exists a maximal abelian subgroup that is locally generated.
This result is applied to find a boundary Hamitonian of a two-dimensional commuting Pauli Hamiltonian
by considering the commutant of the bulk terms near a boundary.

To obtain these results,
we review and develop the perspective for translation-invariant Pauli algebras using Laurent polynomials~\cite{Haah2013}.
At a minimum level the machinery is simply a notation change;
it is a compact way of expressing data of Pauli algebras.
But it will become clear that this allows us to use powerful theorems from commutative algebra.

Let us clarify our convention.
By a {\bf Pauli operator} we mean a tensor product of Pauli matrices;
the identity operator is also considered as a Pauli operator.
Unless explicitly stated, 
we assume every Pauli operator has finitely many non-identity tensor components,
i.e., it has finite support.
By a {\bf commuting Pauli Hamiltonian}
we mean a local Hamiltonian whose terms are Pauli operators and are pairwise commuting.
In addition, we always assume that each term takes the minimum eigenvalue on a ground state; 
the Hamiltonian is assumed to be frustration-free.
The term ``stabilizer Hamiltonian'' is commonly used in literature in connection to ``stabilizer codes,''
but we use more descriptive terminology.

We will mostly directly work with infinite systems of finite dimensional degrees of freedom.
However, we will not worry about defining Hilbert spaces 
on these infinite number of degrees of freedom.
The main reason is that we will only be concerned with a group of Pauli operators,
which is an infinite direct sum of single-qudit Pauli groups,
but never their $\CC$-linear combinations.
Hence, we will not encounter any infinite sequences,
and thus the convergence is out of the question.
As we will only consider translation invariant groups,
it is always possible to transcribe everything into a statement on a sequence of finite periodic systems
and any finiteness will be transcribed into a uniform constant bound in this sequence of systems;
however, we will not do the transcription.

Unless explicitly stated otherwise,
any algebra is generated by finitely supported operators.
Hence, a commutant of an algebra is the set of all finitely supported operators 
that commute with every element of the given algebra.
Here, the verb ``generate'' follows the usage in algebra:
A group generated by a set of elements is the collection of all finite products of the elements.
An algebra generated by a set of elements is the collection of all finite linear combinations
of finite products of the elements.
In this section, we use $R$ to denote the base ring $\FF_p[x_1^\pm,\ldots, x_D^\pm]$;
in previous sections $R$ was the range of QCA.

\subsection{Introduction to polynomial methods}

This subsection is to review the machinery developed in Ref.~\cite{Haah2013},
which combines the emphasis of symplectic structure 
in Pauli stabilizer codes~\cite{CalderbankRainsShorEtAl1997Quantum}
and the polynomial representation of cyclic codes in classical error correction~\cite{MacWilliamsSloane1977}.
This allows us to use tools from homological algebra.
A reader should be able to find more explanations in a lecture note~\cite{Haah2016}
for any unproven statements here.

\paragraph{Stabilizer group.}

We begin with an example.
Consider a one-dimensional array of qubits (spin-$\frac 1 2$'s)
and a Hamiltonian
\begin{align}
H_0 = - \sum_{j} Z_{j-1} X_{j} Z_{j+1}.
\end{align}
This Hamiltonian is invariant under translation,
and consists of commuting Pauli operators of uniformly bounded support.
(This Hamiltonian is known as the ``cluster state'' Hamiltonian~\cite{BriegelRaussendorf}.)
The ground space is determined by a set of eigenvalue equations
\begin{align}
 Z_{j-1} X_j Z_{j+1} \ket \psi = \ket \psi.
\end{align}
Since $\ket \psi$ is stabilized by a set of operators that are pair-wise commuting,
it is natural to think of the ground space as a trivial representation space
of the \emph{group} $\stab$, called a {\bf stabilizer group},
 generated by all the Hamiltonian terms $Z_{j-1} X_j Z_{j+1}$ (without the minus sign).
We will almost always treat the stabilizer group as a whole, 
rather than focusing on individual elements or a certain generating set.

\paragraph{Pauli group.}

The stabilizer group $\stab$ is abelian and every element of $\stab$ squares to become the identity
(said to have \emph{exponent} 2).
Hence, the stabilizer group can be regarded as a vector space over the binary field.
However, this perspective on its own is not too useful since it does not reveal
how the stabilizer group sits in the full algebra of local operators.
The stabilizer group being comprised of Pauli operators
suggests that we should first investigate the multiplicative group $\pauligroup$ of local Pauli operators,
which we call {\bf Pauli group} $\pauligroup$.
The Pauli group has a special property that its commutator subgroup $[\pauligroup,\pauligroup]$,
which is generated by all elements of form $ghg^{-1}h^{-1}$,
consists of phase factors.
In addition any Pauli operator squares to become a phase factor.
Thus, factoring out the commutator subgroup ({\bf abelianization})
the resulting abelian group $\pauligroup/[\pauligroup,\pauligroup]$
``loses'' only phase factors and has exponent 2.
Hence, the abelianized Pauli group is a vector space over the binary field.
The smallest case is when the Pauli group consists of all single-qubit Pauli matrices,
and the abelianized Pauli group is $\FF_2^2$ because $Y = iXZ$.
We make an explicit convention in the correspondence:
\begin{align}
X \leftrightarrow \begin{pmatrix} 1 \\ 0 \end{pmatrix}, \qquad
Z \leftrightarrow \begin{pmatrix} 0 \\ 1 \end{pmatrix}
\end{align}

\paragraph{Symplectic form and generalization to qudits.}

Commutation relations in the Pauli group can still be kept track of in the abelianized Pauli group
by means of a bilinear form.
In $H_0$ two terms $Z_0 X_1 Z_2$ and $Z_1 X_2 Z_3$ commute because there are two tensor factors
that are anticommuting.
That is, the commutation relation is revealed by counting the number of anticommuting pairs 
of Pauli tensor components
\begin{align}
(Z_0, I_0) + (X_1, Z_1) + (Z_2,X_2) + (I_3, Z_3) = 0 + 1 + 1 + 0 = 0 \mod 2
\end{align}
where the pairings $(Z_0,I_0)$, etc. take value 0 if commuting or 1 if anti-commuting.
Under the convention of the previous paragraph
the pairing is identified with a bilinear pairing as promised:
\begin{align}
(P,Q) = p^T \begin{pmatrix} 0 & 1 \\ -1 & 0 \end{pmatrix} q \label{eq:sympform}
\end{align}
where $p,q$ are column vectors in the binary vector space 
corresponding to the Pauli operators $P,Q$, respectively.

In \cref{eq:sympform} we have included an apparently superfluous minus sign in the $\FF_2$-valued form.
This is to have the notation consistent in a generalization to qudits ($p$-state ``spin''s).
In this generalization the Pauli group is defined to be generated by
\begin{align}
  X = \sum_{j=0}^{p-1} \ket{j+1 \mod p} \bra j, \qquad Z = \sum_{j=0}^{p-1} \exp(2\pi i j / p) \ket j \bra j
\end{align}
and their tensor products. The bilinear pairing is justified by
\begin{align}
  X^a Z^b = \exp(2\pi i ab /p) Z^b X^a, \text{ or } Z^b X^a = \exp(-2\pi iab /p) X^a Z^b.
\end{align}

In summary,
the commutation relation in $n$-qubit Pauli group is given by
a bilinear form on the abelianized Pauli group, which is a vector space, 
whose matrix representation is
\begin{align}
\lambda_n = \begin{pmatrix} 0 & I_n \\ -I_n & 0 \end{pmatrix}. \label{eq:lambdaMatrix}
\end{align}
A stabilizer group on $n$-qubits then corresponds to a subspace and the bilinear form (in fact, symplectic)
restricted to this subspace must identically vanish.
If we represent the generators of a stabilizer group by the columns of a binary matrix $\sigma$,
then this matrix $\sigma$ has to satisfy an equation
\begin{align}
  \sigma^T \lambda_n \sigma = 0. \label{eq:ScalarvaluedSymplecticCondition}
\end{align}

\paragraph{Pauli module and stabilizer module.}

In an infinite lattice or in a family of finite-sized lattices,
the matrix $\sigma$ that represents a stabilizer group would be large.
Recalling the example in the beginning of this subsection,
we realize that the matrix $\sigma$ would have a cyclic structure due to the translation invariance,
which ought to be used to ``compress'' the matrix $\sigma$.

A natural way for the compression is to regard the vector space of the abelianized Pauli group
as a module over the translation group.
Let $R = \FF_2[x^\pm]$ be the translation group algebra that acts on the abelianized Pauli group.
The~``variable''~$x$ is the generator of the translation group in one dimension.
Since the Pauli operator $X_j$ acting on site $j$ is a translate of $X_0$,
if $X_0$ is represented by a vector $e_1$,
its translate $X_j$ must be represented by a vector $x^j e_1$.
Likewise, if $Z_0$ is represented by a vector $e_2$,
then $Z_j$ must be represented by $x^j e_2$.
Since any finitely supported Pauli operator is a finite product of these,
we conclude that the entire module of Pauli operators 
is identified with a free module $R e_1 \oplus R e_2 = R^2$.
A term in the example $H_0$ is then expressed as a member of the module $R^2$ as follows.
\begin{align}
\begin{array}{c|c}
\text{ Pauli operators on Hilbert space } & \text{ Vectors in the abelianized Pauli group } \\
\hline
Z_{j-1} & x^{j-1} \begin{pmatrix} 0 \\ 1 \end{pmatrix} \\
X_j     & x^j     \begin{pmatrix} 1 \\ 0 \end{pmatrix} \\
Z_{j+1} & x^{j+1} \begin{pmatrix} 0 \\ 1 \end{pmatrix} \\
\hline
\text{multiply} & \text{add} \\
\hline
Z_{j-1} X_j Z_{j+1} & x^j \begin{pmatrix} 1 \\ x + x^{-1} \end{pmatrix} \\
\end{array}\label{eq:egCluster}
\end{align}
This way the stabilizer group generated by a translation-invariant Hamiltonian $H_0$
is represented by a submodule of $R^2$ generated over $R$ 
by a single member (column vector) in the bottom right in \cref{eq:egCluster}.
Since the abelianized Pauli group $R^2$ is abelian (tautologically),
we have used additive notation in place of multiplicative one.
The appearance of Laurent polynomials is rooted in the fact that the translation group
is the multiplicative group of monomials $\{ x^j ~:~ j \in \ZZ\}$
where $x$ is an indeterminant.

\paragraph{Anti-hermitian form.}

This compact description of the stabilizer group is not yet complete,
since we have not adopted the symplectic form that captures the commutation relation
as some form over the \emph{Pauli module} $P = R^2$.
The solution is straightforward once we recall the fact that the symplectic form
is designed to count the number of anticommuting pairs of tensor factors.
In terms of Laurent polynomials,
a tensor factor of a Pauli operator is a term of a (two-component) Laurent polynomial,
and we need to count the number of terms of the same exponents between two Laurent polynomials, say $f(x)$ and $g(x)$.
A trick is to consider the product $f(1/x) g(x)$,
so that the terms of $f(x),g(x)$ of the same exponents would multiply to become 1,
and the constant term of $f(1/x) g(x)$ is precisely the sum of all of these.
Since the coefficient is in $\FF_2$,
the constant term of $f(1/x)g(x)$ is the parity of the number of common terms in $f(x)$ and $g(x)$.
Since common terms must be cross-counted between two-component vectors representing Pauli operators,
we conclude that the symplectic form value between two vectors $v,w$ is
the constant term (the coefficient of 1) in 
\begin{align}
v^\dagger \lambda_1 w .
\end{align}
Here $\dagger$ is the transpose followed by an involution, that we usually denote by a ``bar,''
of $R$ such that $x \mapsto \bar x = x^{-1}$.
In particular, 
\emph{
if $v^\dagger \lambda w$ vanishes identically as a Laurent polynomial,
then every translate of the Pauli operator of $v$ commutes with every translate of the Pauli operator of $w$.
}
Therefore, if the columns of a matrix $\sigma$ over $R$ generates the \emph{stabilizer module}
then $\sigma$ satisfies a matrix equation
\begin{align}
\sigma^\dagger \lambda \sigma = 0 .
\end{align}
Note the similarity and difference of this equation 
to \cref{eq:ScalarvaluedSymplecticCondition}.
With the example $H_0$ above we see
\begin{align}
\begin{pmatrix} x^{-1} + x & 1 \end{pmatrix}
\begin{pmatrix} 0 & 1 \\ -1 & 0 \end{pmatrix}
\begin{pmatrix} 1 \\ x + x^{-1} \end{pmatrix}  = 0
\end{align}
which encodes the fact that $H_0$ is a commuting Hamiltonian.

\paragraph{Larger unit cell.}

If a unit cell contains $q$ qubits rather than just a single qubit,
then the Pauli group (after abelianization)
is identified with a free module $R^{2q}$.
For instance as we will use below, 
if two qubits reside on each edge of a simple cubic lattice of three-dimensions,
then $q=6$.
The commutation relation have motivated us to introduce an {\bf anti-hermitian}%
\footnote{
A form that is linear in one argument and involution-linear in the other is called a {\bf sesquilinear form}.
A sesquilinear form is {\bf (anti)hermitian} if $(v,w) = \pm \overline{(w,v)}$.
The usual notion of being hermitian over complex vector space is a special case of this
where the involution is the complex conjugation.
With qubits ($p=2$) any anti-hermitian form is hermitian since the base field is of characteristic~2,
but we insist to call it anti-hermitian with qudit generalizations in mind.
}
form on $R^{2q}$ defined by
\begin{align}
R^{2q} \times R^{2q} \ni (v,w) \mapsto v^\dagger \lambda_q w \in R 
\end{align}
where $\lambda_q$ is defined in \cref{eq:lambdaMatrix}.

\paragraph{Higher dimensions.}

Higher dimensions pose no further complication.
We would need more formal variables to denote translations in all directions.
The base ring becomes $\FF_{p}[x_1^\pm, \ldots,x_D^\pm]$ for $D$-dimensional lattice,
where $p=2$ for a system of qubits.
Note that the perspective through modules is forgetful of fine structures of a lattice
such as whether it is triangular or simple cubic, but retains only the dimension and the unit cell size.
Also, the formal translation variables $x_1,\ldots,x_D$ 
are forgetful of the geometry of translation vectors,
and may correspond to nonorthogonal directions in a real space.

\paragraph{Exactness.}

In all situations of interest in this paper 
(and almost all cases where this polynomial method is used elsewhere),
except for one-dimensional cases where $D=1$,
the {\bf generating matrix} $\sigma$ of a stabilizer module
satisfies
\begin{align}
\ker \sigma^\dagger \lambda_q = \im \sigma \subset R^{2q} \label{eq:exactness}
\end{align}
where the kernel and image are over $R$.
Let us interpret this equation.
The kernel on the left-hand side consists of (abelianized) Pauli operators $v$ in $R^{2q}$
such that the anti-hermitian form value is zero with every $w$ in the stabilizer module.
That is, the kernel represents all operators that commute with every term in a Pauli Hamiltonian.
Since $R^{2q}$ represents operators of finite support,
we may say that the kernel is the collection of all \emph{local} Pauli operators that acts within the ground space.
The equation says every such Pauli operator must belong to the stabilizer group up to a forgotten phase factor,
and hence the action on the ground space must be a scalar multiplication by the forgotten phase factor.
We should say that a commuting Pauli Hamiltonian with \cref{eq:exactness} satisfied is {\bf topologically ordered}.

\paragraph{Conditions for exactness.}

The equation $\ker \sigma^\dagger \lambda_q = \im \sigma$ is easy to state,
but it may be nontrivial to check this condition for an explicit instance of $\sigma$.
This is a question that the present module perspective can give an algorithmic answer.
To explain criteria we need to introduce {\bf determinantal ideals $I_t$}.

Given a rectangular matrix $M$ over a ring $R$ and a nonnegative integer $t$,
consider the collection $I_t$ of all $t$-minors (the determinant of a $t$-by-$t$ submatrix)
and their $R$-linear combinations.
By convention we set $I_0 = R$.
If $t$ is larger than any of the matrix dimensions, then $I_t = 0$.
The cofactor expansion formula for the determinant implies that $I_{t+1} \subseteq I_t$,
so
\begin{align}
I_0 \supseteq I_1 \supseteq \cdots \supseteq I_r \supsetneq 0
\end{align}
for some $r \ge 1$.
The maximum integer $r$ is the {\bf rank} of the matrix $M$.
Observe the difference between this definition of the rank and that of a linear operator as the dimension of the image.
The ring $R$ is not a field in general, so we cannot speak of the dimension of the image;
however, using determinantal ideals we obtain a consistent generalization of the notion of rank, 
even when the dimension of the image is not defined.%
\footnote{
When the image is a free module, which is the case over a field,
this rank coincides with the rank of the image.
}

In our case $R$ is a Laurent polynomial ring,
and for any ideal $I$ of $R$ we may consider the set $V(I)$ of zeros of the ideal $I$.%
\footnote{
Formally, one should initially define from which set the zeros are seek.
However, for the present pedestrian treatment, one's intuitive polynomial solving is sufficient.
}
If the ideal is given by an explicit set of generators, which is the case for any determinantal ideal,
the set of zeros of the ideal is the same as the common zeros of the generators.
For example, if $M = \begin{pmatrix} x-1 & y-1 \end{pmatrix}$ over $R = \FF_p[x^\pm,y^\pm]$,
then
\begin{align}
I_0 &= R & V(I_0) &= \emptyset, \nonumber \\
I_1 &= (x-1,y-1) & V(I_1) &= \{(1,1)\}, \label{eq:eg-point} \\
I_2 &= 0 & V(I_2) &= \text{everything}. \nonumber
\end{align}
The {\bf height}
of a proper ideal $I \subseteq R = \FF_p[x_1^\pm,\ldots, x_D^\pm]$
is defined to be the difference between $D$ and the ``largest geometric dimension'' 
of the set of zeros of $I$.%
\footnote{
A proper definition is the minimum integer $\ell$ such that
$\ell$ is the maximum length of chain of primes 
starting with zero and ending with a prime ideal above $I$.
But, the intuitive geometric dimension hardly causes any confusion in our situation.
}
The set of zeros of $I$ may consists of e.g. a point and a line,
in which case the ``largest geometric dimension'' is 1, not 0.
If $I = R$, the height is defined to be $+\infty$.
In the example of \cref{eq:eg-point} the heights of $I_0,I_1,I_2$ are $\infty,2,0$, respectively.

The following lemma states conditions for the equality $\ker \sigma^\dagger \lambda_q = \im \sigma$.
The conditions are not complete, but are sufficient for our purpose later.
One can in fact completely determine whether $\ker \sigma^\dagger \lambda_q = \im \sigma$
using Gr\"obner basis, but we will not use this technique.
An interested reader is referred to a textbook chapter on Gr\"obner basis~\cite[Chap.~15]{Eisenbud}.

\begin{lemma}\label{lem:chaincondition}
Let $\sigma$ be a matrix over $R = \FF_p[x_1^\pm,\ldots, x_D^\pm]$
that satisfies $\sigma^\dagger \lambda_q \sigma = 0$.
If $\ker \sigma^\dagger \lambda_q = \im \sigma$,
then the rank of $\sigma$ is equal to $q$ and the height of $I_q(\sigma)$ is at least $2$.
If $\rank \sigma = q$ and $I_q(\sigma) = R$,
then $\ker \sigma^\dagger \lambda_q = \im \sigma$.
\end{lemma}
Note that the necessary conditions in the first claim are not sufficient conditions.
Here is a counterexample of a 4-by-4 matrix: 
\begin{align}
\sigma = \sigma_1 \oplus (\lambda_1 \bar \sigma_1) \text{ where } 
\sigma_1 = \begin{pmatrix} xy & x^2 \\ y^2 & xy \end{pmatrix} .
\end{align}
\begin{proof}
The full proof is well beyond the present exposition,
so we will be brief assuming some familiarity with homological algebra.
The first claim is a half of the Buchsbaum-Eisenbud theorem on the exactness of chain complexes~\cite{BuchsbaumEisenbud1973Exact}\cite[Thm.~20.9]{Eisenbud}, 
together with the Hilbert syzygy theorem on the existence of finite free resolution 
on $\ker \sigma$~\cite[Cor.~15.11]{Eisenbud}.
The Buchbaum-Eisenbud theorem will be used importantly,
so we quote the theorem here.
\begin{theorem}[Buchsbaum-Eisenbud~\cite{BuchsbaumEisenbud1973Exact}]
A chain complex 
\begin{align}
0 \xrightarrow{\varphi_{n+1} = 0} F_n \xrightarrow{\varphi_n} F_{n-1} \to \cdots \to F_1 \xrightarrow{\varphi_1} F_0
\end{align}
of finitely generated free modules over a Noetherian ring is exact,
if and only if $\rank F_k = \rank \varphi_{k+1} + \rank \varphi_k$ 
and $\mathrm{depth}~ I(\varphi_k) \ge k$ for $k =1,\ldots,n$.
\end{theorem}
The second claim of our lemma 
is the other half of the Buchsbaum-Eisenbud theorem together with 
Corollary~20.12 of Ref.~\cite{Eisenbud}.
The cited theorem~\cite[Cor.~20.12]{Eisenbud}
implies that any finite resolution of $\ker \sigma$ has connecting maps
with the determinantal ideal being unit.
That is, there exists a finite chain complex,
extending $R^q \xrightarrow{\sigma} R^{2q} \xrightarrow{\sigma^\dagger \lambda_q} R^q$ on the left, 
and ending with zero, where all determinantal ideals of the connecting maps are unit,
and the ranks of the maps going in and out sum to the rank of the module.
The unit determinantal ideal trivially satisfies the depth condition of 
the Buchsbaum-Eisenbud criterion for finite exact sequences,
and the rank condition is also satisfied.
Therefore, the extended chain complex is exact, and in particular is exact at $R^{2q}$
between $\sigma$ and $\sigma^\dagger \lambda_q$, which is the claim of the lemma.
\end{proof}
%%%%%%%%%%%%%%%%%%%

\paragraph{Taking a smaller translation group: the functor $\phi^{(n)}_\#$.}

Consider a ring homomorphism $\phi: S \to R$ between commutative rings
and a module $A$ over $R$.
With the help of the morphism $\phi$, we can say that $A$ is an $S$-module;
all we need to do is to define the ring action on the module,
but there is a canonical way by letting $s \in S$ act on $m \in A$ by 
$s \cdot m := \phi(s)m$.
Since $\phi$ is a ring homomorphism, all requirement for the ring action
is satisfied.

Suppose we have a $R$-linear module map $f: M \to N$ 
between $R$-modules $M$ and $N$.
If we regard $M$ and $N$ as $S$-module via $\phi$,
then a natural question is whether $f$ is $S$-linear.
The answer is yes; $f(s \cdot m) = f(\phi(s)m) = \phi(s) f(m) = s \cdot f(m)$.
Thus, $\phi$ induces a \emph{covariant functor} $\phi_\#$ from the category of $R$-modules
to the category of $S$-modules.
If $S$ is in fact isomorphic to $R$ as rings (but the morphism $\phi$ may not even be invertible),
we may regard $\phi_\#$ as a functor from the category of $R$-modules to itself,
and we may consider composition of two or more such functors.
Naturally, we write $\phi_\#(f)$ to denote the map $f$ viewed as an $S$-linear map.

If we have an $R$-linear map between free $R$-modules with a basis chosen,
the map can be written as a matrix $T$ with entries in $R$
and the map is given by matrix multiplication on a vector from the left.
The matrix element at $(a,b)$ of $T$
is a map from $b$-th direct summand $R$ to $a$-th direct summand $R$.
Under the functor $\phi_\#$, any summand $R$ is an $S$-module,
and every matrix element of $T$ 
must be described as a $S$-linear map $T_{ab}$ between $S$-modules.

If $R$ happens to be a free $S$-module (which will be true in our case),
$T_{ab}$ can be written as a matrix with entries in $S$,
and we obtain a concrete implementation of $\phi_\#(T)$.
Caution: for a bilinear (or sesquilinear) form $\Xi$ on a free $R$-module $\mathcal M$,
the matrix representation of $\Xi$ under $\phi_\#$ should be obtained,
not by interpreting $\Xi$ as a matrix (a linear map),
but by evaluating form values on a chosen $S$-basis of $S$-module $\phi_\#(\mathcal M)$;
however, see \cref{rem:FormAsMap} below.

In our case, we consider a ring homomorphism 
$\phi^{(n)} : S = \FF_p [y^{\pm}] \ni y \mapsto x^n \in R = \FF_p [x^\pm]$,
which is injective for any integer $n \ge 1$.
(Here we explain with one-dimensional group algebra for notational convenience, 
but generalizations to higher dimensions can be easily obtained by
$y_i \mapsto x_i^n$ for $i = 1,\ldots, D$.)
The ring $S$ is of course isomorphic to $R$,
but we interpret $S$ as a coarser translation group embedded in $R$
by $\phi^{(n)}$.
Physically, this only means that we do not use the largest translation group available,
but only a smaller translation group.
The translation variable~$y$ denotes the translation by $n$ units ($x^n$).
As the translation group gets smaller,
the unit cell gets larger,
and there are more data we need to provide to describe a module~$A$.
Indeed, if a basis for $R$-module $A$ has $m$ elements,
then under $\phi^{(n)}$, we need $nm$ basis elements for $S$-module $A$.

To explicitly find an $S$-basis from an $R$-basis,
we think of $A$ as the image of a $R$-linear map (or matrix) $B : R^m \to R^{2q}$.
Then, an $S$-basis for $A$ is going to be the columns of $\phi^{(n)}_\#(B)$.
Each direct summand $R$ of $R^m$ or $R^{2q}$ is a free $S$-module of rank $n$
with basis $\{1,x,x^2,\ldots,x^{n-1}\}$
and the finer translation variable $x$ acts on this $S$-module $R = S^n$ (on the left) as
\begin{align}
\phi^{(n)}_\# (x) = \begin{pmatrix} 
0 &   &  & y \\
1 & 0 &  &    \\
  & 1 &0 &   \\
  &   & \ddots & 
\end{pmatrix} , \label{eq:phiImpl}
\end{align}
a cyclic permutation matrix except for the top right corner.
Note that $\phi^{(n)}_\#$ is a homomorphism from $R$ into a matrix ring over $S$.
Hence, the prescription to write $\phi^{(n)}_\#(B)$ is to 
\emph{
write $B$ as a matrix with Laurent polynomial entries in $x$,
and then replace every Laurent polynomial entry $g(x)$ with a matrix 
$g\left( \phi^{(n)}_\#(x) \right)$.
}

\begin{remark}\label{rem:FormAsMap}
Let $\Xi$ be a bilinear or sesquilinear form over $R$.
A convenient feature of our chosen basis $\{1,x,x^2,\ldots,x^{n-1}\}$ of an $S$-module $R = S^n$
is that the explicit matrix representation of $\phi^{(n)}_\#(\Xi)$ follows the \emph{same prescription}
as if the matrix representation $\Xi$ were an $R$-linear map.
(The prescription would not be valid 
if we had chosen a different basis for the $S$-module $R=S^n$.)
Thus, we will always use this basis 
whenever we apply the functor $\phi^{(n)}_\#$.
In particular, under this basis choice, $\phi^{(n)}_\#(\lambda_q) = \lambda_{nq}$,
where $\lambda_q$ is the anti-hermitian form on a Pauli module $R^{2q}$.
\end{remark}

\subsection{Locally flippable separators from nondegeneracy}

We assume throughout that there are finitely many Hamiltonian terms up to translations;
otherwise, the sum of all terms supported on a disk would have unbounded norm.

\begin{theorem}\label{thm:ticp-disentangler}
Let $H$ be a translation-invariant commuting Pauli Hamiltonian (Pauli stabilizer Hamiltonian) 
in $D$-dimensional lattice with $q$ qudits of prime dimension $p$ per site.
Suppose that $H$ has a unique (nondegenerate) ground state 
on every finite lattice with periodic boundary conditions of any sufficiently large size.

Then, there exists a translation-invariant Clifford QCA 
that maps the ground state of $H$ into a product state.
Moreover, when $p=2$, if we assume further that
every term of $H$ is invariant under complex conjugation,
% the group generated by all terms of $H$ is invariant under complex conjugation,
then the Clifford QCA can be chosen so that it maps any real operator to a real operator.
\end{theorem}

This theorem will follow from the general construction of \cref{sec:disent}
after \cref{lem:separatorexists,lem:lf-Clif-separator} below,
but here we give a proof entirely within the polynomial framework.

We begin with a basic property.
\begin{lemma}\label{lem:nondegeneracy-exactness}
Suppose $H$ has a nondegenerate ground state 
on every finite periodic lattice of sufficiently large size.
If the column span of a matrix $\sigma$ over $R = \FF_p[x_1^\pm,\ldots,x_D^\pm]$ 
is the stabilizer module of $H$,
then $\ker \sigma^\dagger \lambda_q = \im \sigma$ and $I_q(\sigma) = R$.
\end{lemma}
\begin{proof}
Since $\ker \sigma^\dagger \lambda_q \supseteq \im \sigma$ by the assumption that $H$ is commuting,
we have to show $\ker \sigma^\dagger \lambda_q \subseteq \im \sigma$.
Once this is shown, the second claim $I_q(\sigma) = R$ follows from~\cite[Cor.~4.2]{Haah2013}
and \cref{lem:chaincondition}.

By the nondegeneracy assumption, any ``logical'' Pauli operator
(those that act within the ground space) on a periodic lattice of linear size $L$
must be a stabilizer and belong to $\im \sigma + \mathfrak{b}_L R^{2q}$, 
where 
\begin{align}
\mathfrak b _L = (x_1^L -1,\ldots,x_D^L -1 )
\end{align}
is an ideal of $R$ that imposes periodic boundary conditions in all $D$ directions.
Thus, we have
\begin{align}
\ker \sigma^\dagger \lambda_q \subseteq \im \sigma + \mathfrak{b}_L R^{2q}
\end{align}
for any $L$ sufficiently large.
We may replace $\mathfrak b _L$ by an ideal power, 
only to enlarge the module on the right-hand side as follows.
For any $n \ge 1$ consider $p^m \ge n$.
Then, $(x-1)^n | (x-1)^{p^m}$ but $(x-1)^{p^m} = x^{p^m} - 1$ 
due to the positive characteristic $p$ of $R$,
implying
$\mathfrak b_{p^m} \subseteq \mathfrak b_1^n$.%
\footnote{
The ideal power $\mathfrak b_1^n$ 
is the ideal generated by all $n$-fold products of elements of $\mathfrak b_1$.
See~\cite[Lem.~7.3]{Haah2013} with $N = R/\mathfrak b_1^n$ for a more general result.
}
Thus,
\begin{align}
\ker \sigma^\dagger \lambda_q  / \im \sigma \subseteq \bigcap_{n=1}^\infty \left( \mathfrak{b}_1^n R^{2q} / \im \sigma \right).
\end{align}
The Krull intersection theorem~\cite[Cor.~5.4]{Eisenbud} implies that the right-hand side is zero.
\end{proof}

The following result is not used elsewhere.
\begin{corollary}\label{cor:disentanglingQCA-is-manifesting}
Under the assumption of \cref{lem:nondegeneracy-exactness},
a translation-invariant Clifford QCA $Q$
maps the ground state into a product state $\bigotimes \ket 0$,
if and only if
$Q$ transforms the stabilizer group into the trivial stabilizer group $\langle Z_j \rangle$.
\end{corollary}
In particular, any Clifford circuit that disentangles the ground state of a commuting Pauli Hamiltonian
has to disentangle it manifestly.
This might not be true in general;
as we explain in \cref{discussionsection}
there is a subtle possibility that
the nondegenerate ground state of a general commuting Hamiltonian $H$
may be disentangled by some quantum circuit $U$ (not necessarily Clifford)
but $U^\dagger H U$ cannot be connected to the negative sum of the elements of a trivial separator
via some path in a space of commuting Hamiltonians.
This corollary says that this possibility is ruled out 
for translation invariant commuting Pauli Hamiltonians.
\begin{proof}
Certainly, any QCA that trivializes the stabilizer group disentangles the ground state.
Conversely, we have to show that for any product state
the stabilizer group is unique.
Suppose that a translation-invariant stabilizer group is chosen.
Consider any Pauli stabilizer $g$ of finite support for the product state.
It has to commute with any other Pauli stabilizer,
and by \cref{lem:nondegeneracy-exactness} the stabilizer $g$
belongs to the chosen stabilizer group.
That is, 
the stabilizer group is precisely the collection of all Pauli stabilizers of finite support.
\end{proof}

Now we construct a separator.
\begin{lemma}\label{lem:separatorexists}
Under the assumption of \cref{lem:nondegeneracy-exactness},
there exists a translation-invariant separator.
\end{lemma}
\begin{proof}
We will find a local generating set for the stabilizer group
such that there is no relation among the generators.
(A relation is a nonempty product of generators that results to the identity operator.
In terms of the vector representation of the Pauli operators over $\FF_p$,
no relation means the linear independence of the corresponding vectors over $\FF_p$.)
Then, from the theory of quantum stabilizer codes,
any common eigenspace of all generators is nonzero,
and is isomorphic to each other.
Since we assume the nondegeneracy of the ground state subspace,
which is a common eigenspace of all generators,
the generating set is a separator.
As we will work in the polynomial framework,
the translation-invariance will be automatic.

Algebraically, the no-relation condition amounts to the nullity
of the kernel of the matrix $\sigma$, whose image is the stabilizer module.
This nullity has to hold for every periodic finite lattice.%
\footnote{
For a nonexample, the toric code in two-dimensions on the square lattice
has
$
\sigma = \begin{pmatrix} x-1 & 0 \\ y-1 & 0 \\ 0 & -\bar y +1 \\ 0 & \bar x - 1 \end{pmatrix}
$
whose kernel is zero over $R = \FF_p[x^\pm,y^\pm]$ (infinite lattice),
but nonzero over $R/(x^L -1, y^L-1)$ ($L$-by-$L$ periodic lattice).
}

Consider any finite free resolution of $\ker \sigma^\dagger \lambda_q$, 
which is equal to $\im \sigma$ by \cref{lem:nondegeneracy-exactness},
of the minimum length $n$:
\begin{align}
0 \to F_n \to \cdots \to F_3 \xrightarrow{T} F_2 \xrightarrow{\sigma'} R^{2q} \xrightarrow{\sigma^\dagger \lambda_q} R^t
\end{align}
where every $F_j$ is a finitely generated free $R$-module
and $t$ is the number of Hamiltonian terms up to translations.
By construction, $\im \sigma' = \ker \sigma^\dagger \lambda_q = \im \sigma$.
We claim that $T = 0$ and $n=2$ for $n$ to be the minimum.
This will establish that there exists a generating set 
for the stabilizer module with no relation on the infinite lattice; 
we will treat finite periodic lattices shortly.

Since the determinantal ideal of $\sigma$ (or that of $\sigma^\dagger \lambda_q$)
is unit by \cref{lem:nondegeneracy-exactness},
the determinantal ideals of $\sigma'$ and $T$ 
must be all unit as well~\cite[Cor.~20.12]{Eisenbud}.
It follows that $\coker T$ is locally free with respect to 
any maximal ideal of $R$~\cite[Prop.~20.8]{Eisenbud}.
The Quillen-Suslin-Swan theorem~\cite{Suslin1977Stability,Swan1978}
says that $\coker T$ is in fact free.
In elementary terms, this means that there exist invertible matrices $A$ and $B$ over $R$
such that $ATB = \begin{pmatrix}I & 0 \\ 0 & 0 \end{pmatrix}$.%
\footnote{The results of Ref.~\cite{ParkWoodburn1995Algorithmic} 
gives an algorithmic proof of the existence of $A$ and $B$.}
Therefore, the we can choose $\sigma'$ such that it is kernel free.
Due to the minimality of the resolution, we must have $n=2$ and $T = 0$.

Finally, it remains to show that $\ker \sigma'$ is still zero over $R/\mathfrak b_L$ for any $L \ge 1$.%
\footnote{Formally, this amounts to examine chain complexes after tensoring an $R$-module
$R/ \mathfrak b_L$.}
This is another consequence of the Quillen-Suslin-Swan theorem applied to $\coker \sigma'$.
We know there are invertible matrices $A'$ and $B'$ such that $A' \sigma' B'$ is diagonal
with the identity matrix on the main diagonal. 
$B'$ being invertible remains true over $R/\mathfrak b _L$, 
and therefore $\sigma'$ has zero kernel over $R/\mathfrak b _L$.
\end{proof}

\begin{lemma}\label{lem:lf-Clif-separator}
Any separator of \cref{lem:separatorexists} is locally flippable.
Moreover, the local flipper is translation-invariant.
\end{lemma}
\begin{proof}
In \cref{lem:separatorexists} we have proved existence of a free basis $\sigma$ 
(as a collection of columns)
for the stabilizer module under the nondegeneracy assumption.
The collection of all translates of Pauli operators represented as the columns of $\sigma$
is a separator.
Also we know by \cref{lem:nondegeneracy-exactness} that 
$\ker \sigma^\dagger \lambda_q = \im \sigma$.
Thus, we have an exact sequence
\begin{align}
0 \to R^q \xrightarrow{\sigma} R^{2q} \xrightarrow{\sigma^\dagger \lambda_q} R^q 
\end{align}
where we have used \cref{lem:chaincondition} 
to fix the rank $q$ of the first and the third free module from the right.
This rank is equal to the number of generators for the stabilizer module 
(or equivalently the stabilizer group up to translations).
Due to Quillen-Suslin-Swan theorem~\cite{Suslin1977Stability,Swan1978}, $\coker \sigma^\dagger \lambda_q$ has to be free,
but since $R^q$ has rank $q$, so does $\sigma^\dagger$, it follows that $\coker \sigma^\dagger \lambda_q = 0$.
This means that there exists a vector $v_a$ such that $\sigma^\dagger \lambda_q v_a = e_a$ where
$e_a$ is the unit column vector with the sole 1 at $a$-th component.
Translating the Pauli operator corresponding to $v_a$,
we see that any single stabilizer generator can be flipped alone.
The collection of all $v_a$ and their translates form a local flipper,
and it is translation invariant.
\end{proof}

\begin{proof}[Proof of \cref{thm:ticp-disentangler}]
We defer the construction of the real QCA to the last part of the proof.

By \cref{lem:separatorexists,lem:lf-Clif-separator} 
there exists a locally flippable separator,
which can be regarded as a free basis for the stabilizer module,
represented as the columns of a $2q \times q$ Laurent polynomial matrix $\sigma$.
By construction, $\sigma$ satisfies $\sigma^\dagger \lambda_q \sigma = 0$.
Let $F$ be a matrix representing a local flipper,
i.e., a $2q \times q$ Laurent polynomial matrix such that
\begin{align}
F^\dagger \lambda_q \sigma = I_q, \quad \text{ or } \quad \sigma^\dagger \lambda_q F = -I_q. \label{eq:flipping}
\end{align}
We will modify $F$ by $\sigma$ to define $T = F - \sigma E$ for some matrix $E$,
so that
\begin{align}
T^\dagger \lambda_q \sigma &= I_q, \nonumber \\
T^\dagger \lambda_q T &= 0 .
\end{align}
The first equation holds for any $E$,
but the second equation is satisfied only by some $E$.
Once such $T$ is constructed,
the promised Clifford QCA is easy.
Define $Q$ by adjoining the columns of $\sigma$ on the right to $T$
\begin{align}
Q = \begin{pmatrix} T & \sigma \end{pmatrix} .
\end{align}
Then, it follows that
\begin{align}
Q^\dagger \lambda_q Q &= \lambda_q, \\
- \lambda_q Q^\dagger \lambda_q &= Q^{-1} . \nonumber
\end{align}
Both $Q$ and $Q^{-1}$ define automorphisms of the Pauli module 
and preserve the anti-hermitian form that encodes commutation relations.
Thus, these define automorphisms of the Pauli group, 
and hence of the quasi-local algebra of operators on a complex Hilbert space,
up to conjugations by Pauli operators~\cite[Prop.~2.2]{Haah2013}.
Clearly, the automorphism $Q^{-1}$ maps a separator into the trivial separator.

Now we find $E$.
The equation we have to solve is $F^\dagger \lambda_q F - E + E^\dagger = 0$;
that is, 
\begin{align}
E - E^\dagger = F^\dagger \lambda_q F. \label{eq:EEFF}
\end{align}
If $p \neq 2$ (characteristic not 2),
this equation is trivial to solve by setting $E = \frac 1 2 F^\dagger \lambda_q F = - E^\dagger$.
If $p = 2$ (or any other prime),
we observe that the diagonal of $M = F^\dagger \lambda_q F$ 
does not have any nonzero ``constant'' term and consists of anti-hermitian entries
by the virtue of~$\lambda_q$ being symplectic.
Thus, we can choose~$E$ to be the collection of all upper triangular submatrix of~$M$ excluding the diagonal,
together with a ``half'' of the terms from the diagonal of~$M$.
The ``half'' can be chosen to be the sum of all terms $c x_1^{a_1} \cdots x_D^{a_D}$ 
with $c \in \FF_p$ and $a_j \in \ZZ$
such that the real number $\sum_j a_j \pi^j$ is positive.
Here, $\pi = 4 \arctan(1) \approx 3.14$ is an arbitrarily chosen real number
that is transcendental over rational numbers.
This completes the proof for the first claim.

It remains to specialize the construction for a ``real'' stabilizer group with $p=2$.
The terms of $H$ being real is a property of the group generated by the terms,
since the complex conjugation is an automorphism of the ring of operators.
(The complex conjugation is not an automorphism of a complex algebra, 
and is excluded under our definition of QCA.)
Hence, the locally flippable separator constructed 
in \cref{lem:separatorexists} consists of real operators.
We have to construct a real local flipper.

%We work in the basis where $X$ is a permutation matrix of real entries, 
%i.e., in the so-called ``computational basis'' where $X = \sum_{j \in \FF_p} \ket{j+1}\bra j$.
%Note that $Z = \sum_{j \in \FF_p} \exp(2\pi i j / p) \ket j \bra j = (Z^*)^\dagger$.
%If $p > 2$, 
%then for any stabilizer $O = \omega \prod X \prod Z$ where $\omega$ is some phase factor,
%the product $OO^*$ has no $Z$ factors, and $X$ factors get squared.
%Since $2 \in \FF_p$ has multiplicative inverse,
%we see that the stabilizer group is in fact a direct sum of two subgroups,
%one is generated by products of $X$'s and the other by products of $Z$'s.
%(One may call the stabilizer group ``CSS''-type~\cite{CalderbankShor1996Good,Steane1996Multiple}.)
%Hence, each local flipper can be chosen to consist of products of either $X$'s or $Z$'s, unmixed.

Let $O^*$ denote the complex conjugation of any operator $O$ on a complex Hilbert space.
Among $I,X,Y,Z$, only $Y = - Y^*$ is nonreal,
and in a tensor product of these matrices,
the parity of the number of $Y$ determines whether the tensor product is real.
To count the number of tensor factors $Y$,
we use a similar trick as we counted the number of distinct Pauli components:
When $v  = \begin{pmatrix} v_X \\ v_Z \end{pmatrix}$
is a column vector of Laurent polynomials representing a Pauli operator $O$,
it is easy to see that the coefficient of $1 = x_1^0 \cdots x_D^0 \in \FF_2[x_1^\pm,\ldots,x_D^\pm]$
in $v_X^\dagger v_Z$, which we denote as 
\begin{align}
\Coe: R \ni v_X^\dagger v_Z \mapsto \Coe(v_X^\dagger v_Z) \in \FF_2 ,
\end{align}
is the number of tensor factors $Y$ in $O$.
Hence, the reality condition of the flipper $T = \begin{pmatrix} T_X \\ T_Z \end{pmatrix} = F - \sigma E$
is cast into an equation $\diag(\Coe(T_X^\dagger T_Z)) = 0$, or
\begin{align}
\diag \Coe ( F_X^\dagger F_Z + E^\dagger \sigma_X^\dagger \sigma_Z E + E^\dagger \sigma_X^\dagger F_Z + F_X^\dagger \sigma_Z E ) = 0.
\end{align}
Since the reality of $H$ is a property of the group generated by the terms of $H$,
we know $\diag \Coe (E^\dagger \sigma_X^\dagger \sigma_Z E) = 0$.
Using \cref{eq:flipping},
we know $F_X^\dagger \sigma_Z = F_Z^\dagger \sigma_X + I$ and
this equation becomes
\begin{align}
\diag \Coe( F_X^\dagger F_Z + E + E^\dagger \sigma_X^\dagger F_Z + F_Z^\dagger \sigma_X E ) = 0.
\end{align}
The last two terms sum to zero since $\diag \Coe$ is unchanged under $\dagger$.
Hence, we are left with
\begin{align}
\diag\Coe(F_X^\dagger F_Z + E) = 0. \label{eq:FFE}
\end{align}
This has to be simultaneously solved with \cref{eq:EEFF}.
But \cref{eq:EEFF} allows arbitrary diagonal $\FF_2$ entries in $E$,
so \cref{eq:FFE} can always be satisfied by some $E$.

We thus have shown that $Q$ can be chosen to preserve reality;
a QCA given its polynomial representation $Q$ is unique up to
Pauli conjugations, which always preserve reality.
Finally, we show that if $Q$ is reality-preserving, so does $Q^{-1}$,
which will complete the proof.
This is easy: In the complex operator algebra,
if a real QCA $\alpha$ maps $O = \eta O'$ to a real operator,
where $\eta \in \CC$ is a phase factor and $O'$ is real, 
then the equation $\alpha(\eta O') = \eta \alpha(O')$ implies $\eta$ is real.%
\footnote{
If one insists not to leave the polynomial framework,
then one can think of a quadratic form on an infinite dimensional $\FF_2$-vector space
defined by
\[
R^{2q} \ni v \mapsto \Coe\left[ v^\dagger \begin{pmatrix}
0 & I_q \\ 0 & 0 \end{pmatrix} v \right] \in \FF_2
\]
whose polar form is the scalar-valued symplectic form~\cite{Kniga}.
This $\FF_2$-valued quadratic form captures the parity of the number of $Y$ components.
The constructed symplectic transformation $Q$ is an isometry of this quadratic form.
}
\end{proof}

\subsection{Exponent of qubit Clifford QCAs}

\begin{theorem}\label{thm:CQCAsqIsCircuit}
Let $D \ge 0$ be any integer, and
$\alpha$ be any translation-invariant Clifford QCA in $D$-dimensional lattice 
comprised of qubits (i.e., $p=2$).
Then, $\alpha^2 \otimes \Id$ is a Clifford circuit up to shifts.
\end{theorem}

Before presenting a proof of the theorem,
we briefly review how translation-invariant Clifford circuits 
are represented in the polynomial framework.
Since any translation-invariant Clifford QCA is determined
by the image of basis elements $X_j, Z_j$,
the matrix $Q$ representing a Clifford QCA $\alpha$
is one that has the polynomial representations of Pauli operators 
$\alpha(X_j)$ and $\alpha(Z_j)$
in its columns.
Since a QCA must preserve the commutation relation,
the matrix $Q$ satisfies $Q^\dagger \lambda_q Q = \lambda_q$
where $q$ is the number of qudits per site.
That is, $Q$ is {\bf symplectic}.%
\footnote{
Although an automorphism of a sesquilinear form is usually called \emph{unitary},
here we instead call it \emph{symplectic}.
This is to keep the terminology consistent for the cases $D=0$ and $D \ge 1$ (the lattice dimension),
and to disambiguate automorphisms on Pauli modules from those on complex Hilbert spaces.
}
This mapping $\alpha \mapsto Q$ is not injective 
inasmuch as the Pauli module is forgetful of the phase factor.
However, $Q$ determines $\alpha$ up to a conjugation action by 
a Pauli operator of possibly infinite support~\cite[Prop.~2.2]{Haah2013}.

Elementary Clifford circuits 
--- the control-NOT gate, the Hadamard gate, and the phase gate,
(and one more that maps $Z$ to $Z^m$ for $m \neq 0 \mod p$ if $p > 2$) ---
that are translation-invariant
induce {\bf elementary symplectic transformations} on the Pauli module.
They are represented by elementary row operations as listed:
\begin{align}
 \left[ E_{i,j}(a) \right]_{\mu \nu} &= \delta_{\mu \nu} + \delta_{\mu i} \delta_{\nu j} a 
 & \text{ where $\delta$ is the Kronecker delta,} \nonumber\\
 \text{Hadamard: } \quad&E_{i,i+q}(-1) E_{i+q,i}(1) E_{i,i+q}(-1) &\text{ where } 1 \le i \le q,\nonumber\\
 \text{control-Phase: }\quad & E_{i+q,i}(f)&\text{ where }f = \bar f, 1 \le i \le q,\\
 \text{control-NOT: } \quad & E_{i,j}(a) E_{j+q,i+q}(-\bar a) &\text{ where }1 \le i \ne j \le q, \nonumber \\
\text{extra gate for }p \neq 2: \quad & E_{i,i}(a-1) E_{i+q,i+q}(a^{-1}-1) & \text{ where } a \in \FF_p^\times, 1 \le i \le q. \nonumber
\end{align}
In particular, from CNOT we see that any elementary row operation on the upper half block
can be compensated by an elementary row operation on the lower half block
to become an elementary symplectic transformation.
This remark is useful when we determine whether a symplectic transformation is a composition of CNOTs.
The appearance of $-1$ in the extra gate for $p \neq 2$ is 
just to cancel $\delta_{\mu=i, \nu=j}$ in the definition of $E_{i,j}(a)$.

On the other hand, a shift is given by
\begin{align}
\text{shift:} \quad & E_{i,i}(x_1^{s_1} \cdots x_D^{s_D} - 1 ) E_{i+q,i+q}(x_1^{s_1} \cdots x_D^{s_D} - 1 ) 
&\text{ where } s_j \in \ZZ, 1 \le i \le q.
\end{align}
Note that every elementary symplectic transformation has determinant~1,
while a shift has determinant a square of a monomial.

\begin{lemma}\label{lem:elemCphase}
Over $\FF_p[x_1^\pm,\ldots,x_D^\pm]$ where $p$ is any prime, if
\begin{align}
Q = \begin{pmatrix}
I_q & B \\
0 & I_q
\end{pmatrix}
\end{align}
satisfies $Q^\dagger \lambda_q Q = \lambda_q$, then $Q$ is a product of elementary symplectic transformations.
\end{lemma}
This lemma characterizes the difference between two translation-invariant Clifford QCAs, 
say $\alpha$ and $\beta$,
that map a trivial separator to the same separator.
Indeed, $\alpha \circ \beta^{-1}$ maps the trivial separator to itself,
and hence has a symplectic matrix representation $Q$ in the supposition.
The lemma says that $\alpha \circ \beta^{-1}$ is a Clifford circuit.
\begin{proof}
The assumption is that $B = B^\dagger$.
In particular, the diagonal of $B$ is hermitian 
(i.e., $B^\dagger = B$ with $\dagger$ being the transpose followed by monomial inversion),
and can hence be canceled by the control-Phase gates multiplying on the left of $Q$.
Therefore, we may assume $B$ has zero diagonal.
After Hadamard on the $q$-th qudit (the bottom row of $B$),
the lower right block has zero column on the right,
but $(q-1)\times (q-1)$ identity matrix on the upper $q-1$ rows, 
which can be used to eliminate, by the control-NOT, the bottom that is copied from $B$.
This control-NOT does not alter the left half of the post-Hadamard $Q$.
Applying the inverse Hadamard on the $q$-th qudit,
we again have identity on the full diagonal of $Q$,
but the upper right block $B'$ has the zero bottom row.
By the condition $B' = B'^\dagger$, the rightmost column of $B'$ must be zero as well.
By induction in $q$, the proof is complete.
\end{proof}

\begin{proof}[Proof of \cref{thm:CQCAsqIsCircuit}]
For a technical reason later in this proof,
we assume that there are at least two qubits per site ($q \ge 2$).
This is easily satisfied by considering a smaller translation group if necessary.

Note that $(\alpha \otimes \Id)\circ \Swap \circ(\alpha^{-1} \otimes \Id)$ is a Clifford circuit.
To see this, if suffices to consider one gate $W$ in $\Swap$.
Conjugating a gate $W$ that swaps a pair of qubits is a Clifford gate, by $\alpha \otimes \Id$,
we obtain a Clifford QCA that is nonidentity only on a ball of radius $3r$,
where $r$ is the range of $\alpha$.
Such a local Clifford QCA can be written as a Clifford circuit of depth $O(r^{2D}) = O(1)$.
Therefore,
\begin{align}
\alpha \otimes \alpha^{-1} = (\alpha \otimes \Id) \circ\Swap \circ(\alpha^{-1} \otimes \Id) \circ\Swap
\end{align}
is a Clifford circuit, and we will show that $\alpha \otimes \alpha$ is a Clifford circuit.
It suffices to consider the symplectic transformation
on the Pauli module induced by $\alpha \otimes \alpha$,
and show that the induced symplectic transformation is a product of elementary symplectic transformations.
Suppose $\alpha$ induces the symplectic matrix 
\begin{align}
Q = \begin{pmatrix} A & B \\ C & D \end{pmatrix}
\end{align}
on the Pauli module
where $A,B,C,D$ are $q \times q$ Laurent polynomial matrices.
By definition, $Q$ satisfies $Q^\dagger \lambda_q Q = \lambda_q$.
In particular,
\begin{align}
&-\lambda_q Q^\dagger \lambda_q 
\begin{pmatrix} A & B \\ C & D \end{pmatrix} 
=
\begin{pmatrix}
I_q & 0 \\ 0 & I_q
\end{pmatrix},\\
%%%%%%%%%%%%%%%%%%%%%%%%%%%
&\begin{pmatrix}
-\lambda_q Q^\dagger \lambda_q & 0 \\ 0 & Q^\dagger
\end{pmatrix}^\dagger
\lambda_{2q}
\begin{pmatrix}
-\lambda_q Q^\dagger \lambda_q & 0 \\ 0 & Q^\dagger
\end{pmatrix}
=
\lambda_{2q} .\label{eq:ESp-through-Suslin}
\end{align}
We claim that
\begin{align}
Q \oplus Q 
=
\begin{pmatrix} 
A & 0 & B & 0 \\
0 & A & 0 & B \\
C & 0 & D & 0 \\
0 & C & 0 & D
\end{pmatrix}
\end{align}
is a product of the elementary symplectic matrices.
Our representation of $Q \oplus Q$ here is to retain the convention
that the symplectic matrix $\lambda$ for the doubled system, on which $Q\oplus Q$ acts, 
is $\lambda_{2q}$ as defined in \cref{eq:lambdaMatrix}.

A crucial fact we use is Suslin's stability theorem~\cite{Suslin1977Stability},
which shows that the special linear group $SL(n, R)$ 
is generated by elementary matrices if $n \ge 3$ and
$R$ is a Laurent polynomial ring over a field.
We use this theorem as follows. The matrix $Q$ belongs to $GL(2q,R)$.
The determinant of $Q$ has to be a monomial $m$ because only a monomial is invertible in $R$.
By multiplying $m^{-1}$ to, say, the first row of $Q$,
we obtain a member $Q'$ of $SL(2q,R)$.
By Suslin's stability theorem we conclude that
$\begin{pmatrix}
-\lambda_q Q'^\dagger \lambda_q & 0 \\ 0 & Q'^\dagger
\end{pmatrix}$
is a product of elementary symplectic transformation,
representing a circuit of CNOTs.
In other words, the symplectic transformation of \cref{eq:ESp-through-Suslin}
corresponds to a Clifford circuit up to shifts.

Now we complement the stability theorem 
with a few more elementary symplectic transformations on the \emph{left} of $Q\oplus Q$.
We focus on the left half block;
the right half block will be treated through \cref{lem:elemCphase} later.
Working over a ring of characteristic 2,
we can ignore $\pm$ signs.
\begin{align}
&\begin{pmatrix} 
A & 0  \\
0 & A  \\
C & 0  \\
0 & C 
\end{pmatrix}
\xrightarrow{CNOT_{2 \to 1}}%%%%%%%%%%%%
\begin{pmatrix} 
A & A  \\
0 & A  \\
C & 0  \\
C & C 
\end{pmatrix}
\xrightarrow{Hadamard_{2}}%%%%%%%%%%%%%%%%%
\begin{pmatrix}
A & A \\
C & C  \\
C & 0  \\
0 & A
\end{pmatrix} 
\xrightarrow{
\begin{pmatrix}
\lambda_q Q^\dagger \lambda_q & 0 \\ 0 & Q^\dagger
\end{pmatrix}
}%%%%%%%%%%%%%%%%%%%%
\begin{pmatrix}
I_q & I_q \\
0   &   0 \\
G &    G \\
I_q+F & F
\end{pmatrix} \nonumber \\
&
\xrightarrow{Hadamard_2}%%%%%%%%%%%%%%%%%%
\begin{pmatrix}
I_q &  I_q \\
I_q + F & F  \\
G &   G \\
0   &  0 
\end{pmatrix}
\Longrightarrow %%%%%%%%%%%%%%
\begin{pmatrix}
I_q &  I_q \\
0 & I_q  \\
0 &   0 \\
0 &  0 
\end{pmatrix}
\xrightarrow{CNOT_{2 \to 1}} %%%%%%%%%%%%%
\begin{pmatrix}
I_q &  0 \\
0 & I_q  \\
0 &   0 \\
0 &  0 
\end{pmatrix} \label{eq:sigmaManipulation}
\end{align}
where 
\begin{align}
\begin{pmatrix} F \\ G \end{pmatrix} = \lambda_q Q^\dagger \lambda_q \begin{pmatrix} A \\ 0 \end{pmatrix} .
\end{align}
The arrow $\Longrightarrow$ needs to be explained.
The block $I_q + F$ can be eliminated by CNOTs since we have $I_q$ in the upper left corner.
The block $G$ has to be hermitian 
since $(Q \oplus Q)^\dagger \lambda_{2q} (Q \oplus Q) = \lambda_{2q}$,
and hence $G$ can also be eliminated by \cref{lem:elemCphase}.
%Thus, this $4q \times 2q$ block has been transformed into a diagonal matrix with $I_{2q}$ 
%in the upper half block.

The full symplectic matrix after the transformations above,
including the right half block that we neglected in \cref{eq:sigmaManipulation},
has to be $\begin{pmatrix} I_{2q} & J \\ 0 & I_{2q} \end{pmatrix}$
for some $J = J^\dagger$ due to the symplectic structure.
By \cref{lem:elemCphase} once again,
this matrix is elementary.
\end{proof}

\subsection{Gapped boundaries of 2D translation-invariant commuting Pauli Hamiltonians}

Given a commuting Pauli Hamiltonian in two dimensions which is obviously gapped,
can we introduce boundary terms such that that the overall Hamiltonian is still commuting,
and no more boundary terms can be defined without ruining commutativity and locality?
Certainly, one can start by investigating all possible terms that are supported 
on any interval of length, say, 1 that commutes with the bulk term.
After choosing some maximally commuting subset of all these possible terms,
one can examine possible terms on all intervals of length 2,
and iterate the process.
We should then ask if this process would end;
that is, after we examine boundary terms on intervals of length $n$ for some finite $n$,
will any other operator that commutes with already chosen terms 
be redundant?
The following theorem guarantees that there is a good choice of boundary terms,
at least in the special case of translation-invariant commuting Pauli Hamiltonians in two dimensions.

\begin{theorem}\label{thm:maximal-commutative-algebra}
Let $G$ be a translation-invariant subgroup of the Pauli group in one dimension.
There exists an abelian subgroup $\stab$ of the commutant of $G$ 
with a local translation invariant generating set consisting of Pauli operators
(under some translation group which may be a proper subgroup of the original translation group),
such that any Pauli operator of bounded support that commutes with $\stab$ and $G$ belongs to $\stab$.
\end{theorem}

We think of $G$ as the bulk terms near a boundary.
Then, the local generating set of $\stab$ qualifies to be a good set of boundary terms.

We will prove this theorem by going through the polynomial framework.
The transcription of the theorem into the polynomial framework is as follows.
Recall that the anti-hermitian form over the translation group algebra $R$ encodes 
whether two Pauli operators represented by $v$ and $w$ commute,
by the coefficient of ``$1$'' in the pairing $v^\dagger \lambda_q w$.
We distinguish the anti-hermitian form that is valued in $R$,
from this {\bf scalar-valued symplectic form} that is defined to be the coefficient of $1$ in $v^\dagger \lambda_q w$.
The distinction is only necessary as we consider smaller translation groups than the full one.

\begin{theorem}[Equivalent to \cref{thm:maximal-commutative-algebra}]\label{thm:mca-moduleversion}
Let $R = \FF_p[x^\pm]$ be the translation group algebra in one dimension,
and let $R_b = \FF_p[x^{\pm b}]$ 
denote a subring of $R$ parametrized by an integer $b \ge 1$.
Let $A$ be any $R$-submodule of a finitely generated Pauli module.
Then, there exist an integer $b \ge 1$ 
and a free $R_b$-submodule $\stab$ of $A^\perp$ such that
the orthogonal complement $\stab^\perp$ within $A^\perp$ is equal to $\stab$,
where the orthogonal complement is 
with respect to the scalar-valued symplectic form.
\end{theorem}

It is important that the module $\stab$ is over a \emph{subring} $R_b$ 
that is in general not equal to $R$.
For example,
consider an $R$-module $A$ 
generated by $v_0 = \begin{pmatrix}
1 \\ x
\end{pmatrix} \in R^2$.
This represents the \emph{nonabelian} group generated by Pauli operators $X_j Z_{j+1}$.
No nonzero $R$-submodule of $A$ can have the anti-hermitian form vanishing on it;
that is, every nontrivial subgroup that is translation-invariant 
with respect to the original translation group is nonabelian.
To see this, suppose $v \in A \setminus \{ 0 \}$.
Then, $v = r v_0$ for some $r \in R$
and $v^\dagger \lambda_1 v = \bar r r (x - \bar x)$.
If this is to vanish as a polynomial, 
we must have $r = 0$.
However, the $R_2$-submodule generated by $v_0$,
which corresponds to the group generated by $X_{2j} Z_{2j+1}$,
is a maximal submodule on which the scalar-valued symplectic form vanishes.

The rest of this section constitutes the proof of \cref{thm:mca-moduleversion}.
Our proof will rely heavily on the fact that we are working in one dimension,
which implies that $R = \FF_p[x^\pm]$ is a polynomial ring with \emph{one} variable.

\subsubsection{Free basis and Matrix of commutation relations}

The ring $R$ is a Euclidean domain;
if we define the {\bf (absolute) degree} of a Laurent polynomial as
the difference of the highest exponent to the lowest exponent,
then the long division gives, for any two Laurent polynomials $f$ and $g$,
an equation $f = q g + r$ for some $q, r \in R$ with $\deg r < \deg g$.
In particular, $R$ is a principal ideal domain,
and hence $A$ is free and has a basis.
(To see this, one considers the Smith normal form of a matrix
whose column span over $R$ is $A$. 
Any submodule of a finitely generated free module 
over a principal ideal domain is free.)
Given a basis of $A$ written in the columns of a matrix $B$,
we can capture the commutation relations
among our generators (up to translations) of $G$, 
a subgroup of the Pauli group, by the matrix
\begin{align}
\Xi = B^\dagger \lambda_q B  = - \Xi^\dagger.
\end{align}
Formally, $\Xi$ represents the anti-hermitian form
restricted to $A$.
Note that the diagonal elements of~$\Xi$ may be nonzero.

The following lemma characterizes $\Xi$.
\begin{lemma}\label{lem:structureXi}
  A matrix $\Xi$ over $R = \FF_p[x^\pm]$ is equal to $B^\dagger \lambda_q B$
  for some $q$ and $B$
  if and only if 
  (i) $\Xi = - \Xi^\dagger$, and
  (ii) any diagonal element $\Xi_{jj}$ has zero constant term, 
  i.e., the coefficient of $x^0 = 1$ is zero.
\end{lemma}
  The condition (ii) in the lemma is redundant if $p \neq 2$.
\begin{proof}
  $(\Rightarrow)$ (i) is obvious, and (ii) is because 
  $\Xi_{jj} = v^\dagger \lambda_q  v = \bar a b - a \bar b$
  where $a$ and $b$ are upper $q$ and lower $q$ components of $v$, respectively.
  $(\Leftarrow)$ By (i) we know $\Xi_{jk} = - \bar \Xi_{kj}$, and by (ii)
  any diagonal $\Xi_{jj} = r_j - \bar r_j$ for some $r_j$, 
  which may be chosen to be all terms with positive exponents.
  Set $q$ be the number of rows (or equivalently columns) of $\Xi$,
  and define a upper triangular matrix $B$ as $B_{jk} = \Xi_{jk}$ for $ j < k$
  and $B_{jj} = r_j$. Then,
\begin{align}
  \Xi = 
  \begin{pmatrix} I & B^\dagger \end{pmatrix}
  \begin{pmatrix} 0 & I_q \\ -I_q & 0 \end{pmatrix}
  \begin{pmatrix} I \\ B \end{pmatrix}.
\end{align}
\end{proof}

We are going to find a canonical form of $\Xi$ under congruent transformations
$\Xi \cong E^\dagger \Xi E$ for any invertible matrix $E$.
This is a subject of long history 
under the name of quadratic forms or bilinear forms~\cite{Lam}.
We will use some of the results of this classic subject,
but we first simplify the problem by taking a smaller translation group,
which will almost solve the problem.

From here and below, the symbol $\Xi,\Xi',\ldots$ 
denotes the matrix of an anti-hermitian form.
The indeterminant $x$ denotes the generator for the translation group algebra in one dimension.
We will also use $y$ to denote $y = x^n$ for some $n \ge 1$
when we take a smaller translation group.
The variable~$y$ in this subsection 
should not be confused with the translation variable in two or higher dimensions in other sections.
A {\bf unit} of $R$ is any invertible element of $R$.

\subsubsection{Standard anti-hermitian form}

The {\bf discriminant} of an anti-hermitian form over $\FF_p[x^\pm]$
is defined to be the determinant modulo square elements 
$\{ x^2 : x \in \FF_p \setminus \{0\} \}$
of the matrix $\Xi$ of the form in any basis.
The following lemma implies that we may only consider $\Xi$ of nonzero discriminant.

\begin{lemma}\label{lem:singularXi}
For any $\Xi$ there is a basis change $E$ such that
each column (and hence each row) of $E^\dagger \Xi E$ generates an ideal 
that is generated by a elementary divisor of $\Xi$ 
(a diagonal element of the Smith normal form of $\Xi$).
In particular, if $\det \Xi = 0$, then there is a basis change $E$ such that
\begin{align}
E^\dagger \Xi E = \begin{pmatrix}
\Xi' & 0 \\ 0 & 0 
\end{pmatrix}
\end{align}
where $\det \Xi'$ is nonzero.%
\footnote{
In fact, the lemma is true for an arbitrary square matrix over a principal ideal domain;
the involution in $\dagger$ could be the identity automorphism of $R$.
}
\end{lemma}
\begin{proof}
Let $D \Xi E$ be the Smith normal form of $\Xi$;
$D\Xi E$ is diagonal composed of elementary divisors of $\Xi$.
Then, the matrix $E^\dagger \Xi E = E^\dagger D^{-1} D \Xi E$ consists
of columns $v_i = E^\dagger D^{-1} d_i$ where $d_i$ is a column vector
whose sole nonzero component is the $i$-th elementary divisor of $\Xi$.
This proves the first claim.

Let the rank (defined by the determinantal ideals) of $\Xi$ be $m$.
Then, $d_i$ with $i > m$ is zero, and so is $v_i$.
Since $\Xi^\dagger = - \Xi$, the rows below $m$-th row must be zero.
This proves the second claim.
\end{proof}

\begin{definition}
An anti-hermitian form $\Xi$ of nonzero discriminant is {\bf standard}
if all of its elementary divisors (the diagonal elements in the Smith normal form)
are either $1$ or $x-1$ up to units of $\FF_p[x^\pm]$.
\end{definition}
In particular, if $\det \Xi$ has a unit discriminant, then it is standard.

\begin{lemma}\label{lem:toStandardXi}
For any anti-hermitian form $\Xi$ of nonzero discriminant over $R=\FF_p[x^\pm]$,
there exists an integer $b \ge 1$ such that
$\phi^{(b)}_\# (\Xi)$ is standard.
If $\Xi$ is standard, 
then $\Xi$ has discriminant $(2-x-\bar x)^s = (x-1)^s(\bar x -1)^s$ for some integer $s \ge 0$,
and $\phi^{(b)}_\#(\Xi)$ is also standard for any $b \ge 1$
with discriminant $(2-y - \bar y)^s$
where the exponent $s$ remains unchanged.
In particular, if $\Xi$ has a unit discriminant, then its discriminant is $1$.
\end{lemma}
\begin{proof}
Thanks to \cref{rem:FormAsMap}, we can compute $\phi^{(b)}_\#(\Xi)$
as if $\Xi$ were a linear map,
and for any linear map $\Xi$ over $\FF_p[x^\pm]$
it is shown~\cite[Lem.~6.1, Lem.~7.3]{Haah2013} 
that for some $b \ge 1$ the elementary divisor of $\phi^{(b)}_\#(\Xi)$ 
is either a unit or a scalar multiple of $y-1$.%
\footnote{
The integer $b$ (= the size of the unit cell under the smaller translation group
= the index of ``spontaneous translation symmetry breaking'')
can be exponentially large in the total polynomial degree of an element of $\Xi$
($\approx$ the range of a basis operator for the original algebra).
For example, consider a $1\times 1$ matrix $\Xi = ( f(x + x^{-1}) )$ 
for any primitive polynomial $f$
of an extension field $\FF_{2^n}$ over $\FF_2$ with $n$ odd.
Then, the field extension $\FF_2[x^\pm]/(f(x+x^{-1}))$ has extension degree $2n$ over $\FF_2$
and $x^b -1$ must be a multiple of $f(x+x^{-1})$,
which implies $(2^{2n}-1) | b$.
}
Hence, the determinant of $\Xi' = \phi^{(b)}_\#(\Xi)$
is a power of $y-1$ up to a unit.
Since a unit is a monomial in $R' = \FF_p[y^\pm]$,
we must have
\begin{align}
\det \Xi' = c y^s (y-1)^t
\end{align}
with $c \in \FF_p^\times$, $s,t \in \ZZ$, and $t \ge 0$.
We know that $\Xi' = - (\Xi')^\dagger$, and hence
$c y^{-s} (y-1)^t = (-1)^m c y^s ( y^{-1} - 1)^t$
where $m$ is the dimension of $\Xi'$.
Since the polynomials on the both sides have to be identical,
this implies that $m+t$ is even and $-s = s -t$,
so $\det \Xi' = c y^{-s} (y -1 )^{2s} = c (y + \bar y -2)^s$.
We have to show that $(-1)^s c$ is a square.
If $p = 2$, every element of $\FF_p$ is a square, and we are done.
If $p \neq 2$, then since $-1$ is self-inverse, 
the diagonal of $\Xi'$ is vanishing at $y = -1$,
but the determinant of $\Xi'$ is nonzero, 
and hence $\Xi'|_{y=-1}$ is a symplectic matrix over $\FF_p$, 
and $\det \Xi'|_{y=-1} = 4^s (-1)^s c$ is a square,
and therefore $(-1)^s c$ is a square.

The elementary divisors of $\phi^{(b)}_\#(x-1)$ are $1$'s and a single $y-1$
up to units, since the cokernel of the $1 \times 1$ matrix $x-1$ is $\FF_p$
which is independent of $b$.
Thus, if $\Xi$ is standard, 
the number of $x-1$'s, which is $2s$, in the list of all elementary divisors of $\Xi$,
is the same as $y-1$'s in the list of all elementary divisors of $\phi^{(b)}_\#(\Xi)$.
\end{proof}

\subsubsection{Congruence classes of standard forms}

A bilinear or sesquilinear form $\Xi$ is said to be {\bf isotropic} 
if there exists a nonzero vector $v$ such that $v^\dagger \Xi v = 0$.%
\footnote{
Sometimes a subspace of a symplectic vector space is called isotropic
if the symplectic form vanishes on the subspace.
Here, the isotropy only means there is at least one nonzero vector whose value is zero.
Usually in literature on quadratic/bilinear forms, the term ``totally isotropic''
is used to describe a subspace on which the form vanishes identically.
}
That is, an isotropic form can be represented by a square matrix with a zero in the diagonal.
Note that being isotropic is independent of 
whether we regard the form over a ring $\FF_p[x^\pm]$, an integral domain,
or over its quotient field $\FF_p(x)$, 
the field of all rational ``functions'' in $x$ with coefficients in $\FF_p$.
We fix a notation for an isotropic standard anti-hermitian form:
\begin{align}
\xi = \begin{pmatrix} 0 & x-1 \\ -x^{-1}+1 & 0 \end{pmatrix}. \label{eq:smallxi}
\end{align}

To warm up, let us inspect forms over $\FF_p$:
\begin{lemma}\label{lem:kerFpmatrix}
Any square matrix $A$ of dimension $n$ over $\FF_p$ is isotropic if $n \ge 3$.%
\footnote{
The bound is sharp; $v^T \begin{pmatrix} 1 & 1 \\ 0 & 1 \end{pmatrix} v = 0$
implies $v=0$ over $\FF_2$.
}
\end{lemma}
\begin{proof}
If $\det A = 0$, the claim is trivial, so assume $\det A \neq 0$.

When $p > 2$, consider $A_\pm = (A \pm A^T)/2$ so that $A = A_+ + A_-$.
For any vector $v$, we know $v^T A_- v = 0$, so it suffices to consider $A = A_+$.
The Witt group of all finite dimensional symmetric forms over $\FF_{p>2}$ 
is either  $\ZZ/2\ZZ \oplus \ZZ/2\ZZ$ if $p \equiv 1 \mod 4$,
in which case $n\ge 3$ implies the existence of a hyperbolic plane,
or $\ZZ/4\ZZ$ generated by $\mathrm{diag}(1)$ if $p \equiv 3 \mod 4$
in which case $\mathrm{diag}(1,1,1) \cong \mathrm{diag}(1,-1,-1)$
and we have a nonzero radical (kernel) of $A$ whenever $n \ge 3$;
this is a well-known fact,
but one can see e.g.~\cite[App.~E]{HHPW2017} for elementary computation.

When $p=2$, it suffices to consider that $A$ is $3\times 3$ and its diagonal is all $1$.
Consider the following sequence of congruent transformations
where any $\star$ indicates an arbitrary entry.
\begin{align}
A=\begin{pmatrix}
1 & \star & \star \\
\star & \star & \star \\
\star & \star & \star
\end{pmatrix}
 \to
\begin{pmatrix}
1 & 0 & 0 \\
\star & \star & \star \\
\star & \star & \star
\end{pmatrix}
\to
A' =
\begin{pmatrix}
1 & 0 & 0 \\
0 & \star & \star \\
0 & \star &a
\end{pmatrix} \text{ or }
A'=\begin{pmatrix}
1 & 0 & 0 \\
1 & \star & \star \\
0 & \star & a
\end{pmatrix}.
\end{align}
If $a=0$, we are done.
If $a=1$, then $v^T=\begin{pmatrix}1&0&1\end{pmatrix}$ satisfies $v^T A' v= 0$.
\end{proof}

\begin{lemma}\label{lem:dimtwoisotropic}
Any $2 \times 2$ standard anti-hermitian form over $\FF_p[x^\pm]$ is isotropic.
\end{lemma}
\begin{proof}
We consider a hermitian form $\Xi$ over the field $K = \FF_p(x)$;
this extension does not affect whether $\Xi$ is isotropic.
\Cref{lem:toStandardXi} implies that the determinant of $\Xi$ is~$1$
modulo $\{ r \bar r : r \in K^\times \}$.
If $\Xi$ was anisotropic, any diagonal is nonzero,
and we can diagonalize it as $\Xi \cong \diag(h,h^{-1})$ where $h = - \bar h$.
Then, $ h\cdot 1 \cdot 1 + h^{-1} \cdot \bar h \cdot h = 0 $ 
so $(1,h)$ is in the radical (kernel) of $\Xi$.
\end{proof}

\begin{lemma}\label{lem:StandardToDetS1}
Let $\Xi$ be a standard anti-hermitian form 
of discriminant $(2-x-\bar x)^s$.
If $s \ge 2$,
then there exists an integer $b \ge 1$ such that $\phi^{(b)}_\#(\Xi) \cong \xi \oplus \Xi'$.
\end{lemma}
\begin{proof}
By choosing a sufficiently large $b$ in \cref{lem:toStandardXi}
we may assume that $\phi^{(b)}_\#(\Xi)$ is not only standard but also
\begin{align}
\phi^{(b)}_\#(\Xi) = A (y-1) - A^T (\bar y - 1) + B
\end{align}
where $A$ and $B = - B^T$ are matrices over $\FF_p$.
That is, every entry of $\phi^{(b)}_\#(\Xi)$ is linear in $1,y,\bar y$.
Due to \cref{lem:structureXi}, $B$ has zero diagonal.
We know $\det \phi^{(b)}_\#(\Xi) = (2-y - \bar y)^s$ up to nonzero squares of $\FF_p$
by \cref{lem:toStandardXi}.
To avoid clutter in notation, let us assume that $b = 1$ for the rest of the proof.

Since $s > 0$,
the matrix $\Xi|_{x=1}$ has zero discriminant,
and by \cref{lem:singularXi} there exists a basis change over $\FF_p$
such that
\begin{align}
\Xi \cong \begin{pmatrix}
\Xi_1 & \Xi_2 \\
-\Xi_2^\dagger & \Xi_0
\end{pmatrix}
\end{align}
where $\Xi_2|_{x=1} = 0$, $\Xi_0|_{x=1} = 0$, and $\det \Xi_1|_{x=1} = 1$.
Let $n$ be the dimension of $\Xi$, and $m$ be the dimension of $\Xi_1$.
With $x=1$ no submatrix of dimension larger than $m$ can have nonzero determinant.
This means that there are exactly $m$ elementary divisors of $\Xi$
that do not vanish at $x=1$.
Since $\Xi$ is standard, there are exactly $n-m$ $x-1$'s in the list of all elementary divisors,
whose product is the discriminant up to units.
In particular, the dimension of $\Xi_0$ is $n-m = 2s$, an even number $\ge 4$.

Since $\Xi_0$ meets the dimension requirement of \cref{lem:kerFpmatrix},
there is a congruent transformation over $\FF_p$
that gives a zero diagonal element, say, in the lower right corner in $A$.
The last row and the second-to-last row (or equivalently column) of $\Xi$
must generate the ideal $(y-1)$,
since they are the image of $(y-1)e_n$ and $(y-1)e_{n-1}$, respectively
(where $e_j$ is the unit column vector),
under an invertible matrix; see the proof of \cref{lem:singularXi}.
Hence after some congruent transformation,
which in general involves higher powers of $y^\pm$,
we obtain $\Xi$ in the following form:
\begin{align}
\Xi 
&\cong 
\begin{pmatrix}
\Xi' & c & 0_{(n-2) \times 1} \\
-c^\dagger & d & y - 1  \\
0_{1 \times (n-2)} & -\bar y + 1 & 0 
\end{pmatrix}
\cong
\begin{pmatrix}
\Xi' & 0 & 0_{(n-2) \times 1} \\
0 & d & y - 1  \\
0_{1 \times (n-2)} & -\bar y + 1 & 0 
\end{pmatrix}\nonumber \\
&\cong
\begin{pmatrix}
\Xi' & 0 & 0_{(n-2) \times 1} \\
0 & 0 & y - 1  \\
0_{1 \times (n-2)} & -\bar y + 1 & 0 
\end{pmatrix}
\end{align}
In the last congruence
we used the fact that since $d = r - \bar r$ for some $r \in \FF_p[y]$ by \cref{lem:structureXi},
it holds that $d = (y-1)q - (\bar y -1)\bar q $ where $r = (y-1)q + f$ with $f \in \FF_p$.
The lower right corner is $\xi$ as claimed.
\end{proof}

\begin{lemma}\label{lem:s1decompose}
Let $\Xi$ be an isotropic standard anti-hermitian form 
of discriminant $1$ or $(2-x-\bar x)$ over $R = \FF_p[x^\pm]$.
Then, $\Xi$ is congruent to either $\xi \oplus \Xi'$ or $\lambda_1 \oplus \Xi'$ where
$\xi$ and $\lambda_1$ are as in \cref{eq:smallxi,eq:lambdaMatrix}.
\end{lemma}
\begin{proof}
Suppose the top left diagonal of $\Xi$ is zero.
Then, by elementary row and column operations we have
\begin{align}
\Xi 
\cong 
\begin{pmatrix}
0 & b & 0_{1 \times(n-2)} \\
-\bar b & d & \star \\
0_{(n-2)\times 1} & \star & \Xi'
\end{pmatrix}.
\end{align}
Note that $\det \Xi = b \bar b ~\det \Xi'$ by the cofactor expansion formula for determinants.
Since the determinant of $\Xi$ is $2-x-\bar x$ up to nonzero squares of $\FF_p$,
the Laurent polynomial $b \bar b$ must divide $2-x-\bar x$.
The only options are $b = 1$ or $b = x-1$ up to units.
In either case, $d$ that has no constant term by \cref{lem:structureXi},
can be eliminated by some congruent transformation.

If $b = 1$ which has to be the case if the discriminant of $\Xi$ is $1$,
then $\star$ can be eliminated, and $\Xi \cong \lambda_1 \oplus \Xi'$.
If $b = x-1$,
then $\star$ can be made so that its entries are all in $\FF_p$.
If $\star$ does not become zero, then we are back to the situation where $b=1$.
If $\star$ becomes zero, then $\Xi \cong \xi \oplus \Xi'$.
\end{proof}

\begin{lemma}\label{lem:dimthreeisotropic}
Suppose $m \ge 3$.
Then any $m \times m$ anti-hermitian form $\Xi$ over $\FF_p[x^\pm]$ is isotropic.
\end{lemma}
\begin{proof}
This trivially follows from a known fact~\cite[Example~2 of App.~2]{MilnorHusemoller},
but let us explain a bit more.
Let $F$ be the subfield of $K$ fixed element-wise by the involution $x \mapsto x^{-1}$.
Then, $K$ is a quadratic extension field over $F$;
put $t = x+x^{-1}$ and $F = \FF_p(t)$ so that $K = F[y]/(y^2 -ty +1 )$.

Suppose $p \neq 2$.
We consider a hermitian form $\Phi = (x - x^{-1})\Xi$ over the field $K = \FF_p(x)$;
$\Phi$ is isotropic if and only if $\Xi$ is.
The form $\Phi$ is then a symmetric bilinear form of dimension~$2m$ over $F$
by the identification $K = F^2$ as an $F$-vector space,
and for the rest of the proof we regard $\Phi$ as a form over $F$.
Then~\cite[XI~1.5]{Lam} implies that $\Phi$ is isotropic.%
\footnote{
The proof of this involves ``Hasse--Minkowski principle''~\cite{Lam}.
}

Suppose $p = 2$. We can use~\cite[Lem.~36.8]{Kniga}.
The form $\Xi$ is hermitian, and we may regard it as a quadratic form~$\xi$ over~$F$ of dimension~$2m$
by $K^m = F^m \oplus y F^m \ni f_0 + y f_1 \mapsto (f_0^\dag + \tfrac 1 y f_1^\dag) \Xi (f_0 + y f_1) \in F$.
If $\Xi$ is anisotropic, then, tautologically, $\xi$ must also be anisotropic.
\cite[Lem.~36.8]{Kniga} (cf.~\cite[Prop.~7.31]{Kniga})
says that $\tfrac 1 2 (2m) \le [F:F^2] = 2$.
\end{proof}

\begin{corollary}[Structure theorem of anti-hermitian forms]\label{cor:structure}
For any anti-hermitian form~$\Xi$ over~$\FF_p[x^\pm]$,
there exist integers $b \ge 1, s \ge 0, t \ge 0$ such that
$\phi^{(b)}_\#(\Xi) \cong \xi^{\oplus s} \oplus \lambda_1^{\oplus t} \oplus 0$.
The integer~$s$ is uniquely determined by $\Xi$, irrespective of $b$.
\end{corollary}
\begin{proof}
By \cref{lem:singularXi,lem:toStandardXi} it suffices to prove the claim for a nondegenerate standard $\Xi$.
If a standard $\Xi$ has discriminant $(2-x-\bar x)^s$ with $s \ge 2$,
then by \cref{lem:StandardToDetS1} we can decrease $s$ by extracting a direct summand $\xi$.
Hence, it suffices to prove the claim with a standard $\Xi$ of discriminant $(2-x-\bar x)$ or $1$.
Increasing $b$ by a factor of 2 if necessary, we may assume that $\Xi$ has even dimension.
If $\Xi'$ has dimension $>2$, then, by \cref{lem:dimthreeisotropic}, $\Xi$ is isotropic.
If $\Xi'$ has dimension $2$, then, by \cref{lem:dimtwoisotropic}, $\Xi$ is isotropic.
Hence, \cref{lem:s1decompose} implies that we can still find a wanted direct summand.
The direct summands $\lambda_1$ and $\xi$ both have even dimensions,
and by induction in the dimension of $\Xi$, we are done.
The uniqueness of $s$ follows from the fact that $\coker \Xi \cong \coker \phi^{(b)}_\#(\Xi)$ 
as $\FF_p$-vector spaces of dimension $2s$.
\end{proof}

\subsubsection{Proof of \cref{thm:mca-moduleversion}}
\begin{proof}
Let $B$ be a matrix whose columns freely generate the module $A^\perp$.
By \cref{cor:structure} we may assume that 
$\Xi = B^\dagger \lambda_q B = \xi^{\oplus s} \oplus \lambda_1^{\oplus t} \oplus 0$;
this requires taking a smaller translation group.
We collect exactly one basis vector $v_j$ from each pair $\{v_j, v'_j\}$
that forms one summand $\xi$ or $\lambda_1$.
We further collect every basis vector $w_k$ corresponding to the $0$ summand of $\Xi$.
The desired submodule $\stab$ is generated by $\{v_j, w_k\}$.
Indeed, if $v \in A^\perp$, then $v$ is uniquely written as a linear combination of $v_j,v'_j,w_k$.
Further if $v \in \stab^\perp$, then the coefficients of $v'_j$ must be zero.
\end{proof}

\subsection{Miscellaneous results}

As a direct application of the sufficient condition for $\ker \sigma^\dagger \lambda_q = \im \sigma$
we can now prove \cref{lem:nondegeneracy},
which asserts the nondegeneracy of the ground state 
of the Walker-Wang model Hamiltonian for the 3-fermion theory.
%%%%%%%%%%%%%%%%%
\begin{proof}[Proof of \cref{lem:nondegeneracy}]\label{pflem:nondegeneracy}
The model Hamiltonian of \cref{eq:WWHam} 
is a translation-invariant commuting Pauli Hamiltonian
on a system of qubits ($p=2$);
we will soon check that the Hamitonian is frustration-free.
The stabilizer map is
\begin{align}
\sigma = \left(
\begin{array}{cccccccc}
 1+\frac{1}{x} & 0 & 0 & 0 & y z+\frac{1}{x} & 0 & 0 & y z \\
 1+\frac{1}{y} & 0 & 0 & x z+\frac{1}{y} & 0 & 0 & x z & 0 \\
 1+\frac{1}{z} & 0 & x y+\frac{1}{z} & 0 & 0 & x y & 0 & 0 \\
 0 & 1+\frac{1}{x} & 0 & 0 & \frac{1}{x} & 0 & 0 & y z+\frac{1}{x} \\
 0 & 1+\frac{1}{y} & 0 & \frac{1}{y} & 0 & 0 & x z+\frac{1}{y} & 0 \\
 0 & 1+\frac{1}{z} & \frac{1}{z} & 0 & 0 & x y+\frac{1}{z} & 0 & 0 \\
 0 & 0 & y+1 & z+1 & 0 & 0 & 0 & 0 \\
 0 & 0 & x+1 & 0 & z+1 & 0 & 0 & 0 \\
 0 & 0 & 0 & x+1 & y+1 & 0 & 0 & 0 \\
 0 & 0 & 0 & 0 & 0 & y+1 & z+1 & 0 \\
 0 & 0 & 0 & 0 & 0 & x+1 & 0 & z+1 \\
 0 & 0 & 0 & 0 & 0 & 0 & x+1 & y+1 \\
\end{array}
\right) \label{eq:stabilizermap}
\end{align}
over $R = \FF_2[x^\pm, y^\pm, z^\pm]$.
The convention of the ordering of the qubits in a unit cell 
(three edges with two qubits per edge)
is that we put ``1''-qubits along $x,y,z$-axes in the first three components,
and then ``2''-qubits along $x,y,z$-axes in the last three components.
By the convention on $\lambda_3$ of \cref{eq:lambdaMatrix},
Pauli $X$ operators are represented in the upper 6 rows,
and Pauli $Z$ operators are in the lower 6 rows.

To show that the Hamiltonian is frustration-free,
we have computed the kernel of $\sigma$.
One can easily check that $\sigma K = 0$ where
\begin{align}
K = \left(
\begin{array}{cc}
 x y z+1 & x y z \\
 1 & x y z+1 \\
 z+1 & 0 \\
 y+1 & 0 \\
 x+1 & 0 \\
 0 & z+1 \\
 0 & y+1 \\
 0 & x+1 \\
\end{array}
\right)
\end{align}
and also that $I_2(K) = R$.%
\footnote{This is obvious by looking at the uppermost $2\times 2$ block.}
In terms of operators on the complex Hilbert space,
each column of $K$ expresses the fact that the product of the six plaquette terms $B_{P,j}$ for a fixed $j=1,2$ 
around a cube is a product of three vertex terms up to a phase factor.
By a similar argument as in the proof of the second claim of \cref{lem:chaincondition},
we know the columns of $K$ generates the kernel of $\sigma$.
Since the smallest nonzero determinantal ideal of $K$ is unit,
the relation $\ker \sigma = \im K$ continues to hold after factoring out 
$\mathfrak b_L = (x^L - 1, y^L - 1, z^L -1 )$
that imposes periodic boundary conditions of linear system size $L$.
Therefore, it suffices for us to check that
every product of Hamiltonian terms that corresponds to $\sigma K$ 
is \emph{not} $-1$ times the identity operator.
There are only two cases to check, one for each column of $K$.
This is an easy calculation, and we have confirmed that 
the stabilizer group generated by the terms of Hamiltonian 
does not contain $-1$ times the identity operator,
and hence the Hamiltonian of \cref{eq:WWHam} is frustration-free.

Next, it is simple to check that $\sigma^\dagger \lambda_3 \sigma = 0$,
which implies that the Hamiltonian is indeed commuting.
We have computed that $I_6(\sigma) = R$ by Gr\"obner basis computation.
Since $\sigma^\dagger \lambda_3 \sigma = 0$, we know $I_7(\sigma)$ vanishes.
By \cref{lem:chaincondition}, we conclude that $\ker \sigma^\dagger \lambda_3 = \im \sigma$.%
\footnote{
By~\cite[Lem.~3.1]{Haah2013} this implies that the Hamiltonian
satisfies the local topological order condition.
}
Then, by~\cite[Cor.~4.2]{Haah2013} the ground state subspace of the Hamiltonian
is a quantum error correcting code encoding zero logical qubit,
which means that the ground state is nondegenerate.
\end{proof}

\begin{proof}[Alternative proof of \cref{lem:enough-smallfermionloops}]\label{pflem:enough-smallfermionloops}
The problem can be cast into the polynomial framework,
and we transcribe the operators as 
$\sigma_\text{surf}$ for the truncated bulk terms and 
$\sigma_{SFL}$ for the small fermion loop operators in the bottom of \cref{fig:WWsurface}:
\begin{align}
\sigma_\text{surf} &= \left(
\begin{array}{cccccccc}
 1+\frac{1}{x} & 0 & 0 & 0 & 0 & 1 & 0 & 1 \\
 1+\frac{1}{y} & 0 & 0 & 0 & 1 & 0 & 1 & 0 \\
 0 & 1+\frac{1}{x} & 0 & 0 & 0 & 0 & 0 & 1 \\
 0 & 1+\frac{1}{y} & 0 & 0 & 0 & 0 & 1 & 0 \\
 0 & 0 & y+1 & 0 & \frac{1}{x} & 0 & 0 & 0 \\
 0 & 0 & x+1 & 0 & 0 & \frac{1}{y} & 0 & 0 \\
 0 & 0 & 0 & y+1 & 0 & 0 & \frac{1}{x} & 0 \\
 0 & 0 & 0 & x+1 & 0 & 0 & 0 & \frac{1}{y} \\
\end{array}
\right), \quad
\sigma_\text{SFL} = \left(
\begin{array}{cc}
 0 & 1+\frac{1}{x} \\
 0 & 1+\frac{1}{y} \\
 1+\frac{1}{x} & 0 \\
 1+\frac{1}{y} & 0 \\
 0 & y+1 \\
 0 & x+1 \\
 y+1 & y+1 \\
 x+1 & x+1 \\
\end{array}
\right).
\end{align}
It is straightforwardly verified that $\sigma_\text{surf} ^\dagger \lambda_4 \sigma_\text{SFL} = 0$ over $R = \FF_2[x^\pm, y^\pm]$,
which implies that the small fermion loop operators indeed commute with the truncated bulk terms.
We have computed the determinantal ideals:
\begin{align}
I_6(\sigma_\text{surf}^\dagger \lambda_4) &= ( (1+x)^2 , (1+x)(1+y), (1+y)^2 ), \\
I_2(\sigma_\text{SFL}) &= ( (1+x)^2 , (1+x)(1+y), (1+y)^2 ).
\end{align}
They both have depth $2$.
Hence, the sequence
\begin{align}
0 \to 
R^2 \xrightarrow{ \sigma_\text{SFL} } 
R^8 \xrightarrow{ \sigma_\text{surf}^\dagger \lambda_4 }
R^8
\end{align}
satisfies the Buchsbaum-Eisenbud criterion~\cite{BuchsbaumEisenbud1973Exact}, 
and therefore is exact.
In particular, the commutant of the truncated bulk operators
within the Pauli group is generated by the small fermion loop operators.

The last claim is generally implied be the fact that any element $v$ of $\im \sigma_\text{SFL}$
can be written as a linear combination of a unique Gr\"obner basis of $\im \sigma_\text{SFL}$.
Under the total degree order of monomials,
the standard division algorithm gives an algorithm to decompose $v$ in terms of
the small fermion loop operators. See the proof of~\cite[Lem.~3.1]{Haah2013}.
\end{proof}

\begin{lemma}\label{lem:removingRelations2D}
If an algebra $\mathcal A$ 
on an infinite two-dimensional lattice of finitely many qudits of prime dimension $p$
is a commutant of a translation-invariant set of (generalized) Pauli operators,
then there exists a generating set $A$ for $\mathcal A$
consisting of Pauli operators such that $A$ is translation-invariant 
and any nonempty product of elements of $A$ is nonidentity (i.e., locally nonredundant).
\end{lemma}
\begin{proof}
Since the translation group ring $R = \FF_p[x^\pm,y^\pm]$ is Noetherian,
the group of Pauli operators whose commutant is $\mathcal A$ is a finitely generated $R$-module.
Hence, we may represent this group as a matrix $G$ whose columns are Laurent polynomial vectors
representing generators of the group.
The commutant as a subgroup of the full group of all Pauli operators
is then the kernel of $G^\dagger \lambda_q$ 
where $q$ is the number of qudits in a translation unit cell.

Now, consider any finite free resolution of $\ker G^\dagger \lambda_q$:
(The rest of the proof is a slight variant of the proof of \cref{lem:separatorexists}.)
\begin{align}
0 \to F_n \to \cdots \to F_3 \xrightarrow{T} F_2 \xrightarrow{C} 
R^{2q} \xrightarrow{G^\dagger \lambda_q} F_0.
\end{align}
By construction $\im C = \ker G^\dagger \lambda_q$,
and as a $\CC$-algebra the Pauli operators of $\im C$ generates $\mathcal A$.
By the Buchsbaum-Eisenbud theorem we must have $\mathrm{depth}~ I(T) \ge 3$,
but the base ring is two-dimensional.
Therefore, $I(T) = R$, and by Quillen-Suslin-Swan theorem~\cite{Suslin1977Stability,Swan1978}
there exist invertible matrices $E$ and $E'$
such that $E T E'$ is diagonal with the identity matrix on the diagonal.
This means that the nonzero columns of $CE$ forms a locally nonredundant generating set for $\im C$.
\end{proof}

\begin{lemma}\label{lem:NoredunNochargeFlippable-qudit}
Suppose with the terms of a translation-invariant commuting Pauli Hamiltonian 
in an infinite $D$-dimensional lattice with $q$ qudits of prime dimension $p$ at each site,
any nonempty product of terms is nonidentity.
If every excitation of finite energy created by a Pauli operator of possibly infinite support,
can in fact created by a Pauli operator of finite support,
then there exists a translation-invariant set of operators of finite support, 
each of which flips exactly one term of the Hamitonian.
\end{lemma}
Using \cref{def:nonredundant,def:charge},
the lemma can be phrased as 
\emph{
the terms of 
any locally nonredundant translation-invariant commuting Pauli Hamiltonian without nontrivial charges,
can be flipped by local operators.
}
\begin{proof}
The first supposition in the polynomial framework is to say that
the terms define an injection $\sigma: R^t \to R^{2q}$ whose image is the stabilizer module.
In~\cite[Thm.~1]{Haah2013} the set of all equivalence classes of topological charges
is identified with the torsion submodule of $\mathrm{coker}~\sigma^\dagger$.
So, the no charge assumption is transcribed as $\coker~\sigma^\dagger$ is torsion free.
But, since the rank of the injective map $\sigma$ must be equal to $t$,
the ideal inequality~\cite[XIX.2.5]{Lang}
\begin{align}
(\ann~\coker~\sigma^\dagger)^t \subseteq I_t(\sigma) \subseteq \ann~\coker~\sigma^\dagger
\end{align}
shows that $\coker~\sigma^\dagger$ is pure torsion.
Therefore, $\coker \sigma^\dagger = 0$.
This means that there exists a vector $v_a$ such that $\sigma^\dagger \lambda_q v_a = e_a$ where
$e_a$ is the unit column vector with the sole 1 at $a$-th component.
Taking translates of the Pauli operator corresponding to $v_a$,
we see that any single stabilizer generator (a term of the Hamiltonian) can be flipped alone.
The collection of all $v_a$ and their translates form a local flipper,
and it is translation invariant.
\end{proof}

\begin{remark}\label{rem:includingFermions}
The statements of \cref{lem:removingRelations2D,lem:NoredunNochargeFlippable-qudit}
can be modified to handle translation-invariant systems with qubits and fermionic modes.
The analogue of the Pauli operator is a finite product of Majorana operators and qubit Pauli matrices.
The prime $p$ that appears in both lemmas, has to be fixed as $p=2$.
The terms of the Hamiltonian is always assumed to be fermion parity even,
and the $\lambda_q$ has to be replaced by
\begin{align}
\lambda_q \oplus I_m
\end{align}
to capture the commutation relations,
where $q$ is the number of qubits per site, and $m$ is the number of Majorana modes per site.
No other change is necessary. 
With this understanding incorporated into the terminology ``commuting Pauli Hamiltonians,''
the conclusions of the both lemmas hold.
After all, we have not used anything special about $\lambda_q$ in these lemmas.
\end{remark}

\section{Discussion}
\label{discussionsection}
The method that we have used to constructing the QCA $\alpha_{WW}$ may be more general.  The Walker-Wang model that we considered was just a particular example of a model without bulk topological order but with surface topological order.  Other such models may perhaps give rise to other interesting QCA.

Indeed, this approach suggests that the notion of a ``trivial state" should perhaps be re-considered.  
Commonly a quantum state is called ``trivial''
if it can be mapped to a product state using a quantum circuit (perhaps with some weakening of strict locality).
However, here we see that the Walker-Wang model ground state 
can be mapped to a product state using a QCA $\alpha_{WW}$ that we believe to be nontrivial.
Of course, it is possible that given some separator (such as that obtained from the Walker-Wang model), 
there is a circuit that maps a particular eigenstate of all the $\opZ_a$ (such as the Walker-Wang ground state) 
to a product state without mapping the separator to the trivial separator.  
So, it is unclear whether or not there exists a circuit to map the Walker-Wang ground state to a product state.  

On the other hand, we note that the $3$-fermion Walker-Wang model certainly 
defines an invertible phase of matter, 
in the sense that there exists another model that can be stacked on top of it (its `inverse') 
such that the pair can be disentangled by a circuit.
Indeed, the inverse can be taken as just another $3$-fermion Walker-Wang model. 
This follows from the fact that $\alpha_{WW} \otimes\alpha_{WW}$ is a circuit up to translations, 
as shown in \cref{thm:CQCAsqIsCircuit}.
Alternatively, one can note that the product of two $3$-fermion models is equivalent, 
as a unitary modular tensor category, to the product of two toric codes.  
Hence a stack of two $3$-fermion Walker-Wang models 
is equivalent to a stack of two Walker-Wang models based on the toric code modular tensor category; 
the latter can be disentangled with a circuit, as shown in \cite{RBH}.%
\footnote{
The authors of~\cite{RBH} start with a model with $7$-body interaction terms on a simple cubic lattice of qubits,
but in the course of discussion they consider a slightly different Hamiltonian of $5$-body interaction.
The latter can be regarded as being put on a lattice of interpenetrating cubic lattice with qubits on edges,
and has term $X$ on an edge multiplied by a plaquette $Z$-operator around the edge,
which is a Walker-Wang model for the toric code.
Ref.~\cite{RBH} also considers an encoding scheme at the boundary,
which is the surface topological order. See \cite{Roberts2016}.
}
Thus, if the $3$-fermion Walker-Wang model could not be disentangled with a circuit (possibly with tails),
it would define a new invertible phase of matter in $3$ spatial dimensions. 
Such a phase, which does not require any symmetries to protect it, 
does not appear in any currently proposed classifications of phases of matter~\cite{Kapustin, Freed}.
The question of how trivial the $3$-fermion Walker-Wang model ground state is
thus remains an interesting open problem.

Another interesting question is whether $\alpha_{WW}$ is a finite depth circuit when stabilized by local fermionic degrees of freedom.
Although we do not know if this is true, let us make some observations.
First, consider the Walker-Wang model based on the braided fusion category $\{1,f\}$. 
This is a bosonic model whose only non-trivial pointlike excitation is a fermion $f$,
and whose only other non-trivial excitation is a loop-like `vison' that has full braiding phase $-1$ with $f$.
It is thus very similar to the $3$d toric code, 
except that the pointlike excitation is a fermion not a boson.  
It can be thought of as a $\ZZ_2$ gauge field coupled to fermionic matter, 
i.e., it is what one obtains after gauging the $\Z_2$ fermion parity symmetry in a trivial fermionic insulator.

Now consider stacking our $3$-fermion Walker-Wang model with the $\{1,f\}$ Walker-Wang model.
The stacked model can be thought of as the Walker-Wang model 
built on the premodular category $\{1,\fone,\ftwo,\fthree \} \times \{ 1,f \}$.
However, after a relabeling this is the same category as $\{1,e,m,\varepsilon \} \times \{1,f\}$, 
where $\{1,e,m,\varepsilon\}$ is the toric code (just set $e=\fone f, m=\ftwo f$).
This shows that after stacking with the $\{1,f\}$ Walker-Wang model,
the $3$-fermion and toric code Walker-Wang models become equivalent.
However, as we already noted, the toric code Walker-Wang model is known to be finite circuit disentanglable~\cite{RBH}.
Thus, there exists a circuit which transforms the $3$-fermion Walker-Wang ground state 
tensored with the ground state of the $\{1,f\}$ model 
into a trivial product state tensored with the ground state of the $\{1,f\}$ model. 
It is not clear if the same is true when the stabilization is by a model with local fermionic degrees of freedom,
though perhaps this can be addressed by using higher dimensional versions of the Jordan-Wigner bosonization 
duality~\cite{Kapustin2d, Ellison, Kapustin3d}. 
Also, even if a disentangling circuit exists, it is not clear whether $\alpha_{WW}$ is such a circuit.

We note that with local fermionic degrees of freedom,
there are other candidates for non-trivial QCA, even if $\alpha_{WW}$ ends up being trivial. 
Indeed, there exists a 3d fermionic commuting projector model 
which is conjectured to be in the same phase 
as the root phase of the $\Z_{16}$ classification of topological superconductors in class DIII~\cite{TopoSC}.
This model realizes a surface topological order~$SO(3)_3$, equal to the integral spin sub-theory of~$SU(2)_6$
--- which, when realized purely in $2d$, has chiral central charge~$\frac{9}{4}$, 
which, modulo~$\frac{1}{2}$, is half of the minimal chiral central charge~$\frac{1}{2}$
that can be realized with a 2d short range entangled phase (namely a $p+ip$ superconductor).
This is analogous to the $3$-fermion topological order in the bosonic setting,
which has a chiral central charge of $4$, 
one half of the minimal allowed value of $8$ realized in the $E_8$ bosonic phase.
If the~$SO(3)_3$ Walker-Wang model can be fermionized into a model built on fermionic degrees of freedom, 
a QCA which disentangles it would be a natural candidate for a non-trivial fermionic QCA.

\appendix

\section{Structure of a ground state of commuting 2-local Hamiltonians with fermions} \label{sec:fBV}

In this Appendix, we prove a fermionic generalization of 
a result of Bravyi and Vyalyi~\cite{bravyi2005commutative}
regarding $2$-local Hamiltonians which are sums of commuting projectors.
This generalization is not related to the rest of the paper, 
other than that elsewhere we do spend some effort analyzing two-dimensional (rather than $2$-local) commuting projector Hamiltonians 
with fermionic degrees of freedom.
We discovered the generalization in this appendix 
while trying to better understand those two-dimensional fermionic Hamiltonians, 
and, though it did not help us there, we felt that this generalization might be of independent interest.

A $2$-local Hamiltonian is a Hamiltonian in which each term acts on at most $2$ sites.
Bravyi and Vyalyi considered the case that
each site was a qudit and the Hamiltonian was a sum of commuting terms, showing that the ground state was a product state up to conjugation by a quantum circuit.  The quantum circuit is composed of unitary gates, each of which acts on a pair of sites $i,j$ which are neighbors on the interaction graph (i.e., this is a graph $G$ with a vertex for every site and an edge for every pair of sites for which some term in the Hamiltonian acts on those two sites); these unitaries all commute with each other so if the interaction graph has bounded degree then this is a circuit of bounded depth.
We generalize this to the case that for each site there is some qudit as well as some Majorana degrees of freedom.

First, we emphasize that it is too much to expect in the fermionic case that the state can be disentangled by a quantum circuit.  Consider the ``Majorana chain.''
This is a one-dimensional system.  
On each site $j$, we have Majorana modes $\gamma_j,\gamma'_j$, where now we take $i$ to be an integer.
We have the Hamiltonian
\begin{align}
H=i \sum_j \gamma'_j \gamma_{j+1}.
\end{align}
The operator $i\gamma'_j \gamma_{j+1}$ is not a projector, but its square is equal to the identity, 
so up to a linear rescaling these terms are projectors.
Famously this chain cannot be disentangled by a quantum circuit~\cite{kitaev2001unpaired}.

Thus, we will show a weaker result.  
In a sense, we will show that the ground state is, up to conjugation by a quantum circuit, 
a ``generalized Majorana chain.''
For any pair of sites $i,j$ which are neighbors on the interaction graph,
let $P_{i,j}$ denote the projector acting on those sites so that $H=\sum_{(i,j)} P_{i,j}$, where the sum is over neighboring $i,j$ and $(i,j)$ denotes an unordered pair.
Our result is:
\begin{theorem} \label{fermionicbv}
There exist unitaries $U_{i,j}$ supported on each edge $(i,j)$ 
in the interaction graph and projectors $\Pi_i$ supported on each site $i$
with the following properties.
\begin{enumerate}
\item The unitaries $U_{i,j}$ are mutually commuting,
 and each unitary $U_{i,j}$ commutes with $P_{k,l}$ if $(i,j) \neq (k,l)$.
\item The conjugated Hamiltonian $\tilde H = \sum_{(i,j)} U_{i,j}^\dagger P_{i,j} U_{i,j}$ 
commutes with $\Pi_i$ for all $i$
\item Some ground state of $\tilde H$ is in the $+1$ eigenspace of all $\Pi_i$, 
and, restricted to this eigenspace, each $U^\dagger_{i,j} P_{i,j} U_{i,j}$ is either $0$, $1$, or
$\frac 1 2 (1 + i \tau_{j \to i} \tau_{i \to j})$
where $\tau_{j \to i }, \tau_{i \to j}$ are fermion parity odd hermitian operators that square to $1$.
Here, for any $i \neq j \neq k$
\[
\{ \tau_{i \to j},\tau_{j \to k} \}=0
\]
whenever they exist.
\end{enumerate}
\end{theorem}
\noindent
We refer to the ground state of $\tilde H$ as a generalized Majorana chain: 
the Majorana chain~\cite{kitaev2001unpaired} can be written in this way for example.

Note that the operators $\tau_{j\rightarrow i}$ represent the algebra of Majorana operators, 
i.e., an algebra of operators $\gamma_a$ obeying the relations $\{\gamma_a,\gamma_b\}=\delta_{a,b}$. 
However, there are various possible ways to represent this algebra.  
As an example, given a qubit with Pauli operators $X,Y,Z$ and a single Majorana mode $\gamma$, 
then the operators $\gamma X, \gamma Y, \gamma Z$ give such a representation of the algebra of 
three Majorana operators $\gamma_1,\gamma_2,\gamma_2$.

If the degree of the graph $G$ is bounded, then conjugation by the unitaries $U_{i,j}$ can be implemented by a quantum circuit of bounded depth.

We give two related but slightly different approaches to show this theorem, 
one an analytic approach and one an algebraic approach.  Each has its advantages.
Both approaches make use of the idea of interaction algebras~\cite{knill2000theory}.
Let ${\cal A}_{j\rightarrow i}$ denote the interaction algebra of $P_{i,j}$ on site $i$.
This is the minimal algebra such that $P_{i,j}$ can be written 
as a sum of products of some term supported on $j$ times some term in ${\cal A}_{j\rightarrow i}$.

This algebra ${\cal A}_{j \rightarrow i}$ is a $\ZZ_2$ graded algebra. 
A term with odd grading is a sum of terms with an odd product of Majoranas 
and similarly for a term with even grading.  
Let ${\cal A}^{even}_{j \rightarrow i}$ denote the even subalgebra of ${\cal A}_{j \rightarrow i}$.
For any pair $j\neq k$ which are both neighbors of $i$, 
the algebras ${\cal A}_{j \rightarrow i}$ and ${\cal A}_{k \rightarrow i}$ {\it supercommute}.
That is, given any two terms, one from each of these algebras, 
if one term is even graded then they commute, 
and if both terms are odd graded then they anti-commute.

\subsection{Analytic approach}
In this subsection, we will use the solution of a differential equation to find a unitary that brings the Hamiltonian to a simpler form.
Consider ${\cal A}^{even}_{j \rightarrow i}$.  By standard results, this algebra is isomorphic to a direct sum of matrix algebras.  For each such matrix algebra, the identity operator in that matrix algebra (and the zero operator in all other matrix algebras) corresponds to some central element of ${\cal A}^{even}_{j \rightarrow i}$.
If there is only one matrix algebra in the direct sum, then ${\cal A}^{even}_{j \rightarrow i}$ is simple.
Now choose a $V \in {\cal A}^{even}_{j \rightarrow i}$ as follows: choose it to correspond to a direct sum of {\it diagonal} matrices in these matrix algebras such that all diagonal entries are distinct.  That is,
it corresponds to some $\oplus_a V_a$ where each $V_a$ is diagonal and where every diagonal entry of every $V_a$ is distinct from every
diagonal entry of that $V_a$ and every diagonal entry of every other $V_b$.

Now, define a flow equation:
\begin{align}
\partial_t H(t)=[H(t),[H(t),V]],
\end{align}
where $t\geq 0$ is some real  parameter, with boundary conditions at $t=0$ that $H(t)=H$.
This flow equation can be written as
\begin{align}
\partial_t H(t)=[H(t),\eta(t)],
\end{align}
where $\eta(t)=[H(t),V]$.
By the commutation properties of the Hamiltonian, we have
$H(t)=\sum_{(k,l)} P_{k,l}(t)$,
where for $(k,l) \neq (i,j)$ we have $P_{k,l}(t)=P_{k,l}$
and
\begin{align}
\partial_t P_{i,j}(t)=[P_{i,j}(t),[P_{i,j}(t),V]],
\end{align}
i.e., we can reduce the flow equation to a flow for $P_{i,j}(t)$ while all other $P_{k,l}(t)$ are independent of~$t$.

This flow equation for $P_{i,j}(t)$ is a gradient flow for $\Tr((P_{i,j}(t)-V)^2)$ and has a 
limit~\cite{bloch1992completely,brockett1991dynamical} 
as $t\rightarrow \infty$ which commutes with $V$.%
\footnote{
For our purpose, it suffices to note that the flow preserves the norm, 
but does not increase the potential $\Tr((P_{i,j}(t)-V)^2) \ge 0$,
and hence the $t$-derivative of the potential converges to zero at any accumulation point of the flow,
which implies that at any accumulation point $R_{i,j}$ we have $[R_{i,j},V] = 0$.
}
Call this limit $R_{i,j}$ so that $[R_{i,j},V]=0$.
This flow is generated by an infinitesimal anti-Hermitian $\eta$ acting on site $i,j$,
and so integrating from $t=0$ to $t=+\infty$ there is some unitary supported on sites $i,j$,
which conjugates $P_{i,j}$ to $P_{i,j}(+\infty)$, leaving all other $P_{k,l}$ unchanged.
This preserves all the commutation properties of the Hamiltonian.

Note also that now all terms in the Hamiltonian commute with $V$ 
so that some ground state of the Hamiltonian may be chosen to be an eigenvector of $V$.
So, we restrict to the eigenspace of $V$ with that eigenvalue.  Having made this restriction,
the even subalgebra of the interaction algebra of $R_{i,j}$ on $i$ is now just the algebra of complex scalars.
This means that the interaction algebra of $R_{i,j}$ on $i$ 
is generated by some fermion parity odd operator whose square is a scalar.

We can repeat this simplification by considering the interaction algebra of $R_{i,j}$ on $j$.  Let
$\tilde P_{i,j}$ be the limit of that flow equation.  Integrating this flow means that $\tilde P_{i,j}$ is related to $R_{i,j}$ by some unitary.  Again, the even subalgebra of the interaction algebra of $\tilde P_{i,j}$ on $j$ 
is just the algebra of complex scalars,
so that the interaction algebra of $\tilde P_{i,j}$ on $j$ is generated by a fermion parity odd operator 
whose square is a scalar.
Hence, either $\tilde P_{i,j}$ is a scalar, or $\tilde P_{i,j}$ is a linear combination of a scalar
and $\tau_{j \rightarrow i} \tau_{i \rightarrow j}$ 
where $\tau_{j \rightarrow i}$ and $\tau_{i \rightarrow j}$ 
are both fermion odd operators whose square is a scalar and which are supported $i$ and $j$, respectively.
Hence, since $\tilde P_{i,j}$ is a projector, either $\tilde P_{i,j}=0$ or $\tilde P_{i,j}=1$ 
or $\tilde P_{i,j}= \frac 1 2 (1 + i\tau_{j \rightarrow i} \tau_{i \rightarrow j})$;
if $\tilde P_{i,j}= \frac 1 2 (1 - i\tau_{j \rightarrow i} \tau_{i \rightarrow j})$ 
then we could redefine one of $\tau$ with a minus sign.

Note that composing these two unitaries (the unitary in the first flow equation and that used in the second flow equation) gives some unitary that we write $U_{i,j}$.

We then apply this simplification to all other pairs of sites $(k,l)$.
Thus we arrive at \cref{fermionicbv}.  
In the theorem, $\Pi_i$ are projectors on each site 
that enforce the restrictions to the eigenspaces of all the different $V$ used in the construction.

\subsection{Algebraic approach}

\begin{definition}[\cite{Wall1964GradedBrauer}]
An algebra is {\bf central} over a field $\FF$ if its center is $\FF$.
An algebra is {\bf simple} if it does not have any proper (nonzero and nonunit) two-sided ideal.
A {\bf ($\ZZ_2$-)graded} algebra is an algebra $\mathcal{A} = \mathcal{A}_0 \oplus \mathcal{A}_1$ (even and odd parts)
where $\mathcal{A}_j + \mathcal{A}_j = \mathcal{A}_j$ 
but $\mathcal{A}_i \mathcal{A}_j \subseteq \mathcal{A}_{i+j \mod 2}$.
A graded algebra is {\bf simple} if it has no graded proper ideal.
A graded algebra is {\bf central} over $\FF$ if the even part of the center is $\FF$.
A {\bf graded tensor product}, which we denote by $\grotimes$, of two graded algebras
is an ordinary tensor product as a vector space over $\FF$ with the obvious $\ZZ_2$ grading,
but we modify the multiplication by {\bf supercommutativity}:
\begin{align}
(a_i \grotimes b_j )(a'_{i'} \grotimes b'_{j'}) := (-1)^{i' j} a_i a'_{i'} \grotimes b_j b'_{j'}
\end{align}
where $i,j,i',j' = 0,1$ denote the grading.
\end{definition}
The two notions for an algebra to be central and to be simple are independent:
The algebra of complex numbers is simple but not central over the field of real numbers;
the algebra of upper triangular matrices is central over any field, but not simple.
The center $Z(\mathcal{A})$ of a graded algebra $\mathcal{A}$ is graded naturally.
A central graded algebra may have a nonzero odd part in the center.
The graded product of central simple graded algebras is central simple~\cite[Thm.~2]{Wall1964GradedBrauer}.

A physically important example is the complex Clifford algebra generated 
by anti-commuting Majorana operators $\gamma_1,\ldots, \gamma_n$.
This algebra can be regarded as the graded tensor product of 
$n$ two-dimensional central simple graded algebras $\CC+ \gamma_j \CC$.
Note that there exists the Schmidt decomposition 
for any complex bipartite operator,
as the Schmidt decomposition only cares about vector space structure.
We will think of the grading given by the fermion parity as in the algebra of Majorana operators,
but the definition of graded algebras covers systems with both fermions and qudits present.

\begin{lemma}\label{lem:daggercentral-to-simple}
Let $\FF = \RR$ or $\FF = \CC$. Suppose $\mathcal{A}$ is a central $\ZZ_2$-graded algebra 
over $\FF$ acting on an inner product space that is closed under adjoint $\dagger$.
Then, $\mathcal{A}$ is simple as a graded algebra.
\end{lemma}
\begin{proof}
Due to the $\dagger$-closedness, the algebra $\mathcal{A}$ as an ungraded algebra is semi-simple.
If $\mathcal{A}$ has trivial center with zero odd component,
then $\mathcal{A}$ is simple as an ungraded algebra, and hence simple as a graded algebra.
If the center $Z(\mathcal{A}) = Z_0 \oplus Z_1$ of $\mathcal{A}$ has a nonzero odd component $Z_1$,
we claim that $Z_1$ is one-dimensional.
For any $\gamma \in Z_1$, $\gamma^\dagger \gamma$ is positive semidefinite, and is thus nonzero. 
But $\gamma^\dagger \gamma$ belongs to $Z_0 = \FF$, and thus is a positive real.
Rescaling, we may assume $\gamma^\dagger  = \gamma$ and $\gamma^2 = 1$.
For any other $\gamma' \in Z_1$, $\gamma \gamma' = a \in \FF$,
so $\gamma' = a \gamma$; hence $Z_1 = \gamma \FF$.
Considering projectors $(1 \pm \gamma) /2$
we see that the ungraded $\mathcal{A}$ is a direct sum of two simple algebras.

Take any graded two-sided ideal $J = J_0 + J_1$ of $\mathcal{A}$.
Since $J$ is an ideal of the ungraded $\mathcal{A}$, if proper, 
$J$ has to be one of the simple summand of $\mathcal{A}$.
Without loss of generality, we put $J = (1+\gamma) \mathcal{A} /2$.
Then, the even component $J_0$ includes the even part of $(1+\gamma)\mathcal{A}_0/2$, which is $\mathcal{A}_0$.
Likewise, the odd part $J_1$ includes the odd part of $(1+\gamma)\mathcal{A}_1/2$, which is $\mathcal{A}_1$.
Therefore, $J$ is not proper.
\end{proof}

Wall~\cite{Wall1964GradedBrauer} gives complete invariants for graded central simple algebras.
Over $\CC$, a central simple graded algebra is either $(+)$, 
meaning that the algebra is central simple as an ungraded algebra, 
or $(-)$, 
meaning that the algebra's even part is central simple.
The graded algebra's isomorphism class is determined by, in $(+)$ case, the dimensions of even and odd parts,
and in $(-)$ case, the dimension of the even part.
\begin{remark}\label{rem:grrep}
A faithful representation of a graded central simple algebras is as follows.
In the $(+)$ case, it is the algebra of all matrices of form $\begin{pmatrix} A & B \\ C & D \end{pmatrix}$
where the block diagonal subalgebra of $A,D$ is the even part,
and the off diagonal subspace of $B,C$ is the odd part.
In the $(-)$ case, it is the algebra of all matrices of form $\begin{pmatrix} A & B \\ B & A \end{pmatrix}$
(each block $A$ or $B$ is repeated),
where the block diagonal submatrix $A$ is the even part, and off diagonal submatrix $B$ is the odd part.
The faithful representation for $(+)$ is irreducible,
but not for $(-)$.
\end{remark}
A useful fact is that any two embeddings of a ungraded central simple algebra 
into another central simple algebra are related by an inner automorphism;
in particular, any automorphism of an ungraded central simple algebra is inner.

Let us fix some elements of a graded central simple algebra $\mathcal{D}$.
If $\mathcal{D}$ is simple as an ungraded algebra (the type $(+)$),
then there exists $u_\mathcal{D} \in \mathcal{D}$ such that 
$u_\mathcal{D} x u_\mathcal{D}^{-1}$ is $x$ if $x \in \mathcal{D}_0$ and $-x$ if $x \in \mathcal{D}_1$;
such $u_\mathcal{D}$ must exist because it gives an inner automorphism of $\mathcal{D}$.
The chosen $u_\mathcal{D}$ has to be even since $u_\mathcal{D} u_\mathcal{D} u_\mathcal{D}^{-1} = u_\mathcal{D}$.
Since $u_\mathcal{D}^2$ is in the even center of $\mathcal{B}$, it is a scalar, which we choose to be~$1$.
If $\mathcal{D}_0$ is central simple (the type $(-)$),
then the center $Z(\mathcal{D})$ has to be $\FF \oplus v_\mathcal{D} \FF$ where $v_\mathcal{D}^2 = 1$.
We see $\mathcal{D}_1 = v_\mathcal{D}^2 \mathcal{D}_1 \subseteq v_\mathcal{D} \mathcal{D}_0$, but $v_\mathcal{D} \mathcal{D}_0 \subseteq \mathcal{D}_1$, so they are equal.
For example, in the algebra of Majorana operators (the Clifford algebra)
$u_\mathcal{D}$ or $v_\mathcal{D}$ is the product of all Majorana operators.

If $\mathcal{A}$ is a graded algebra and $\mathcal{B}$ is a graded subalgebra of $\mathcal{A}$,
then the {\bf supercommutant} of $\mathcal{B}$ in $\mathcal{A}$ is defined to be the direct sum $\mathcal{B}' = \mathcal{B}'_0 \oplus \mathcal{B}'_1$
where
\begin{align*}
\mathcal{B}'_0 &= \{ x \in \mathcal{A}_0 ~:~ xb -b x = 0 \quad \forall b \in \mathcal{B} \},\\
\mathcal{B}'_1 &= \{ x \in \mathcal{A}_1 ~:~ xb_1 + b_1 x = x b_0 - b_0 x = 0 \quad \forall b_0,b_1 \in \mathcal{B}_1 \}.
\end{align*}
The supercommutant is a graded algebra.
Note that if $\mathcal{B}$ is central simple with a nonzero odd part,
then it suffices to consider the condition $b_1 x + x b_1 = 0$ to compute $\mathcal{B}'_1$ 
since $\mathcal{B}_1^2 = \mathcal{B}_0$~\cite[Lem.~1]{Wall1964GradedBrauer}.
For example, in the algebra of Majorana operators $\gamma_1,\ldots,\gamma_n$,
the supercommutant of the subalgebra of $\gamma_1,\ldots,\gamma_k$ is 
the subalgebra of $\gamma_{k+1},\ldots, \gamma_n$.

\begin{lemma}\label{lem:supercommutant}
Let $\mathcal{A}$ be a finite dimensional $\dagger$-closed graded central algebra over $\FF=\CC$
acting on an inner product space, and $\mathcal{B}$ be a $\dagger$-closed central graded subalgebra of $\mathcal{A}$.
Then, $\mathcal{A}$ is the graded tensor product of $\mathcal{B}$ and its supercommutant,
and the supercommutant is also $\dagger$-closed central as a graded algebra.
\end{lemma}
This is an analogue of the statement for the finite dimensional ungraded case:
\emph{
If $\mathcal{A}$ is a finite dimensional central simple algebra over $\CC$, 
and $\mathcal{B}$ a central simple subalgebra of $\mathcal{A}$,
then $\mathcal{A}$ is isomorphic to the tensor product of $\mathcal{B}$ and its commutant.
}
\begin{proof}
We use elements $u_\mathcal{A},u_\mathcal{B},v_\mathcal{A},v_\mathcal{B}$ whenever they exist.
By \cref{lem:daggercentral-to-simple} both $\mathcal{A}$ and $\mathcal{B}$ are simple.

(1) If $\mathcal{A}$ is $(+)$ and $\mathcal{B}$ is $(+)$, then 
consider the commutant $\mathcal{C}$ of $\mathcal{B}$ in $\mathcal{A}$.
$\mathcal{C}$ has grading naturally: $\mathcal{C} = (\mathcal{C} \cap \mathcal{A}_0) \oplus (\mathcal{C} \cap \mathcal{A}_1)$.
We find $\mathcal{B}' = \mathcal{C}_0 \oplus (\mathcal{C}_1 u_\mathcal{B})$.
(The even part $\mathcal{B}'_0 = \mathcal{C}_0$ is by definition. 
For any odd element $f \in \mathcal{B}'_1$, we see $u_\mathcal{B} f \in \mathcal{A}_1$ commutes with $\mathcal{B}_1$.)
We need to check if the natural map $\varphi : \mathcal{B} \grotimes \mathcal{B}' \to \mathcal{A}$ is injective;
the map being a homomorphism is clear by definition of $\grotimes$, and
the surjectivity will follow by dimension counting, since we know $\mathcal{A} \cong \mathcal{B} \otimes \mathcal{C}$ as ungraded algebras.
Thanks to inner automorphisms $u_\mathcal{A}$ and $u_\mathcal{B}$,
we see $\mathcal{B} \grotimes \mathcal{B}' = \bigoplus_{j,k = 0,1} \mathcal{B}_j \otimes \mathcal{B}'_k$,
but since each direct summand is isomorphic to $\mathcal{B}_j \otimes \mathcal{C}_k$ as $\CC$-vector spaces,
and the ungraded product map $\mathcal{B} \otimes \mathcal{C} \to \mathcal{A}$ is injective,
we see that $\varphi$ is injective.

(2) If $\mathcal{A}$ is $(-)$ and $\mathcal{B}$ is $(+)$, then 
$\mathcal{A} = \mathcal{A}_0 \oplus v_\mathcal{A} \mathcal{A}_0$ with $\mathcal{A}_0$ central simple.
Let $\mathcal{C}$ be the commutant of $\mathcal{B}$ in $\mathcal{A}$.
We find $\mathcal{B}' = \mathcal{C}_0 \oplus (\mathcal{C}_1 u_\mathcal{B})$ for the same reason as in (1).
Since $v_\mathcal{A} \in Z(\mathcal{A})$, $\mathcal{C}_1 u_\mathcal{B} v_\mathcal{A} \subseteq \mathcal{A}_0$ commutes with $\mathcal{B}$,
and hence is equal to $\mathcal{C}_0$.
That is, $\mathcal{B}' = \mathcal{C}_0 \oplus v_\mathcal{A} \mathcal{C}_0$.

Consider $\tilde{\mathcal{B}} = \mathcal{B}_0 + v_\mathcal{A} \mathcal{B}_1 \subseteq \mathcal{A}_0$.
The sum $\mathcal{B}_0 + v_\mathcal{A} \mathcal{B}_1$ is in fact direct, 
since $v_\mathcal{A}$ is in the center of $\mathcal{A}$ 
and $u_\mathcal{B}$ can tell $\mathcal{B}_0$ from $\mathcal{B}_1$.
Moreover, $\tilde{\mathcal{B}}$ is $\dagger$-closed, and if there is a central element $b_0 + b_1 v_\mathcal{A}$
then $b_0 + b_1$ is a central element of $\mathcal{B}$, so the ungraded $\tilde{\mathcal{B}}$ is central, and hence also simple.
Since $v_\mathcal{A} \in Z(\mathcal{A})$,
the commutant of $\tilde{\mathcal{B}}$ in $\mathcal{A}_0$ is precisely $\mathcal{C}_0$, which in turn has to be central simple as an ungraded algebra.
So we have the ungraded isomorphism $\tilde{\mathcal{B}} \otimes \mathcal{C}_0 \to \mathcal{A}_0$;
in particular, $\dim \mathcal{B} \otimes \mathcal{C}_0 = \dim \mathcal{A}_0$,
and thus $\dim \mathcal{B} \otimes \mathcal{B}' = \dim \mathcal{A}$.

It remains only to show that the map $\mathcal{B} \grotimes \mathcal{B}' \to \mathcal{A}$ is injective,
for the same reason as in~(1).
Due to the grading of $\mathcal{A}$ that is given, and the grading of $\mathcal{B}$ by $u_\mathcal{B}$,
it suffices to consider $\mathcal{B}_i \otimes \mathcal{B}'_j \to \mathcal{A}_{i+j \mod 2}$,
whose injectivity follows from that of $\tilde{\mathcal{B}} \otimes \mathcal{C}_0 \to \mathcal{A}_0$. 

(3) If $\mathcal{A}$ is $(+)$ and $\mathcal{B}$ is~$(-)$, then 
we introduce an auxiliary central simple algebra $\mathcal{E} = \CC \oplus \CC \gamma$ of type $(-)$. 
Consider $\mathcal{E} \grotimes \mathcal{A}$ and its subalgebra $\mathcal{E} \grotimes \mathcal{B}$,
We then in the situation of case (2), 
and the desired result follows by taking the supercommutant of $\mathcal{E}$.

(4) If $\mathcal{A}$ is $(-)$ and $\mathcal{B}$ is $(-)$, 
then similarly we take $\mathcal{E} \grotimes \mathcal{A} \supseteq \mathcal{E} \grotimes \mathcal{B}$,
to use the argument of case~(1).

Finally, it is clear that $\mathcal{B}'$ is $\dagger$-closed.
The centrality follows because any even central element of $\mathcal{B}'$
is a central element of $\mathcal{A}$.
\end{proof}

\begin{proof}[Proof of \cref{fermionicbv}]
The argument is parallel to that of \cite{bravyi2005commutative} 
using \cref{lem:supercommutant}
in place of the ungraded version that we have noted in between the statement of the lemma and its proof.

The interaction algebra $\mathcal A_{i \to j}$ is a $\dagger$-closed graded algebra over $\CC$, 
since $P_{i,j}$ is hermitian.
If $\mathcal A_{i \to j}$ is not central,
there would be an even subalgebra that commutes with all other interaction algebras at site $j$.
Choosing an even minimal projector in the center of $\mathcal A_{i \to j}$ for each neighbor $i$ of $j$,
we construct a projector $\Pi'_j$ at site $j$.
Restricting to the image of $\Pi'_j$,
we may assume that $\mathcal A_{i \to j}$ is central, for any $i$ and $j$.
The same remark goes into the supercommutant of all interaction algebras at site $j$,
and hence we may assume that the algebra generated by all interaction algebras at site $j$ is central,
for all $j$.

By inductively applying \cref{lem:supercommutant}
we see that the algebra generated by all interaction algebras at site $j$
is a graded tensor product of $\mathcal A_{i \to j}$ 
where $i$ ranges over all neighboring sites of $j$.
So, the problem is reduced to an instance where there are only two sites and one Hamitonian term $P$.
Let $\mathcal A_L$ and $\mathcal A_R$ be the interaction algebras of the term on the left and right sites,
respectively.

Choose any nonzero even minimal projectors $Q_L \in \mathcal A_L$ and $Q_R \in \mathcal A_R$,
where the minimality means that $Q_{L(R)} O Q_{L(R)} \propto Q_{L(R)}$ for any even $O \in \mathcal A_{L(R)}$.
(The minimality is a way of speaking of ``rank one'' projectors without referring to a representation space,
which may not be physical on its own since $\mathcal A_{L,R}$ may be $(-)$,
in which case we have an unpaired Majorana mode.)
Let us think in terms of the faithful representations of \cref{rem:grrep}.
In every case out of four cases $(\mathcal A_L,\pm)$ and $(\mathcal A_R,\pm)$,
$Q_L \grotimes Q_R$ is some even projector, and 
$P$ can be conjugated by an even unitary $U_{LR} \in \mathcal A_L \grotimes \mathcal A_R$
to become $\tilde P$ that commutes with $Q_L \grotimes Q_R$.
Further restricting the Hilbert space of the two sites $L,R$ by $Q_L$ and $Q_R$,
the problem is reduced to the situation where $Q_L = 1$ and $Q_R = 1$.
Then, the only possibility is that $\tilde P$ is
either $0,1$ or $(1+i \tau_L \tau_R)/2$.

The projector $\Pi_i$ is the product of $\Pi'_i$ and $Q$'s that restricts the algebra on each site.
\end{proof}

\begin{acknowledgments}
We thank Michael Freedman for encouraging discussions
and for the suggestion to use the Hall marriage theorem.
JH thanks Vadym Kliuchnikov for suggesting Ref.~\cite{Kniga}.
LF is supported by NSF DMR-1519579.
\end{acknowledgments}

\bibliography{nta3-ref}

\begin{thebibliography}{59}%
\makeatletter
\providecommand \@ifxundefined [1]{%
 \@ifx{#1\undefined}
}%
\providecommand \@ifnum [1]{%
 \ifnum #1\expandafter \@firstoftwo
 \else \expandafter \@secondoftwo
 \fi
}%
\providecommand \@ifx [1]{%
 \ifx #1\expandafter \@firstoftwo
 \else \expandafter \@secondoftwo
 \fi
}%
\providecommand \natexlab [1]{#1}%
\providecommand \enquote  [1]{``#1''}%
\providecommand \bibnamefont  [1]{#1}%
\providecommand \bibfnamefont [1]{#1}%
\providecommand \citenamefont [1]{#1}%
\providecommand \href@noop [0]{\@secondoftwo}%
\providecommand \href [0]{\begingroup \@sanitize@url \@href}%
\providecommand \@href[1]{\@@startlink{#1}\@@href}%
\providecommand \@@href[1]{\endgroup#1\@@endlink}%
\providecommand \@sanitize@url [0]{\catcode `\\12\catcode `\$12\catcode
  `\&12\catcode `\#12\catcode `\^12\catcode `\_12\catcode `\%12\relax}%
\providecommand \@@startlink[1]{}%
\providecommand \@@endlink[0]{}%
\providecommand \url  [0]{\begingroup\@sanitize@url \@url }%
\providecommand \@url [1]{\endgroup\@href {#1}{\urlprefix }}%
\providecommand \urlprefix  [0]{URL }%
\providecommand \Eprint [0]{\href }%
\providecommand \doibase [0]{http://dx.doi.org/}%
\providecommand \selectlanguage [0]{\@gobble}%
\providecommand \bibinfo  [0]{\@secondoftwo}%
\providecommand \bibfield  [0]{\@secondoftwo}%
\providecommand \translation [1]{[#1]}%
\providecommand \BibitemOpen [0]{}%
\providecommand \bibitemStop [0]{}%
\providecommand \bibitemNoStop [0]{.\EOS\space}%
\providecommand \EOS [0]{\spacefactor3000\relax}%
\providecommand \BibitemShut  [1]{\csname bibitem#1\endcsname}%
\let\auto@bib@innerbib\@empty
%</preamble>
\bibitem [{\citenamefont {Bravyi}\ and\ \citenamefont
  {Vyalyi}(2005)}]{bravyi2005commutative}%
  \BibitemOpen
  \bibfield  {author} {\bibinfo {author} {\bibfnamefont {Sergey}\ \bibnamefont
  {Bravyi}}\ and\ \bibinfo {author} {\bibfnamefont {Mikhail}\ \bibnamefont
  {Vyalyi}},\ }\bibfield  {title} {\enquote {\bibinfo {title} {Commutative
  version of the local hamiltonian problem and common eigenspace problem},}\
  }\href@noop {} {\bibfield  {journal} {\bibinfo  {journal} {Quantum
  Information \& Computation}\ }\textbf {\bibinfo {volume} {5}},\ \bibinfo
  {pages} {187--215} (\bibinfo {year} {2005})}\BibitemShut {NoStop}%
\bibitem [{\citenamefont {Gross}\ \emph {et~al.}(2012)\citenamefont {Gross},
  \citenamefont {Nesme}, \citenamefont {Vogts},\ and\ \citenamefont
  {Werner}}]{Gross_2012}%
  \BibitemOpen
  \bibfield  {author} {\bibinfo {author} {\bibfnamefont {D.}~\bibnamefont
  {Gross}}, \bibinfo {author} {\bibfnamefont {V.}~\bibnamefont {Nesme}},
  \bibinfo {author} {\bibfnamefont {H.}~\bibnamefont {Vogts}}, \ and\ \bibinfo
  {author} {\bibfnamefont {R.~F.}\ \bibnamefont {Werner}},\ }\bibfield  {title}
  {\enquote {\bibinfo {title} {Index theory of one dimensional quantum walks
  and cellular automata},}\ }\href {\doibase 10.1007/s00220-012-1423-1}
  {\bibfield  {journal} {\bibinfo  {journal} {Communications in Mathematical
  Physics}\ }\textbf {\bibinfo {volume} {310}},\ \bibinfo {pages} {419--454}
  (\bibinfo {year} {2012})},\ \Eprint {http://arxiv.org/abs/0910.3675}
  {arXiv:0910.3675} \BibitemShut {NoStop}%
\bibitem [{\citenamefont {Po}\ \emph {et~al.}(2017)\citenamefont {Po},
  \citenamefont {Fidkowski}, \citenamefont {Vishwanath},\ and\ \citenamefont
  {Potter}}]{fermionGNVW1}%
  \BibitemOpen
  \bibfield  {author} {\bibinfo {author} {\bibfnamefont {Hoi~Chun}\
  \bibnamefont {Po}}, \bibinfo {author} {\bibfnamefont {Lukasz}\ \bibnamefont
  {Fidkowski}}, \bibinfo {author} {\bibfnamefont {Ashvin}\ \bibnamefont
  {Vishwanath}}, \ and\ \bibinfo {author} {\bibfnamefont {Andrew~C.}\
  \bibnamefont {Potter}},\ }\bibfield  {title} {\enquote {\bibinfo {title}
  {Radical chiral floquet phases in a periodically driven {Kitaev} model and
  beyond},}\ }\href {\doibase 10.1103/PhysRevB.96.245116} {\bibfield  {journal}
  {\bibinfo  {journal} {Phys. Rev. B}\ }\textbf {\bibinfo {volume} {96}},\
  \bibinfo {pages} {245116} (\bibinfo {year} {2017})}\BibitemShut {NoStop}%
\bibitem [{\citenamefont {{Fidkowski}}\ \emph {et~al.}(2017)\citenamefont
  {{Fidkowski}}, \citenamefont {{Po}}, \citenamefont {{Potter}},\ and\
  \citenamefont {{Vishwanath}}}]{fermionGNVW2}%
  \BibitemOpen
  \bibfield  {author} {\bibinfo {author} {\bibfnamefont {L.}~\bibnamefont
  {{Fidkowski}}}, \bibinfo {author} {\bibfnamefont {H.~C.}\ \bibnamefont
  {{Po}}}, \bibinfo {author} {\bibfnamefont {A.~C.}\ \bibnamefont {{Potter}}},
  \ and\ \bibinfo {author} {\bibfnamefont {A.}~\bibnamefont {{Vishwanath}}},\
  }\bibfield  {title} {\enquote {\bibinfo {title} {{Interacting invariants for
  Floquet phases of fermions in two dimensions}},}\ }\href@noop {} {\
  (\bibinfo {year} {2017})},\ \Eprint {http://arxiv.org/abs/1703.07360}
  {arXiv:1703.07360} \BibitemShut {NoStop}%
\bibitem [{\citenamefont {Kitaev}(2009)}]{Kitaev_2009}%
  \BibitemOpen
  \bibfield  {author} {\bibinfo {author} {\bibfnamefont {Alexei}\ \bibnamefont
  {Kitaev}},\ }\bibfield  {title} {\enquote {\bibinfo {title} {Periodic table
  for topological insulators and superconductors},}\ }\href {\doibase
  10.1063/1.3149495} {\bibfield  {journal} {\bibinfo  {journal} {{AIP
  Conference Proceedings}}\ }\textbf {\bibinfo {volume} {1134}},\ \bibinfo
  {pages} {22} (\bibinfo {year} {2009})},\ \Eprint
  {http://arxiv.org/abs/0901.2686} {arXiv:0901.2686} \BibitemShut {NoStop}%
\bibitem [{\citenamefont {Ryu}\ \emph {et~al.}(2010)\citenamefont {Ryu},
  \citenamefont {Schnyder}, \citenamefont {Furusaki},\ and\ \citenamefont
  {Ludwig}}]{Ryu_2010}%
  \BibitemOpen
  \bibfield  {author} {\bibinfo {author} {\bibfnamefont {Shinsei}\ \bibnamefont
  {Ryu}}, \bibinfo {author} {\bibfnamefont {Andreas~P}\ \bibnamefont
  {Schnyder}}, \bibinfo {author} {\bibfnamefont {Akira}\ \bibnamefont
  {Furusaki}}, \ and\ \bibinfo {author} {\bibfnamefont {Andreas W~W}\
  \bibnamefont {Ludwig}},\ }\bibfield  {title} {\enquote {\bibinfo {title}
  {Topological insulators and superconductors: tenfold way and dimensional
  hierarchy},}\ }\href {\doibase 10.1088/1367-2630/12/6/065010} {\bibfield
  {journal} {\bibinfo  {journal} {New Journal of Physics}\ }\textbf {\bibinfo
  {volume} {12}},\ \bibinfo {pages} {065010} (\bibinfo {year}
  {2010})}\BibitemShut {NoStop}%
\bibitem [{\citenamefont {Freedman}\ and\ \citenamefont
  {Hastings}(2018)}]{FreedmanHastings}%
  \BibitemOpen
  \bibfield  {author} {\bibinfo {author} {\bibfnamefont {M.~H.}\ \bibnamefont
  {Freedman}}\ and\ \bibinfo {author} {\bibfnamefont {M.~B.}\ \bibnamefont
  {Hastings}},\ }\href@noop {} {\  (\bibinfo {year} {2018})},\ \bibinfo {note}
  {in preparation}\BibitemShut {NoStop}%
\bibitem [{\citenamefont {Levin}\ and\ \citenamefont {Wen}(2005)}]{Levin_2005}%
  \BibitemOpen
  \bibfield  {author} {\bibinfo {author} {\bibfnamefont {Michael~A.}\
  \bibnamefont {Levin}}\ and\ \bibinfo {author} {\bibfnamefont {Xiao-Gang}\
  \bibnamefont {Wen}},\ }\bibfield  {title} {\enquote {\bibinfo {title}
  {String-net condensation:{\hspace{1em}}a physical mechanism for topological
  phases},}\ }\href {\doibase 10.1103/physrevb.71.045110} {\bibfield  {journal}
  {\bibinfo  {journal} {Physical Review B}\ }\textbf {\bibinfo {volume} {71}},\
  \bibinfo {pages} {045110} (\bibinfo {year} {2005})},\ \Eprint
  {http://arxiv.org/abs/cond-mat/0404617} {arXiv:cond-mat/0404617} \BibitemShut
  {NoStop}%
\bibitem [{\citenamefont {Kitaev}(2003)}]{Kitaev_2003}%
  \BibitemOpen
  \bibfield  {author} {\bibinfo {author} {\bibfnamefont {A.Yu.}\ \bibnamefont
  {Kitaev}},\ }\bibfield  {title} {\enquote {\bibinfo {title} {Fault-tolerant
  quantum computation by anyons},}\ }\href {\doibase
  10.1016/s0003-4916(02)00018-0} {\bibfield  {journal} {\bibinfo  {journal}
  {Annals of Physics}\ }\textbf {\bibinfo {volume} {303}},\ \bibinfo {pages}
  {2--30} (\bibinfo {year} {2003})},\ \Eprint
  {http://arxiv.org/abs/quant-ph/9707021} {arXiv:quant-ph/9707021} \BibitemShut
  {NoStop}%
\bibitem [{\citenamefont {Freedman}\ \emph {et~al.}(2002)\citenamefont
  {Freedman}, \citenamefont {Meyer},\ and\ \citenamefont
  {Luo}}]{Freedman_2002}%
  \BibitemOpen
  \bibfield  {author} {\bibinfo {author} {\bibfnamefont {Michael}\ \bibnamefont
  {Freedman}}, \bibinfo {author} {\bibfnamefont {David}\ \bibnamefont {Meyer}},
  \ and\ \bibinfo {author} {\bibfnamefont {Feng}\ \bibnamefont {Luo}},\
  }\bibfield  {title} {\enquote {\bibinfo {title} {Z2-systolic freedom and
  quantum codes},}\ }in\ \href {\doibase 10.1201/9781420035377.ch12} {\emph
  {\bibinfo {booktitle} {Computational Mathematics}}}\ (\bibinfo  {publisher}
  {Chapman and Hall/{CRC}},\ \bibinfo {year} {2002})\BibitemShut {NoStop}%
\bibitem [{\citenamefont {Dennis}\ \emph {et~al.}(2002)\citenamefont {Dennis},
  \citenamefont {Kitaev}, \citenamefont {Landahl},\ and\ \citenamefont
  {Preskill}}]{Dennis_2002}%
  \BibitemOpen
  \bibfield  {author} {\bibinfo {author} {\bibfnamefont {Eric}\ \bibnamefont
  {Dennis}}, \bibinfo {author} {\bibfnamefont {Alexei}\ \bibnamefont {Kitaev}},
  \bibinfo {author} {\bibfnamefont {Andrew}\ \bibnamefont {Landahl}}, \ and\
  \bibinfo {author} {\bibfnamefont {John}\ \bibnamefont {Preskill}},\
  }\bibfield  {title} {\enquote {\bibinfo {title} {Topological quantum
  memory},}\ }\href {\doibase 10.1063/1.1499754} {\bibfield  {journal}
  {\bibinfo  {journal} {Journal of Mathematical Physics}\ }\textbf {\bibinfo
  {volume} {43}},\ \bibinfo {pages} {4452--4505} (\bibinfo {year} {2002})},\
  \Eprint {http://arxiv.org/abs/quant-ph/0110143} {arXiv:quant-ph/0110143}
  \BibitemShut {NoStop}%
\bibitem [{\citenamefont {Dijkgraaf}\ and\ \citenamefont
  {Witten}(1990)}]{dijkgraaf1990topological}%
  \BibitemOpen
  \bibfield  {author} {\bibinfo {author} {\bibfnamefont {Robbert}\ \bibnamefont
  {Dijkgraaf}}\ and\ \bibinfo {author} {\bibfnamefont {Edward}\ \bibnamefont
  {Witten}},\ }\bibfield  {title} {\enquote {\bibinfo {title} {Topological
  gauge theories and group cohomology},}\ }\href {\doibase 10.1007/BF02096988}
  {\bibfield  {journal} {\bibinfo  {journal} {Communications in Mathematical
  Physics}\ }\textbf {\bibinfo {volume} {129}},\ \bibinfo {pages} {393--429}
  (\bibinfo {year} {1990})}\BibitemShut {NoStop}%
\bibitem [{\citenamefont {Freedman}\ and\ \citenamefont
  {Hastings}(2016)}]{Freedman_2016}%
  \BibitemOpen
  \bibfield  {author} {\bibinfo {author} {\bibfnamefont {Michael~H.}\
  \bibnamefont {Freedman}}\ and\ \bibinfo {author} {\bibfnamefont {Matthew~B.}\
  \bibnamefont {Hastings}},\ }\bibfield  {title} {\enquote {\bibinfo {title}
  {Double semions in arbitrary dimension},}\ }\href {\doibase
  10.1007/s00220-016-2604-0} {\bibfield  {journal} {\bibinfo  {journal}
  {Communications in Mathematical Physics}\ }\textbf {\bibinfo {volume}
  {347}},\ \bibinfo {pages} {389--419} (\bibinfo {year} {2016})},\ \Eprint
  {http://arxiv.org/abs/1507.05676} {arXiv:1507.05676} \BibitemShut {NoStop}%
\bibitem [{\citenamefont {Haah}(2011)}]{Haah_2011}%
  \BibitemOpen
  \bibfield  {author} {\bibinfo {author} {\bibfnamefont {Jeongwan}\
  \bibnamefont {Haah}},\ }\bibfield  {title} {\enquote {\bibinfo {title} {Local
  stabilizer codes in three dimensions without string logical operators},}\
  }\href {\doibase 10.1103/physreva.83.042330} {\bibfield  {journal} {\bibinfo
  {journal} {Physical Review A}\ }\textbf {\bibinfo {volume} {83}},\ \bibinfo
  {pages} {042330} (\bibinfo {year} {2011})},\ \Eprint
  {http://arxiv.org/abs/1101.1962} {arXiv:1101.1962} \BibitemShut {NoStop}%
\bibitem [{\citenamefont {Walker}\ and\ \citenamefont
  {Wang}(2011)}]{Walker_2011}%
  \BibitemOpen
  \bibfield  {author} {\bibinfo {author} {\bibfnamefont {Kevin}\ \bibnamefont
  {Walker}}\ and\ \bibinfo {author} {\bibfnamefont {Zhenghan}\ \bibnamefont
  {Wang}},\ }\bibfield  {title} {\enquote {\bibinfo {title} {(3+1)-{TQFTs} and
  topological insulators},}\ }\href {\doibase 10.1007/s11467-011-0194-z}
  {\bibfield  {journal} {\bibinfo  {journal} {Frontiers of Physics}\ }\textbf
  {\bibinfo {volume} {7}},\ \bibinfo {pages} {150--159} (\bibinfo {year}
  {2011})},\ \Eprint {http://arxiv.org/abs/1104.2632} {arXiv:1104.2632}
  \BibitemShut {NoStop}%
\bibitem [{\citenamefont {von Keyserlingk}\ \emph {et~al.}(2013)\citenamefont
  {von Keyserlingk}, \citenamefont {Burnell},\ and\ \citenamefont
  {Simon}}]{Curt}%
  \BibitemOpen
  \bibfield  {author} {\bibinfo {author} {\bibfnamefont {C.~W.}\ \bibnamefont
  {von Keyserlingk}}, \bibinfo {author} {\bibfnamefont {F.~J.}\ \bibnamefont
  {Burnell}}, \ and\ \bibinfo {author} {\bibfnamefont {S.~H.}\ \bibnamefont
  {Simon}},\ }\bibfield  {title} {\enquote {\bibinfo {title} {Three-dimensional
  topological lattice models with surface anyons},}\ }\href {\doibase
  10.1103/PhysRevB.87.045107} {\bibfield  {journal} {\bibinfo  {journal} {Phys.
  Rev. B}\ }\textbf {\bibinfo {volume} {87}},\ \bibinfo {pages} {045107}
  (\bibinfo {year} {2013})},\ \Eprint {http://arxiv.org/abs/1208.5128}
  {arXiv:1208.5128} \BibitemShut {NoStop}%
\bibitem [{\citenamefont {Haah}(2013)}]{Haah2013}%
  \BibitemOpen
  \bibfield  {author} {\bibinfo {author} {\bibfnamefont {Jeongwan}\
  \bibnamefont {Haah}},\ }\bibfield  {title} {\enquote {\bibinfo {title}
  {Commuting pauli hamiltonians as maps between free modules},}\ }\href
  {\doibase 10.1007/s00220-013-1810-2} {\bibfield  {journal} {\bibinfo
  {journal} {Commun. Math. Phys.}\ }\textbf {\bibinfo {volume} {324}},\
  \bibinfo {pages} {351--399} (\bibinfo {year} {2013})},\ \Eprint
  {http://arxiv.org/abs/1204.1063} {arXiv:1204.1063} \BibitemShut {NoStop}%
\bibitem [{\citenamefont {Fr{\"o}hlich}\ and\ \citenamefont
  {Gabbiani}(1990)}]{frohlich1990braid}%
  \BibitemOpen
  \bibfield  {author} {\bibinfo {author} {\bibfnamefont {J{\"u}rg}\
  \bibnamefont {Fr{\"o}hlich}}\ and\ \bibinfo {author} {\bibfnamefont
  {Fabrizio}\ \bibnamefont {Gabbiani}},\ }\bibfield  {title} {\enquote
  {\bibinfo {title} {Braid statistics in local quantum theory},}\ }\href
  {\doibase 10.1142/S0129055X90000107} {\bibfield  {journal} {\bibinfo
  {journal} {Reviews in Mathematical Physics}\ }\textbf {\bibinfo {volume}
  {2}},\ \bibinfo {pages} {251--353} (\bibinfo {year} {1990})}\BibitemShut
  {NoStop}%
\bibitem [{\citenamefont {Rehren}(1989)}]{rehren1989braid}%
  \BibitemOpen
  \bibfield  {author} {\bibinfo {author} {\bibfnamefont {Karl-Henning}\
  \bibnamefont {Rehren}},\ }\bibfield  {title} {\enquote {\bibinfo {title}
  {Braid group statistics and their superselection rules},}\ }\href
  {http://www.theorie.physik.uni-goe.de/papers/rehren/89/braid_group_statistics.pdf}
  {\bibfield  {journal} {\bibinfo  {journal} {The algebraic theory of
  superselection sectors. Palermo}\ ,\ \bibinfo {pages} {333--355}} (\bibinfo
  {year} {1989})}\BibitemShut {NoStop}%
\bibitem [{\citenamefont {{Kitaev}}(2006)}]{Kitaev_2005}%
  \BibitemOpen
  \bibfield  {author} {\bibinfo {author} {\bibfnamefont {A.}~\bibnamefont
  {{Kitaev}}},\ }\bibfield  {title} {\enquote {\bibinfo {title} {{Anyons in an
  exactly solved model and beyond}},}\ }\href {\doibase
  10.1016/j.aop.2005.10.005} {\bibfield  {journal} {\bibinfo  {journal} {Annals
  of Physics}\ }\textbf {\bibinfo {volume} {321}},\ \bibinfo {pages} {2--111}
  (\bibinfo {year} {2006})},\ \Eprint {http://arxiv.org/abs/cond-mat/0506438}
  {arXiv:cond-mat/0506438} \BibitemShut {NoStop}%
\bibitem [{\citenamefont {Hastings}(2013)}]{hastings2013classifying}%
  \BibitemOpen
  \bibfield  {author} {\bibinfo {author} {\bibfnamefont {M.~B.}\ \bibnamefont
  {Hastings}},\ }\bibfield  {title} {\enquote {\bibinfo {title} {Classifying
  quantum phases with the {K}irby torus trick},}\ }\href {\doibase
  10.1103/PhysRevB.88.165114} {\bibfield  {journal} {\bibinfo  {journal}
  {Physical Review B}\ }\textbf {\bibinfo {volume} {88}},\ \bibinfo {pages}
  {165114} (\bibinfo {year} {2013})},\ \Eprint {http://arxiv.org/abs/1305.6625}
  {arXiv:1305.6625} \BibitemShut {NoStop}%
\bibitem [{\citenamefont {Aigner}\ and\ \citenamefont
  {Ziegler}(2001)}]{BookProof}%
  \BibitemOpen
  \bibfield  {author} {\bibinfo {author} {\bibfnamefont {Martin}\ \bibnamefont
  {Aigner}}\ and\ \bibinfo {author} {\bibfnamefont {G\"unter~M.}\ \bibnamefont
  {Ziegler}},\ }\href@noop {} {\emph {\bibinfo {title} {Proofs from {THE
  BOOK}}}},\ \bibinfo {edition} {2nd}\ ed.\ (\bibinfo  {publisher}
  {Springer-Verlag},\ \bibinfo {year} {2001})\BibitemShut {NoStop}%
\bibitem [{\citenamefont {Burnell}\ \emph {et~al.}(2014)\citenamefont
  {Burnell}, \citenamefont {Chen}, \citenamefont {Fidkowski},\ and\
  \citenamefont {Vishwanath}}]{BCFV}%
  \BibitemOpen
  \bibfield  {author} {\bibinfo {author} {\bibfnamefont {F.~J.}\ \bibnamefont
  {Burnell}}, \bibinfo {author} {\bibfnamefont {Xie}\ \bibnamefont {Chen}},
  \bibinfo {author} {\bibfnamefont {Lukasz}\ \bibnamefont {Fidkowski}}, \ and\
  \bibinfo {author} {\bibfnamefont {Ashvin}\ \bibnamefont {Vishwanath}},\
  }\bibfield  {title} {\enquote {\bibinfo {title} {Exactly soluble model of a
  {3D} symmetry protected topological phase of bosons with surface topological
  order},}\ }\href {\doibase 10.1103/PhysRevB.90.245122} {\bibfield  {journal}
  {\bibinfo  {journal} {Phys. Rev. B}\ }\textbf {\bibinfo {volume} {90}},\
  \bibinfo {pages} {245122} (\bibinfo {year} {2014})},\ \Eprint
  {http://arxiv.org/abs/1302.7072} {arXiv:1302.7072} \BibitemShut {NoStop}%
\bibitem [{\citenamefont {Fidkowski}\ \emph {et~al.}(2013)\citenamefont
  {Fidkowski}, \citenamefont {Chen},\ and\ \citenamefont
  {Vishwanath}}]{TopoSC}%
  \BibitemOpen
  \bibfield  {author} {\bibinfo {author} {\bibfnamefont {Lukasz}\ \bibnamefont
  {Fidkowski}}, \bibinfo {author} {\bibfnamefont {Xie}\ \bibnamefont {Chen}}, \
  and\ \bibinfo {author} {\bibfnamefont {Ashvin}\ \bibnamefont {Vishwanath}},\
  }\bibfield  {title} {\enquote {\bibinfo {title} {Non-abelian topological
  order on the surface of a 3{D} topological superconductor from an exactly
  solved model},}\ }\href {\doibase 10.1103/PhysRevX.3.041016} {\bibfield
  {journal} {\bibinfo  {journal} {Phys. Rev. X}\ }\textbf {\bibinfo {volume}
  {3}},\ \bibinfo {pages} {041016} (\bibinfo {year} {2013})},\ \Eprint
  {http://arxiv.org/abs/1305.5851} {arXiv:1305.5851} \BibitemShut {NoStop}%
\bibitem [{\citenamefont {Chen}\ \emph {et~al.}(2014)\citenamefont {Chen},
  \citenamefont {Fidkowski},\ and\ \citenamefont {Vishwanath}}]{TI}%
  \BibitemOpen
  \bibfield  {author} {\bibinfo {author} {\bibfnamefont {Xie}\ \bibnamefont
  {Chen}}, \bibinfo {author} {\bibfnamefont {Lukasz}\ \bibnamefont
  {Fidkowski}}, \ and\ \bibinfo {author} {\bibfnamefont {Ashvin}\ \bibnamefont
  {Vishwanath}},\ }\bibfield  {title} {\enquote {\bibinfo {title} {Symmetry
  enforced non-abelian topological order at the surface of a topological
  insulator},}\ }\href {\doibase 10.1103/PhysRevB.89.165132} {\bibfield
  {journal} {\bibinfo  {journal} {Phys. Rev. B}\ }\textbf {\bibinfo {volume}
  {89}},\ \bibinfo {pages} {165132} (\bibinfo {year} {2014})},\ \Eprint
  {http://arxiv.org/abs/1306.3250} {arXiv:1306.3250} \BibitemShut {NoStop}%
\bibitem [{\citenamefont {Chen}\ \emph {et~al.}(2015)\citenamefont {Chen},
  \citenamefont {Burnell}, \citenamefont {Vishwanath},\ and\ \citenamefont
  {Fidkowski}}]{ProjS}%
  \BibitemOpen
  \bibfield  {author} {\bibinfo {author} {\bibfnamefont {Xie}\ \bibnamefont
  {Chen}}, \bibinfo {author} {\bibfnamefont {F.~J.}\ \bibnamefont {Burnell}},
  \bibinfo {author} {\bibfnamefont {Ashvin}\ \bibnamefont {Vishwanath}}, \ and\
  \bibinfo {author} {\bibfnamefont {Lukasz}\ \bibnamefont {Fidkowski}},\
  }\bibfield  {title} {\enquote {\bibinfo {title} {Anomalous symmetry
  fractionalization and surface topological order},}\ }\href {\doibase
  10.1103/PhysRevX.5.041013} {\bibfield  {journal} {\bibinfo  {journal} {Phys.
  Rev. X}\ }\textbf {\bibinfo {volume} {5}},\ \bibinfo {pages} {041013}
  (\bibinfo {year} {2015})},\ \Eprint {http://arxiv.org/abs/1403.6491}
  {arXiv:1403.6491} \BibitemShut {NoStop}%
\bibitem [{\citenamefont {Lu}\ and\ \citenamefont {Vishwanath}(2012)}]{VL}%
  \BibitemOpen
  \bibfield  {author} {\bibinfo {author} {\bibfnamefont {Yuan-Ming}\
  \bibnamefont {Lu}}\ and\ \bibinfo {author} {\bibfnamefont {Ashvin}\
  \bibnamefont {Vishwanath}},\ }\bibfield  {title} {\enquote {\bibinfo {title}
  {Theory and classification of interacting integer topological phases in two
  dimensions: A chern-simons approach},}\ }\href {\doibase
  10.1103/PhysRevB.86.125119} {\bibfield  {journal} {\bibinfo  {journal} {Phys.
  Rev. B}\ }\textbf {\bibinfo {volume} {86}},\ \bibinfo {pages} {125119}
  (\bibinfo {year} {2012})},\ \Eprint {http://arxiv.org/abs/1205.3156}
  {arXiv:1205.3156} \BibitemShut {NoStop}%
\bibitem [{\citenamefont {Bravyi}\ \emph {et~al.}(2010)\citenamefont {Bravyi},
  \citenamefont {Hastings},\ and\ \citenamefont
  {Michalakis}}]{BravyiHastingsMichalakis2010stability}%
  \BibitemOpen
  \bibfield  {author} {\bibinfo {author} {\bibfnamefont {Sergey}\ \bibnamefont
  {Bravyi}}, \bibinfo {author} {\bibfnamefont {Matthew}\ \bibnamefont
  {Hastings}}, \ and\ \bibinfo {author} {\bibfnamefont {Spyridon}\ \bibnamefont
  {Michalakis}},\ }\bibfield  {title} {\enquote {\bibinfo {title} {Topological
  quantum order: Stability under local perturbations},}\ }\href {\doibase
  10.1063/1.3490195} {\bibfield  {journal} {\bibinfo  {journal} {J. Math.
  Phys.}\ }\textbf {\bibinfo {volume} {51}},\ \bibinfo {pages} {093512}
  (\bibinfo {year} {2010})},\ \Eprint {http://arxiv.org/abs/1001.0344}
  {arXiv:1001.0344} \BibitemShut {NoStop}%
\bibitem [{\citenamefont {Schlingemann}\ \emph {et~al.}(2008)\citenamefont
  {Schlingemann}, \citenamefont {Vogts},\ and\ \citenamefont
  {Werner}}]{clifQCA}%
  \BibitemOpen
  \bibfield  {author} {\bibinfo {author} {\bibfnamefont {Dirk-M.}\ \bibnamefont
  {Schlingemann}}, \bibinfo {author} {\bibfnamefont {Holger}\ \bibnamefont
  {Vogts}}, \ and\ \bibinfo {author} {\bibfnamefont {Reinhard~F.}\ \bibnamefont
  {Werner}},\ }\bibfield  {title} {\enquote {\bibinfo {title} {On the structure
  of clifford quantum cellular automata},}\ }\href {\doibase 10.1063/1.3005565}
  {\bibfield  {journal} {\bibinfo  {journal} {Journal of Mathematical Physics}\
  }\textbf {\bibinfo {volume} {49}},\ \bibinfo {pages} {112104} (\bibinfo
  {year} {2008})},\ \Eprint {http://arxiv.org/abs/0804.4447} {arXiv:0804.4447}
  \BibitemShut {NoStop}%
\bibitem [{\citenamefont {Levin}\ and\ \citenamefont
  {Wen}(2003)}]{LevinWen2003Fermions}%
  \BibitemOpen
  \bibfield  {author} {\bibinfo {author} {\bibfnamefont {Michael}\ \bibnamefont
  {Levin}}\ and\ \bibinfo {author} {\bibfnamefont {Xiao-Gang}\ \bibnamefont
  {Wen}},\ }\bibfield  {title} {\enquote {\bibinfo {title} {Fermions, strings,
  and gauge fields in lattice spin models},}\ }\href {\doibase
  10.1103/PhysRevB.67.245316} {\bibfield  {journal} {\bibinfo  {journal} {Phys.
  Rev. B}\ }\textbf {\bibinfo {volume} {67}},\ \bibinfo {pages} {245316}
  (\bibinfo {year} {2003})},\ \Eprint {http://arxiv.org/abs/cond-mat/0302460}
  {arXiv:cond-mat/0302460} \BibitemShut {NoStop}%
\bibitem [{\citenamefont {Bomb\'in}(2014)}]{Bombin2011Structure}%
  \BibitemOpen
  \bibfield  {author} {\bibinfo {author} {\bibfnamefont {H\'ector}\
  \bibnamefont {Bomb\'in}},\ }\bibfield  {title} {\enquote {\bibinfo {title}
  {Structure of 2{D} topological stabilizer codes},}\ }\href {\doibase
  10.1007/s00220-014-1893-4} {\bibfield  {journal} {\bibinfo  {journal}
  {Commun. Math. Phys.}\ }\textbf {\bibinfo {volume} {327}},\ \bibinfo {pages}
  {387--432} (\bibinfo {year} {2014})},\ \Eprint
  {http://arxiv.org/abs/1107.2707} {arXiv:1107.2707} \BibitemShut {NoStop}%
\bibitem [{\citenamefont {Levin}(2013)}]{Levin2013}%
  \BibitemOpen
  \bibfield  {author} {\bibinfo {author} {\bibfnamefont {Michael}\ \bibnamefont
  {Levin}},\ }\bibfield  {title} {\enquote {\bibinfo {title} {Protected edge
  modes without symmetry},}\ }\href {\doibase 10.1103/PhysRevX.3.021009}
  {\bibfield  {journal} {\bibinfo  {journal} {Phys. Rev. X}\ }\textbf {\bibinfo
  {volume} {3}},\ \bibinfo {pages} {021009} (\bibinfo {year} {2013})},\ \Eprint
  {http://arxiv.org/abs/1301.7355} {arXiv:1301.7355} \BibitemShut {NoStop}%
\bibitem [{\citenamefont {Bravyi}\ and\ \citenamefont
  {Terhal}(2009)}]{BravyiTerhal2009no-go}%
  \BibitemOpen
  \bibfield  {author} {\bibinfo {author} {\bibfnamefont {Sergey}\ \bibnamefont
  {Bravyi}}\ and\ \bibinfo {author} {\bibfnamefont {Barbara}\ \bibnamefont
  {Terhal}},\ }\bibfield  {title} {\enquote {\bibinfo {title} {A no-go theorem
  for a two-dimensional self-correcting quantum memory based on stabilizer
  codes},}\ }\href {\doibase 10.1088/1367-2630/11/4/043029} {\bibfield
  {journal} {\bibinfo  {journal} {New J. Phys.}\ }\textbf {\bibinfo {volume}
  {11}},\ \bibinfo {pages} {043029} (\bibinfo {year} {2009})},\ \Eprint
  {http://arxiv.org/abs/0810.1983} {arXiv:0810.1983} \BibitemShut {NoStop}%
\bibitem [{\citenamefont {Chen}\ \emph {et~al.}(2018)\citenamefont {Chen},
  \citenamefont {Kapustin},\ and\ \citenamefont {Radicevic}}]{Kapustin2d}%
  \BibitemOpen
  \bibfield  {author} {\bibinfo {author} {\bibfnamefont {Yu-An}\ \bibnamefont
  {Chen}}, \bibinfo {author} {\bibfnamefont {Anton}\ \bibnamefont {Kapustin}},
  \ and\ \bibinfo {author} {\bibfnamefont {Djordje}\ \bibnamefont
  {Radicevic}},\ }\bibfield  {title} {\enquote {\bibinfo {title} {Exact
  bosonization in two spatial dimensions and a new class of lattice gauge
  theories},}\ }\href {\doibase 10.1016/j.aop.2018.03.024} {\bibfield
  {journal} {\bibinfo  {journal} {Annals of Physics}\ }\textbf {\bibinfo
  {volume} {393}},\ \bibinfo {pages} {234--253} (\bibinfo {year} {2018})},\
  \Eprint {http://arxiv.org/abs/1711.00515} {arXiv:1711.00515} \BibitemShut
  {NoStop}%
\bibitem [{\citenamefont {Calderbank}\ \emph {et~al.}(1997)\citenamefont
  {Calderbank}, \citenamefont {Rains}, \citenamefont {Shor},\ and\
  \citenamefont {Sloane}}]{CalderbankRainsShorEtAl1997Quantum}%
  \BibitemOpen
  \bibfield  {author} {\bibinfo {author} {\bibfnamefont {A.~R.}\ \bibnamefont
  {Calderbank}}, \bibinfo {author} {\bibfnamefont {E.~M}\ \bibnamefont
  {Rains}}, \bibinfo {author} {\bibfnamefont {P.~W.}\ \bibnamefont {Shor}}, \
  and\ \bibinfo {author} {\bibfnamefont {N.~J.~A.}\ \bibnamefont {Sloane}},\
  }\bibfield  {title} {\enquote {\bibinfo {title} {Quantum error correction and
  orthogonal geometry},}\ }\href {\doibase 10.1103/PhysRevLett.78.405}
  {\bibfield  {journal} {\bibinfo  {journal} {Phys. Rev. Lett.}\ }\textbf
  {\bibinfo {volume} {78}},\ \bibinfo {pages} {405--408} (\bibinfo {year}
  {1997})},\ \Eprint {http://arxiv.org/abs/quant-ph/9605005}
  {arXiv:quant-ph/9605005} \BibitemShut {NoStop}%
\bibitem [{\citenamefont {MacWilliams}\ and\ \citenamefont
  {Sloane}(1977)}]{MacWilliamsSloane1977}%
  \BibitemOpen
  \bibfield  {author} {\bibinfo {author} {\bibfnamefont {F.~J.}\ \bibnamefont
  {MacWilliams}}\ and\ \bibinfo {author} {\bibfnamefont {N.~J.~A.}\
  \bibnamefont {Sloane}},\ }\href@noop {} {\emph {\bibinfo {title} {The Theory
  of Error Correcting Codes}}}\ (\bibinfo  {publisher} {North-Holland,
  Amsterdam},\ \bibinfo {year} {1977})\BibitemShut {NoStop}%
\bibitem [{\citenamefont {Haah}(2016)}]{Haah2016}%
  \BibitemOpen
  \bibfield  {author} {\bibinfo {author} {\bibfnamefont {Jeongwan}\
  \bibnamefont {Haah}},\ }\bibfield  {title} {\enquote {\bibinfo {title}
  {Algebraic methods for quantum codes on lattices},}\ }\href@noop {}
  {\bibfield  {journal} {\bibinfo  {journal} {Revista Colombiana de
  Matem\'aticas}\ ,\ \bibinfo {pages} {295--345}} (\bibinfo {year} {2016})},\
  \Eprint {http://arxiv.org/abs/1607.01387} {arXiv:1607.01387} \BibitemShut
  {NoStop}%
\bibitem [{\citenamefont {Briegel}\ and\ \citenamefont
  {Raussendorf}(2001)}]{BriegelRaussendorf}%
  \BibitemOpen
  \bibfield  {author} {\bibinfo {author} {\bibfnamefont {Hans~J.}\ \bibnamefont
  {Briegel}}\ and\ \bibinfo {author} {\bibfnamefont {Robert}\ \bibnamefont
  {Raussendorf}},\ }\bibfield  {title} {\enquote {\bibinfo {title} {Persistent
  entanglement in arrays of interacting particles},}\ }\href {\doibase
  10.1103/PhysRevLett.86.910} {\bibfield  {journal} {\bibinfo  {journal} {Phys.
  Rev. Lett.}\ }\textbf {\bibinfo {volume} {86}},\ \bibinfo {pages} {910}
  (\bibinfo {year} {2001})},\ \Eprint {http://arxiv.org/abs/quant-ph/0004051}
  {arXiv:quant-ph/0004051} \BibitemShut {NoStop}%
\bibitem [{\citenamefont {Eisenbud}(2004)}]{Eisenbud}%
  \BibitemOpen
  \bibfield  {author} {\bibinfo {author} {\bibfnamefont {David}\ \bibnamefont
  {Eisenbud}},\ }\href@noop {} {\emph {\bibinfo {title} {Commutative Algebra
  with a View Toward Algebraic Geometry}}}\ (\bibinfo  {publisher} {Springer},\
  \bibinfo {year} {2004})\BibitemShut {NoStop}%
\bibitem [{\citenamefont {Buchsbaum}\ and\ \citenamefont
  {Eisenbud}(1973)}]{BuchsbaumEisenbud1973Exact}%
  \BibitemOpen
  \bibfield  {author} {\bibinfo {author} {\bibfnamefont {David~A}\ \bibnamefont
  {Buchsbaum}}\ and\ \bibinfo {author} {\bibfnamefont {David}\ \bibnamefont
  {Eisenbud}},\ }\bibfield  {title} {\enquote {\bibinfo {title} {What makes a
  complex exact?}}\ }\href {\doibase 10.1016/0021-8693(73)90044-6} {\bibfield
  {journal} {\bibinfo  {journal} {Journal of Algebra}\ }\textbf {\bibinfo
  {volume} {25}},\ \bibinfo {pages} {259--268} (\bibinfo {year}
  {1973})}\BibitemShut {NoStop}%
\bibitem [{\citenamefont {Suslin}(1977)}]{Suslin1977Stability}%
  \BibitemOpen
  \bibfield  {author} {\bibinfo {author} {\bibfnamefont {A.~A.}\ \bibnamefont
  {Suslin}},\ }\bibfield  {title} {\enquote {\bibinfo {title} {On the structure
  of the special linear group over polynomial rings},}\ }\href {\doibase
  10.1070/IM1977v011n02ABEH001709} {\bibfield  {journal} {\bibinfo  {journal}
  {Mathematics of the USSR-Izvestiya}\ }\textbf {\bibinfo {volume} {11}},\
  \bibinfo {pages} {221} (\bibinfo {year} {1977})}\BibitemShut {NoStop}%
\bibitem [{\citenamefont {Swan}(1978)}]{Swan1978}%
  \BibitemOpen
  \bibfield  {author} {\bibinfo {author} {\bibfnamefont {Richard~G.}\
  \bibnamefont {Swan}},\ }\bibfield  {title} {\enquote {\bibinfo {title}
  {Projective modules over {Laurent} polynomial rings},}\ }\href {\doibase
  10.2307/1997613} {\bibfield  {journal} {\bibinfo  {journal} {Transactions of
  the American Mathematical Society}\ }\textbf {\bibinfo {volume} {237}},\
  \bibinfo {pages} {111--120} (\bibinfo {year} {1978})}\BibitemShut {NoStop}%
\bibitem [{\citenamefont {Park}\ and\ \citenamefont
  {Woodburn}(1995)}]{ParkWoodburn1995Algorithmic}%
  \BibitemOpen
  \bibfield  {author} {\bibinfo {author} {\bibfnamefont {H.}~\bibnamefont
  {Park}}\ and\ \bibinfo {author} {\bibfnamefont {C.}~\bibnamefont
  {Woodburn}},\ }\bibfield  {title} {\enquote {\bibinfo {title} {An algorithmic
  proof of {S}uslin's stability theorem over polynomial rings},}\ }\href
  {\doibase 10.1006/jabr.1995.1349} {\bibfield  {journal} {\bibinfo  {journal}
  {Journal of Algebra}\ }\textbf {\bibinfo {volume} {178}},\ \bibinfo {pages}
  {277--298} (\bibinfo {year} {1995})},\ \Eprint
  {http://arxiv.org/abs/alg-geom/9405003} {arXiv:alg-geom/9405003} \BibitemShut
  {NoStop}%
\bibitem [{\citenamefont {Elman}\ \emph {et~al.}(2008)\citenamefont {Elman},
  \citenamefont {Karpenko},\ and\ \citenamefont {Merkurjev}}]{Kniga}%
  \BibitemOpen
  \bibfield  {author} {\bibinfo {author} {\bibfnamefont {Richard}\ \bibnamefont
  {Elman}}, \bibinfo {author} {\bibfnamefont {Nikita}\ \bibnamefont
  {Karpenko}}, \ and\ \bibinfo {author} {\bibfnamefont {Alexander}\
  \bibnamefont {Merkurjev}},\ }\href
  {https://sites.ualberta.ca/~karpenko/publ/Kniga.pdf} {\emph {\bibinfo {title}
  {The Algebraic and Geometric Theory of Quadratic Forms}}},\ \bibinfo {series}
  {Colloquium Publications}, Vol.~\bibinfo {volume} {56}\ (\bibinfo
  {publisher} {American Mathematical Society},\ \bibinfo {year}
  {2008})\BibitemShut {NoStop}%
\bibitem [{\citenamefont {Lam}(2004)}]{Lam}%
  \BibitemOpen
  \bibfield  {author} {\bibinfo {author} {\bibfnamefont {T.~Y.}\ \bibnamefont
  {Lam}},\ }\href@noop {} {\emph {\bibinfo {title} {Introduction to Quadratic
  Forms over Fields}}},\ Graduate studies in mathematics 67\ (\bibinfo
  {publisher} {American Mathematical Society},\ \bibinfo {year}
  {2004})\BibitemShut {NoStop}%
\bibitem [{\citenamefont {Haah}\ \emph {et~al.}(2017)\citenamefont {Haah},
  \citenamefont {Hastings}, \citenamefont {Poulin},\ and\ \citenamefont
  {Wecker}}]{HHPW2017}%
  \BibitemOpen
  \bibfield  {author} {\bibinfo {author} {\bibfnamefont {Jeongwan}\
  \bibnamefont {Haah}}, \bibinfo {author} {\bibfnamefont {Matthew~B.}\
  \bibnamefont {Hastings}}, \bibinfo {author} {\bibfnamefont {D.}~\bibnamefont
  {Poulin}}, \ and\ \bibinfo {author} {\bibfnamefont {D.}~\bibnamefont
  {Wecker}},\ }\bibfield  {title} {\enquote {\bibinfo {title} {Magic state
  distillation with low space overhead and optimal asymptotic input count},}\
  }\href {\doibase 10.22331/q-2017-10-03-31} {\bibfield  {journal} {\bibinfo
  {journal} {Quantum}\ }\textbf {\bibinfo {volume} {1}},\ \bibinfo {pages} {31}
  (\bibinfo {year} {2017})},\ \Eprint {http://arxiv.org/abs/1703.07847}
  {arXiv:1703.07847} \BibitemShut {NoStop}%
\bibitem [{\citenamefont {Milnor}\ and\ \citenamefont
  {Husemoller}(1973)}]{MilnorHusemoller}%
  \BibitemOpen
  \bibfield  {author} {\bibinfo {author} {\bibfnamefont {John}\ \bibnamefont
  {Milnor}}\ and\ \bibinfo {author} {\bibfnamefont {Dale}\ \bibnamefont
  {Husemoller}},\ }\href {\doibase 10.1007/978-3-642-88330-9} {\emph {\bibinfo
  {title} {Symmetric Bilinear Forms}}}\ (\bibinfo  {publisher} {Springer,
  Berlin, Heidelberg},\ \bibinfo {year} {1973})\BibitemShut {NoStop}%
\bibitem [{\citenamefont {Lang}(2002)}]{Lang}%
  \BibitemOpen
  \bibfield  {author} {\bibinfo {author} {\bibfnamefont {Serge}\ \bibnamefont
  {Lang}},\ }\href@noop {} {\emph {\bibinfo {title} {Algebra}}},\ \bibinfo
  {edition} {revised 3rd}\ ed.\ (\bibinfo  {publisher} {Springer},\ \bibinfo
  {year} {2002})\BibitemShut {NoStop}%
\bibitem [{\citenamefont {Raussendorf}\ \emph {et~al.}(2005)\citenamefont
  {Raussendorf}, \citenamefont {Bravyi},\ and\ \citenamefont
  {Harrington}}]{RBH}%
  \BibitemOpen
  \bibfield  {author} {\bibinfo {author} {\bibfnamefont {Robert}\ \bibnamefont
  {Raussendorf}}, \bibinfo {author} {\bibfnamefont {Sergey}\ \bibnamefont
  {Bravyi}}, \ and\ \bibinfo {author} {\bibfnamefont {Jim}\ \bibnamefont
  {Harrington}},\ }\bibfield  {title} {\enquote {\bibinfo {title} {Long-range
  quantum entanglement in noisy cluster states},}\ }\href {\doibase
  10.1103/PhysRevA.71.062313} {\bibfield  {journal} {\bibinfo  {journal} {Phys.
  Rev. A}\ }\textbf {\bibinfo {volume} {71}},\ \bibinfo {pages} {062313}
  (\bibinfo {year} {2005})},\ \Eprint {http://arxiv.org/abs/quant-ph/0407255v2}
  {arXiv:quant-ph/0407255v2} \BibitemShut {NoStop}%
\bibitem [{\citenamefont {Roberts}\ \emph {et~al.}(2017)\citenamefont
  {Roberts}, \citenamefont {Yoshida}, \citenamefont {Kubica},\ and\
  \citenamefont {Bartlett}}]{Roberts2016}%
  \BibitemOpen
  \bibfield  {author} {\bibinfo {author} {\bibfnamefont {Sam}\ \bibnamefont
  {Roberts}}, \bibinfo {author} {\bibfnamefont {Beni}\ \bibnamefont {Yoshida}},
  \bibinfo {author} {\bibfnamefont {Aleksander}\ \bibnamefont {Kubica}}, \ and\
  \bibinfo {author} {\bibfnamefont {Stephen~D.}\ \bibnamefont {Bartlett}},\
  }\bibfield  {title} {\enquote {\bibinfo {title} {Symmetry protected
  topological order at nonzero temperature},}\ }\href {\doibase
  10.1103/PhysRevA.96.022306} {\bibfield  {journal} {\bibinfo  {journal} {Phys.
  Rev. A}\ }\textbf {\bibinfo {volume} {96}},\ \bibinfo {pages} {022306}
  (\bibinfo {year} {2017})},\ \Eprint {http://arxiv.org/abs/1611.05450v2}
  {arXiv:1611.05450v2} \BibitemShut {NoStop}%
\bibitem [{\citenamefont {{Kapustin}}(2014)}]{Kapustin}%
  \BibitemOpen
  \bibfield  {author} {\bibinfo {author} {\bibfnamefont {A.}~\bibnamefont
  {{Kapustin}}},\ }\bibfield  {title} {\enquote {\bibinfo {title} {{Symmetry
  Protected Topological Phases, Anomalies, and Cobordisms: Beyond Group
  Cohomology}},}\ }\href@noop {} {\  (\bibinfo {year} {2014})},\ \Eprint
  {http://arxiv.org/abs/1403.1467} {arXiv:1403.1467} \BibitemShut {NoStop}%
\bibitem [{\citenamefont {{Freed}}(2014)}]{Freed}%
  \BibitemOpen
  \bibfield  {author} {\bibinfo {author} {\bibfnamefont {D.~S.}\ \bibnamefont
  {{Freed}}},\ }\bibfield  {title} {\enquote {\bibinfo {title} {{Short-range
  entanglement and invertible field theories}},}\ }\href@noop {} {\  (\bibinfo
  {year} {2014})},\ \Eprint {http://arxiv.org/abs/1406.7278} {arXiv:1406.7278}
  \BibitemShut {NoStop}%
\bibitem [{\citenamefont {{Ellison}}\ and\ \citenamefont
  {{Fidkowski}}(2018)}]{Ellison}%
  \BibitemOpen
  \bibfield  {author} {\bibinfo {author} {\bibfnamefont {T.~D.}\ \bibnamefont
  {{Ellison}}}\ and\ \bibinfo {author} {\bibfnamefont {L.}~\bibnamefont
  {{Fidkowski}}},\ }\bibfield  {title} {\enquote {\bibinfo {title}
  {{Disentangling interacting symmetry protected phases of fermions in two
  dimensions}},}\ }\href@noop {} {\  (\bibinfo {year} {2018})},\ \Eprint
  {http://arxiv.org/abs/1806.09623} {arXiv:1806.09623} \BibitemShut {NoStop}%
\bibitem [{\citenamefont {{Chen}}\ and\ \citenamefont
  {{Kapustin}}(2018)}]{Kapustin3d}%
  \BibitemOpen
  \bibfield  {author} {\bibinfo {author} {\bibfnamefont {Y.-A.}\ \bibnamefont
  {{Chen}}}\ and\ \bibinfo {author} {\bibfnamefont {A.}~\bibnamefont
  {{Kapustin}}},\ }\bibfield  {title} {\enquote {\bibinfo {title}
  {{Bosonization in three spatial dimensions and a 2-form gauge theory}},}\
  }\href@noop {} {\  (\bibinfo {year} {2018})},\ \Eprint
  {http://arxiv.org/abs/1807.07081} {arXiv:1807.07081} \BibitemShut {NoStop}%
\bibitem [{\citenamefont {Kitaev}(2001)}]{kitaev2001unpaired}%
  \BibitemOpen
  \bibfield  {author} {\bibinfo {author} {\bibfnamefont {A.~Yu}\ \bibnamefont
  {Kitaev}},\ }\bibfield  {title} {\enquote {\bibinfo {title} {Unpaired
  majorana fermions in quantum wires},}\ }\href {\doibase
  10.1070/1063-7869/44/10S/S29} {\bibfield  {journal} {\bibinfo  {journal}
  {Physics-Uspekhi}\ }\textbf {\bibinfo {volume} {44}},\ \bibinfo {pages} {131}
  (\bibinfo {year} {2001})},\ \Eprint {http://arxiv.org/abs/cond-mat/0010440}
  {arXiv:cond-mat/0010440} \BibitemShut {NoStop}%
\bibitem [{\citenamefont {Knill}\ \emph {et~al.}(2000)\citenamefont {Knill},
  \citenamefont {Laflamme},\ and\ \citenamefont {Viola}}]{knill2000theory}%
  \BibitemOpen
  \bibfield  {author} {\bibinfo {author} {\bibfnamefont {Emanuel}\ \bibnamefont
  {Knill}}, \bibinfo {author} {\bibfnamefont {Raymond}\ \bibnamefont
  {Laflamme}}, \ and\ \bibinfo {author} {\bibfnamefont {Lorenza}\ \bibnamefont
  {Viola}},\ }\bibfield  {title} {\enquote {\bibinfo {title} {Theory of quantum
  error correction for general noise},}\ }\href {\doibase
  10.1103/PhysRevLett.84.2525} {\bibfield  {journal} {\bibinfo  {journal}
  {Physical Review Letters}\ }\textbf {\bibinfo {volume} {84}},\ \bibinfo
  {pages} {2525} (\bibinfo {year} {2000})},\ \Eprint
  {http://arxiv.org/abs/quant-ph/9604034} {arXiv:quant-ph/9604034} \BibitemShut
  {NoStop}%
\bibitem [{\citenamefont {Bloch}\ \emph {et~al.}(1992)\citenamefont {Bloch},
  \citenamefont {Brockett},\ and\ \citenamefont {Ratiu}}]{bloch1992completely}%
  \BibitemOpen
  \bibfield  {author} {\bibinfo {author} {\bibfnamefont {Anthony~M}\
  \bibnamefont {Bloch}}, \bibinfo {author} {\bibfnamefont {Roger~W}\
  \bibnamefont {Brockett}}, \ and\ \bibinfo {author} {\bibfnamefont {Tudor~S}\
  \bibnamefont {Ratiu}},\ }\bibfield  {title} {\enquote {\bibinfo {title}
  {Completely integrable gradient flows},}\ }\href {\doibase
  10.1007/BF02099528} {\bibfield  {journal} {\bibinfo  {journal}
  {Communications in Mathematical Physics}\ }\textbf {\bibinfo {volume}
  {147}},\ \bibinfo {pages} {57--74} (\bibinfo {year} {1992})}\BibitemShut
  {NoStop}%
\bibitem [{\citenamefont {Brockett}(1991)}]{brockett1991dynamical}%
  \BibitemOpen
  \bibfield  {author} {\bibinfo {author} {\bibfnamefont {Roger~W}\ \bibnamefont
  {Brockett}},\ }\bibfield  {title} {\enquote {\bibinfo {title} {Dynamical
  systems that sort lists, diagonalize matrices, and solve linear programming
  problems},}\ }\href {\doibase 10.1016/0024-3795(91)90021-N} {\bibfield
  {journal} {\bibinfo  {journal} {Linear Algebra and its applications}\
  }\textbf {\bibinfo {volume} {146}},\ \bibinfo {pages} {79--91} (\bibinfo
  {year} {1991})}\BibitemShut {NoStop}%
\bibitem [{\citenamefont {Wall}(1964)}]{Wall1964GradedBrauer}%
  \BibitemOpen
  \bibfield  {author} {\bibinfo {author} {\bibfnamefont {C.~T.~C.}\
  \bibnamefont {Wall}},\ }\bibfield  {title} {\enquote {\bibinfo {title}
  {Graded brauer groups},}\ }\href {\doibase 10.1515/crll.1964.213.187}
  {\bibfield  {journal} {\bibinfo  {journal} {J. Reine Angew. Math.}\ }\textbf
  {\bibinfo {volume} {213}},\ \bibinfo {pages} {187} (\bibinfo {year}
  {1964})}\BibitemShut {NoStop}%
\end{thebibliography}%
\end{document}